\PassOptionsToPackage{dvipsnames}{xcolor}
\documentclass[a4paper,twocolumn,superscriptaddress,preprintnumbers,nobalancelastpage,longbibliography,accepted=2022-08-31]{quantumarticle}

 \pdfoutput=1
\usepackage{graphicx, color, graphpap}
\usepackage{enumitem}
\usepackage{amssymb}
\usepackage{amsthm}
\usepackage{multirow}
\usepackage[colorlinks=true,citecolor=blue,linkcolor=magenta]{hyperref}
\usepackage[T1]{fontenc}
\usepackage{bbm}
\usepackage{thmtools,thm-restate}
\usepackage{verbatim}
\usepackage{mathtools}
\usepackage{titlesec}
\usepackage{amsmath}
\usepackage[title]{appendix}
\usepackage[linesnumbered,ruled,vlined]{algorithm2e}
\SetKwInput{kwInit}{Init}

\usepackage[numbers]{natbib}

\usepackage[caption = false]{subfig}

\long\def\ca#1\cb{} 

\newcommand{\braket}[2]{\langle #1 \hspace{1pt} | \hspace{1pt} #2 \rangle}

\newcommand{\ketbra}[2]{| \hspace{1pt} #1 \rangle \langle #2 \hspace{1pt} |}

\newcommand{\ketbraq}[1]{\ketbra{#1}{#1}}
\newcommand{\bramatket}[3]{\langle #1 \hspace{1pt} | #2 | \hspace{1pt} #3 \rangle}

\newcommand{\ket}[1]{|#1\rangle}               
\newcommand{\bra}[1]{\langle #1|}              

\newcommand{\poly}{\operatorname{poly}}

\newcommand{\Zbb}{\mathbb{Z}}
\newcommand{\Cbb}{\mathbb{C}}
\newcommand{\Ubb}{\mathbb{U}}
\newcommand{\Rbb}{\mathbb{R}}

\newcommand{\AC}{\mathcal{A}}

\newcommand{\GC}{\mathcal{G}}
\newcommand{\HC}{\mathcal{H}}

\newcommand{\OC}{\mathcal{O}}

\newcommand{\SC}{\mathcal{S}}

\newcommand{\UC}{\mathcal{U}}

\newcommand{\Tr}{{\rm Tr}}

\newcommand{\Var}{{\rm Var}}

\renewcommand{\geq}{\geqslant}
\renewcommand{\leq}{\leqslant}

\renewcommand{\vec}[1]{\boldsymbol{#1}}  

\newcommand{\ad}{^\dagger}
\newcommand{\ts}{^{\otimes 2}}

\newcommand*{\id}{\openone}

\renewcommand{\th}{\theta } 

\newcommand{\sg}{\sigma }

\newcommand{\thv}{\vec{\theta}}

{}
{}
\newtheorem{theorem}{Theorem}

\newtheorem{corollary}{Corollary}

\newtheorem{conjecture}{Conjecture}
\newtheorem{proposition}{Proposition}

\newtheorem{definition}{Definition}

\def\k{^{(k)}}

\def\liea{\mathfrak{k}}
\def\liea{\mathfrak{g}}

\def\lieg{\mathds{G}}

\def\H{\mathcal{H}}

\def\a{\alpha}
\def\b{\beta}
\def\P{\mathbb{P}}

\newcommand\mf[1]{\mathfrak{#1}}

\newcommand\spn{\text{span}}
\newcommand\kp{\ket{\psi}}

\newcommand\Z{\text{Z}}

\newcommand\TFIM{\text{TFIM}}
\newcommand\LTFIM{\text{LTFIM}}
\newcommand\tH{\tilde{O}}

\usepackage{amssymb}
\usepackage{dsfont}

\begin{document}

\title{Diagnosing barren plateaus with tools from quantum optimal control}

\author{Mart\'in Larocca}
\affiliation{Departamento de Física “J. J. Giambiagi” and IFIBA, FCEyN, Universidad de Buenos Aires, 1428 Buenos Aires, Argentina}
\affiliation{Theoretical Division, Los Alamos National Laboratory, Los Alamos, New Mexico 87545, USA}

\author{Piotr Czarnik} 
\affiliation{Theoretical Division, Los Alamos National Laboratory, Los Alamos, New Mexico 87545, USA}

\author{Kunal Sharma} 
\address{Hearne Institute for Theoretical Physics and Department of Physics and Astronomy, Louisiana State University, Baton Rouge, LA USA}
\affiliation{Theoretical Division, Los Alamos National Laboratory, Los Alamos, New Mexico 87545, USA}

\author{Gopikrishnan Muraleedharan}
\affiliation{Theoretical Division, Los Alamos National Laboratory, Los Alamos, New Mexico 87545, USA}

\author{Patrick J. Coles}
\affiliation{Theoretical Division, Los Alamos National Laboratory, Los Alamos, New Mexico 87545, USA}

\author{M. Cerezo}
\affiliation{Information Sciences, Los Alamos National Laboratory, Los Alamos, NM 87545, USA}
\affiliation{Center for Nonlinear Studies, Los Alamos National Laboratory, Los Alamos, New Mexico 87545, USA}

\begin{abstract}
Variational Quantum Algorithms (VQAs) have received considerable attention due to their potential for achieving near-term quantum advantage. However, more work is needed to understand their scalability. One known scaling result for VQAs is barren plateaus, where certain circumstances lead to exponentially vanishing gradients. It is common folklore that problem-inspired ansatzes avoid barren plateaus, but in fact, very little is known about their gradient scaling. In this work we employ tools from quantum optimal control to develop a framework that can diagnose the presence or absence of barren plateaus for problem-inspired ansatzes. Such ansatzes include the Quantum Alternating Operator Ansatz (QAOA), the Hamiltonian Variational Ansatz (HVA), and others. With our framework, we prove that avoiding barren plateaus for these ansatzes is not always guaranteed. Specifically, we show that the gradient scaling of the VQA depends on the degree of controllability of the system, and hence can be diagnosed through the dynamical Lie algebra $\mathfrak{g}$ obtained from the generators of the ansatz. We analyze the existence of barren plateaus in QAOA and HVA ansatzes, and we highlight the role of the input state, as different initial states can lead to the presence or absence of barren plateaus. Taken together, our results  provide a framework for trainability-aware ansatz design strategies that do not come at the cost of extra quantum resources. Moreover, we prove no-go results for obtaining ground states with variational ansatzes for controllable system such as spin glasses. Our work establishes a link between the existence of barren plateaus and the scaling of the dimension of $\mathfrak{g}$.  
\end{abstract}

\maketitle

\section{INTRODUCTION}

Quantum computers hold the promise to achieve computational speed-ups over classical supercomputers for certain tasks~\cite{shor1994algorithms,harrow2009quantum,berry2015simulating,georgescu2014quantum}. However, despite recent tremendous progress in quantum technologies, present-day quantum devices (known as Noisy Intermediate-Scale Quantum (NISQ) devices) are constrained by the limited number of qubits, connectivity, and by the presence of quantum noise~\cite{preskill2018quantum}. Hence, it becomes crucial to determine what are the capabilities and limitations of NISQ computers to achieving a quantum advantage. 

One of the most promising computational models for making use of  near-term quantum computers are  Variational Quantum Algorithms (VQAs) \cite{cerezo2020variationalreview}. Here, a task of interest is encoded into a parametrized cost function $C(\thv)$ that is efficiently computable on a noisy quantum computer. Part of the computational complexity is pushed onto classical computers by leveraging the power of classical optimizer that  train the parameters $\thv$ and minimize the cost.  VQAs have been proposed for tasks such as solving linear systems of equations~\cite{bravo2020variational,huang2019near,xu2019variational} or performing dynamical quantum simulations~\cite{mcardle2019variational,grimsley2019adaptive,cirstoiu2020variational,commeau2020variational,gibbs2021long,yao2020adaptive,endo2020variational,lau2021quantum}, as well as for many others relevant applications~\cite{peruzzo2014variational,farhi2014quantum,mcclean2016theory,khatri2019quantum,romero2017quantum,larose2019variational,arrasmith2019variational,cerezo2020variationalfidelity,li2017efficient,heya2019subspace,bharti2020quantum,cerezo2020variational,beckey2020variational}.

Despite the wide application range  of VQAs, their widespread use is still limited by several  challenges that can hinder their success.  For instance, it has been shown that the optimization task associated with minimizing the VQA cost function is in general an NP-hard non-convex optimization problem~\cite{bittel2021training}. Moreover, despite the typical difficulties encountered in classical non-convex optimization tasks, there are new challenges that arise when training the parameters of VQAs such as hardware noise, or the limited precision arising from a limited number of shots. These difficulties have then led to several quantum-aware optimizers being developed~\cite{mitarai2018quantum,schuld2019evaluating,kubler2020adaptive,mcardle2019variational,stokes2020quantum,arrasmith2020operator}.

In addition, certain VQAs have been shown to exhibit the so-called barren plateau phenomenon, where the cost function becomes untrainable due to gradients that vanish, on average, exponentially with the system size~\cite{mcclean2018barren,cerezo2021cost,wang2020noise,cerezo2020impact,sharma2020trainability,arrasmith2020effect,holmes2020barren,marrero2020entanglement,patti2020entanglement,pesah2020absence,holmes2021connecting,arrasmith2021equivalence}. Thus, barren plateaus have then been recognized as one of the main limitations to overcome in order to preserve the hope of achieving quantum advantage with VQAs. Recently, a great deal of effort has been put forward to developing methods that can mitigate the effect of barren plateaus~\cite{cerezo2020cost,uvarov2020barren,volkoff2021large,verdon2019learning,grant2019initialization,skolik2020layerwise,bilkis2021semi}, but ideally one would like to devise and employ VQA ansatzes which do not exhibit barren plateaus altogether.

For instance, it is known that one should avoid problem-agnostic ansatzes  such as deep hardware efficient ansatzes, as these can exhibit barren plateaus due to their high expressibility~\cite{mcclean2018barren,cerezo2021cost,holmes2021connecting}. Hence, so-called problem-inspired ansatzes have been speculated to be able to overcome barren plateaus by encoding information about the problem at hand in the ansatz. Here, the intuition is that problem-inspired ansatzes constrain the space explored during the optimization to a space that either contains the solution to the problem, or that at least contains a good approximation to the solution, while maintaining a low expressibility.  

In this work we employ tools from Quantum Optimal Control (QOC) to diagnose  the presence or absence of barren plateaus in certain families of  problem-inspired ansatzes with a periodic structure. QOC theory is a long standing theoretical framework developed to  provide tools for the manipulation of quantum dynamical processes.  As shown in Fig.~\ref{fig:VQA-QOC}, we here make use of the fact that periodic VQAs and QOC systems can be considered as different level formulations of a common variational problem~\cite{magann2021pulses}, as both aim at driving a quantum system with a classical optimization loop. Most importantly, this connection allows us to understand and forecast the presence or absence of barren plateaus in problem-inspired variational ansatzes like the Quantum Alternating Operator Ansatz (QAOA)~\cite{farhi2014quantum,hadfield2019quantum} and the Hamiltonian Variational Ansatz (HVA)~\cite{wecker2015progress,wiersema2020exploring}. We note that the procedure is perfectly suitable for other periodic ansatzes like the adaptive QAOA~\cite{hadfield2019quantum,zhu2020adaptive} and the quantum optimal control ansatz~\cite{choquette2020quantum}. Moreover, our results also extend to quantum neural network architectures used in the quantum machine learning literature~\cite{thanasilp2021subtleties}.   Our results indicate that problem-inspired ansatzes are not immune to barren plateaus, and hence that certain ansatz strategies in the literature need to be revised.

Our main results are organized into propositions and theorems that show  that one can diagnose the existence of barren plateaus by analyzing the controllability of the system, i.e., by studying the Dynamical Lie Algebra (DLA) of the system. The DLA is the subspace of operator space spanned by the nested commutators of the elements in the set of generators of the ansatz (e.g., see~\cite{dalessandro2010introduction} for an introduction to quantum control theory). In an effort to give a comprehensive picture, our results follow the different controllability scenarios shown in Fig.~\ref{fig:Algebra_classification}.

The manuscript is organized as follows. In Section~\ref{sec:VQA} we present the theoretical framework for VQAs, which includes a description of the type of ansatz considered, as well as a  basic review of concepts related to barren plateaus and ansatz expressibility. Then, in Section~\ref{sec:QOC} we introduce the framework of QOC, and recall how in QOC theory the DLA of the ansatz generators is used to study the controllability of the system. Section~\ref{sec:main-results} contains the main results of this work, while in Section~\ref{sec:results} we present our numerical simulations. Finally, in Section~\ref{sec:discussions} we present our discussions and conclusions.

\section{VARIATIONAL QUANTUM ALGORITHMS}\label{sec:VQA}

\begin{figure}[t]
\centering
\includegraphics[width=1\columnwidth]{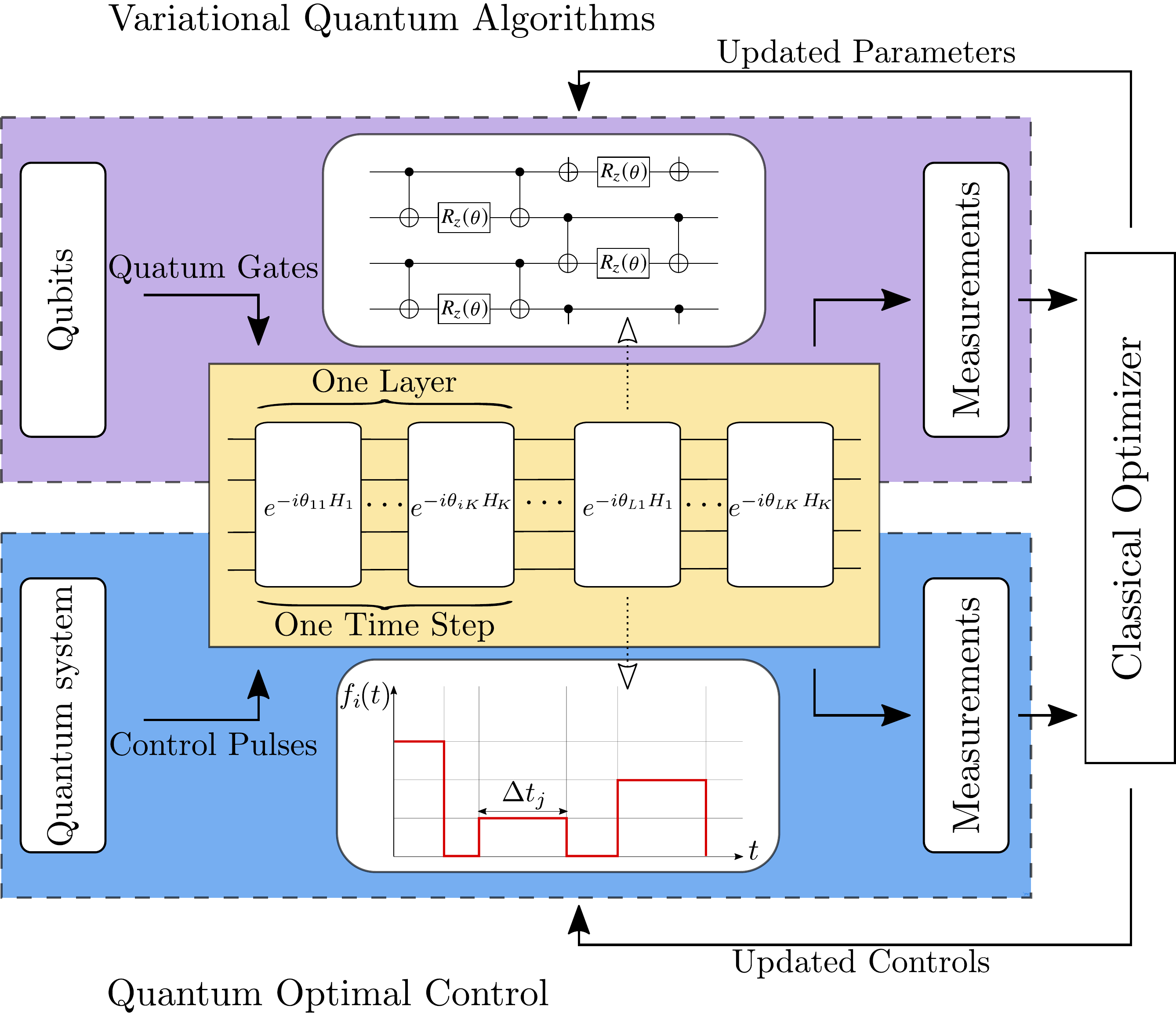}
\caption{\textbf{Framework for Variational Quantum Algorithms (VQA) and Quantum Optimal Control (QOC).} VQAs and QOCs can be regarded as two different levels of a theory that manipulate the evolution of a quantum system by training sets of parameters governing the system's dynamical evolution~\cite{magann2021pulses}. In VQAs (QOC) one applies a series of parametrized quantum gates (control pulses) to an input state. By gathering knowledge on the evolution via measurements on the resulting evolved state, the set of parameters (controls) are trained using a classical optimizer until a given task is completed. In this work we consider VQAs and QOC systems that have periodic structure ansatzes as in Eq.~\eqref{eq:ansatz}. }
\label{fig:VQA-QOC}
\end{figure}

In this section we review the basic framework of Variational Quantum Algorithms (VQAs). In particular, we discuss a general form for ansatzes that have a periodic structure, which we consider throughout this work. Since our goal is to analyze the gradient scaling, we additionally provide an overview of the barren plateau phenomenon.

\medskip 
\subsection{General framework}

We consider an optimization task where the goal is to minimize a cost function of the form
\begin{equation}\label{eq:cost}
    C(\thv)=\Tr[OU(\thv)\rho U\ad(\thv)]\,.
\end{equation}
Here,  $\rho$ is an input state on $n$ qubits in a $d$-dimensional Hilbert space with $d=2^n$,  $U(\thv)$ a  parametrized quantum circuit, and  $O$ is a Hermitian operator that defines the task at hand.

Throughout this work we consider layered parametrized quantum circuits  that, as shown in Fig.~\ref{fig:VQA-QOC},  have a periodic structure of the form 
\begin{equation}
    U(\thv) = \prod_{l=1}^L U_l(\thv_l),\quad U_l(\thv_l)= \prod_{k=0}^K e^{-i H_k \theta_{lk}}\,.
\label{eq:ansatz}
\end{equation}
Here, the index $l$ indicates the layer, $\thv_l=(\th_{l1},\ldots,\th_{lK})$ contains the parameters of such layer (such that $\thv=\{\thv_l\}_{l=1}^L$), and $H_k$ are Hermitian traceless operators that generate the unitaries in the ansatz. For generality, one can also allow for certain layers to be unparametrized, in which case one would simply set certain $\theta_{lk}$ to be constant. In what follows we refer to this type of ansatz as a Periodic Structure Ansatz (PSA). We refer the reader to Appendix~\ref{sec:app:ansatzes} for a detailed discussion of several ansatzes from the literature that are PSAs.

\subsection{Barren plateaus}

Recently, it has been shown that the choice of ansatz can hinder the trainability of the parameters $\thv$ for large problem sizes due to the existence of the so-called barren plateau phenomenon. In this context, deep unstructured problem-agnostic ansatz are known to exhibit barren plateaus~\cite{mcclean2018barren,cerezo2020cost,holmes2021connecting}. Hence, the design of ansatzes that overcome barren plateaus has been recognized as one of the most important challenges  to guarantee the success of VQAs~\cite{cerezo2020variationalreview}, and problem-inspired ansatzes have been proposed as one of the most promising strategies. However, despite their  promise,  little is known about the existence of barren plateaus in problem-inspired ansatzes.

Let us now briefly recall that when the cost exhibits a barren plateau, the gradients are exponentially suppressed (on average) across the optimization landscape. This implies that an exponentially large precision is needed to navigate trough the flat landscape and determine a cost minimizing direction~\cite{mcclean2018barren,cerezo2020cost,arrasmith2020effect}. Hence, consider the following definition.
\begin{definition}[Barren Plateau]\label{def:BP}
A cost function $C(\thv)$ as in Eq.~\eqref{eq:cost} is said to have a barren plateau when training $\theta_\mu\equiv\theta_{pq}\in\thv$,  if the cost function partial derivative $\partial C(\thv)/\partial \theta_{\mu}\equiv \partial_{\mu} C(\thv)$ is such that
\begin{equation}\label{eq:BP}
    \Var_{\thv}[\partial_{\mu} C(\thv)]\leq F(n)\,, \quad \text{with} \quad F(n)\in \OC\left(\frac{1}{b^n}\right)\,,
\end{equation}
for some $b>1$. Here the variance is taken with respect to the set of parameters $\thv$.
\end{definition}
We refer the reader to Appendix~\ref{sec:app:barren plateaus} for additional details on barren plateaus.

It is worth remarking that the barren plateau phenomenon  has been linked to the expressibility of the ansatz, as it has been shown that circuits with large expressibility will exhibit small gradients~\cite{holmes2021connecting}. In this context, one can quantify the expressibility of an ansatz by comparing the distribution of unitaries obtained from $U(\vec{\theta})$ to the maximally expressive uniform (Haar) distribution $U_{ H}$~\cite{sim2019expressibility}. Defining the $t$-th moment superoperator of the distribution generated by the ansatz $U(\thv)$,
\begin{equation}
    M^{(t)}_{U(\thv)}=\int_{\thv} d U(\thv) \, U(\thv)^{\otimes t} \otimes (U^*(\thv))^{\otimes t}\,,
\end{equation}
we recall that its ordinary action on a given operator can be obtained by placing said operator into the center of the representation of $M^{(t)}_{U(\thv)}$ as~\cite{caves1999quantum,rungta2001qudit}
\begin{equation}
    M^{(t)}_{U(\thv)}(\cdot)=\int_{\thv} d U(\thv) \, U(\thv)^{\otimes t}  (\cdot) (U(\thv)^\dagger)^{\otimes t}\,.
\end{equation}
In our case, we will only be interested in second moments. For that reason, we will focus on the deviation of the second moments of the distribution generated by the ansatz $M^{(2)}_{U(\thv)}$ from $M^{(2)}_{U_H}$ the second moments of the Haar distribution, via the norm of the superoperator 
\begin{align}~\label{eq:superop}
    \AC^{(t)}_{U(\thv)}=  M^{(t)}_{U_H}- M^{(t)}_{U(\thv)} \,.
\end{align}
For our purposes here, we find it convenient to define the expressibility as the  infinity norm, $\Vert A\Vert_{\infty}=\lambda_{\max}(A)$, with  $\lambda_{\max}(A)$ the largest singular value of $A$~\footnote{The expressibility can also be defined in terms of other matrix norms such as the diamond norm or the Schatten $p$-norms. However, due to the matrix norm equivalence, if $\Vert \AC^{(t)}_{U(\thv)}\Vert=\epsilon$ for our definition, there always exists an $\epsilon'$ for other expressibility definitions  such that the expressibility is equal to $\epsilon'$ and such that $\epsilon$ and $\epsilon'$ are related via a dimensionallity factor~\cite{hunter2019unitary,nakata2014generating}.}. Thus, the more expressible the ansatz, the smaller the norm $\Vert\AC^{(2)}_{U(\thv)}\Vert_{\infty}$, and the smaller the gradients of the cost partial derivatives~\cite{holmes2021connecting}. The limit  $\Vert\AC^{(2)}_{U(\thv)}\Vert_{\infty}=0$ is reached when $U(\thv)$ forms a $2$-design, in which case the cost exhibits a barren plateau according to Definition~\ref{def:BP}~\cite{mcclean2018barren,cerezo2020cost}.

\section{QUANTUM OPTIMAL CONTROL}\label{sec:QOC}

Quantum Optimal Control (QOC) is a theoretical framework that provides tools for the systematic manipulation of quantum dynamical systems. The connection between VQAs and QOC has been previously established showing that one can use QOC tools to specify the parameters $\thv$ at a device-level~\cite{yang2017optimizing,magann2021pulses,meitei2020gate} and to analyze VQA landscapes~\cite{lee2021towards}. Conversely, tools from VQAs have been employed to determine optimal control sequences ~\cite{li2017hybrid}. In particular, Ref.~\cite{magann2021pulses} notes that VQAs and QOC can be unified as  formulations of variational optimization at the circuit level and pulse level, respectively. In addition, the framework of QOC has been employed to analyze the computational universality of quantum circuits~\cite{ramakrishna1996relation,lloyd2018quantum,morales2020universality}, as well as their reachability~\cite{akshay2020reachability}.

In QOC one is interested in controlling the dynamical evolution of a quantum state $\kp$ in a complex $d$-dimensional Hilbert space $\H=\Cbb^d$ (where $d=2^n$)~\cite{dalessandro2010introduction}. In the typical setting, one has a Hamiltonian
\begin{equation}
    H(\{f_k(t)\}) = H_0+\sum_{k=1}^K f_k(t)H_k
\label{eq:QOC_Hamiltonian}
\end{equation}
that is tunable through some time-dependent functions $\{f_k(t)\}$, know as control fields or protocols. The fixed Hamiltonian $H_0$, usually called the drift, represents the natural or free evolution of the system, whereas the control Hamiltonians  $\{ H_i \}$ are associated with interactions with external degrees of freedom (usually electromagnetic radiation). Thus, $\kp$ evolves through the parametrized propagator $U(t)$ as $\kp(t)=\ket{\psi(t)}$. In turn, $U(t)$ is the solution to the Schr\"{o}dinger equation
\begin{equation}
    \frac{d U(t)}{dt}= -i H(\{f_k(t)\}) U(t),\quad \text{with} \quad U(0)=\id\,.
\label{eq:schrodinger}
\end{equation}
As shown in Appendix~\ref{app:qoc}, under standard assumptions, the Trotrerized QOC propagator of Eq.~\eqref{eq:schrodinger} is a PSA as in Eq.~\eqref{eq:ansatz}.

The variety of different dynamics a quantum control system in the form of Eq.~\eqref{eq:QOC_Hamiltonian} can undergo, upon variation of the control fields, is well understood though group theory. Since the Hamiltonian is Hermitian and traceless, $U(\thv)$ belongs to $\SC\UC(d)$, the Lie group of unitary $d\times d$ complex matrices that preserves the standard inner product on $\HC$. Surprisingly, the set of all unitaries $U(\thv)$ that can be accessed by such a control system forms itself a Lie group, known as the \textit{dynamical Lie group} $\lieg\subseteq \UC(d)$. Hence, a natural question which arises is: how can this group be determined?

First, let us define the set of generators.
\begin{definition}[Set of generators]\label{def:generators} Given a parametrized quantum circuit of the form in Eq.~\eqref{eq:ansatz} we define the set of generators $\GC=\{H_k\}_{k=0}^K$ as the set (of size $|\GC|=K+1$) of the (traceless) Hermitian operators that generate the unitaries in a single layer of $U(\thv)$.
\end{definition}
Naturally, the group $\lieg$ depends on the set of generators $\GC$, yet it is not sufficient to look at the individual elements of $\GC$. Instead, one must consider the Lie algebra that emerges from their nested commutators. Hence, consider the following definition \cite{zeier2011symmetry}.

\begin{definition}[Dynamical Lie Algebra]\label{def:dynamical_lie_algebra} Given a control system with  generators $\GC$ (see Definition~\ref{def:generators}), the Dynamical Lie Algebra (DLA) $\liea$ is the subalgebra of $\mf{s}\mf{u}(d)$ spanned by the repeated nested commutators of the elements in $\GC$, i.e.,
\begin{equation}
\liea={\rm span}\left\langle iH_0, \ldots, iH_K \right\rangle_{\rm Lie}\,\subseteq \mf{s}\mf{u}(d),
\end{equation}
where $\left\langle S\right\rangle_{\rm Lie}$ denotes the Lie closure, i.e., the set obtained by repeatedly taking the nested commutators between the elements in $S$.
\end{definition}
Here, $\mf{s}\mf{u}(d)$ is the special unitary algebra of degree $d$, the Lie algebra formed by the set of $d\times d$ skew-Hermitian, traceless matrices. In Appendix~\ref{app:alg}, we lay down the basic procedure to build DLAs (see Algorithm~\ref{alg:lie}) and provide some discussion on the complexity of such construction.

Once the DLA is obtained from the set of generators, one can determine the set of unitaries that are expressible by the control system. Specifically, one can now properly define the dynamical Lie group as follows.
\begin{definition}[dynamical Lie group]\label{def:dynamical_lie_group}
The set unitaries $\lieg$ that can be generated by a control system is determined by its DLA (see Definition~\ref{def:dynamical_lie_algebra})  through~\footnote{This is grounded in the fact that every possible DLA is a subalgebra of the special unitary algebra $\mf{su}(d)$ and therefore compact. The exponential map is a function $e:\liea\rightarrow \lieg$. If $\liea$ is compact then the exponential map is surjective: every element of $\lieg$ is the image of at least one element of $\liea$.}

\begin{equation}
    \lieg = e^{\liea}:=\{e^{V},\ \ V\in\liea\}\,.
\end{equation}
\end{definition}

The dynamical Lie group in turn determines the set of states $\ket{\psi(\thv)}=U(\thv)\kp$ that can be reached by evolving an initial state $\kp$. Specifically, here $U(\thv)$ can attain values in the Lie group  $\lieg$. In addition, Definition~\ref{def:dynamical_lie_group} crucially shows that one can study the expressibility of a control system (i.e., the unitaries that can be generated, or the set of states that can be reached) via the DLA obtained from the set of generators. As shown in Fig.~\ref{fig:Algebra_classification}, when computing $\liea$ there are several cases of interest that can arise and which we here consider. For the sake of clarity, in what follows we briefly recall several key concepts that will be useful throughout the manuscript. We refer the reader to~\cite{dalessandro2010introduction} for additional details.

First and foremost, we recall the concept of \textit{controllability}. A control system is said to be \textit{controllable}  if  its DLA is full rank, i.e., $\liea=\mf{su}(d)$. This implies that $\lieg=\SC\UC(d)$ and hence every unitary (up to a phase) can be obtained by appropriately choosing control parameters in Eq.~\eqref{eq:QOC_propagator}. In particular, this means that for any two states $\ket{\psi}$ and $\ket{\phi}$, there always exists a unitary $U(\thv)\in\lieg$ such that $U(\thv)\ket{\psi}=\ket{\phi}$.

\begin{figure}[t]
\centering
\includegraphics[width=1\columnwidth]{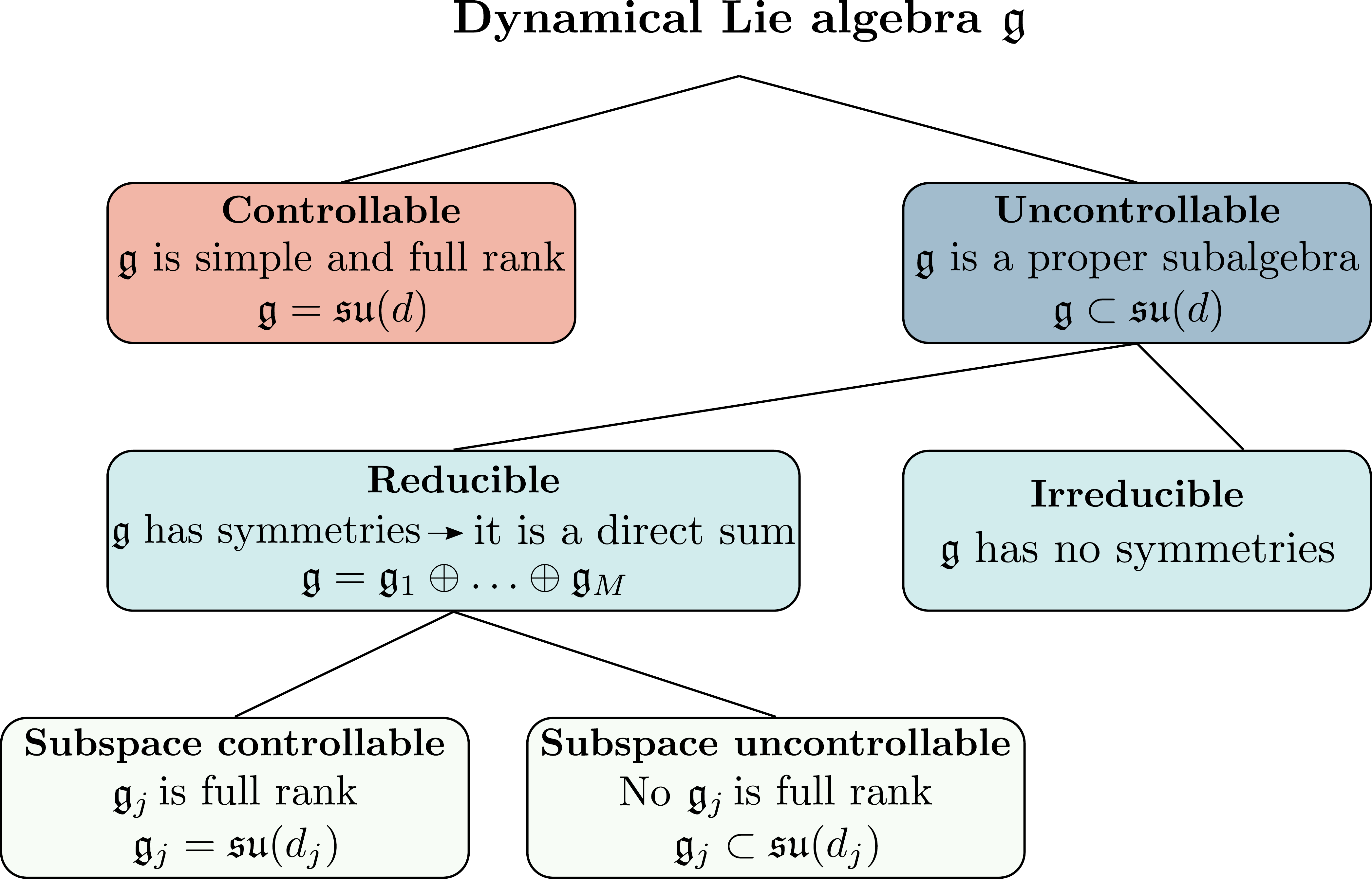}
\caption{\textbf{Cases of interest for the Dynamical Lie Algebra.} The Dynamical Lie Algebra (DLA) $\liea$ determines the set of unitaries expressible, and concomitantly, the set of states reachable. In this figure we show different scenarios that can arise when computing $\liea$. Our main results (on gradient scaling of VQAs and QOC systems) pertain to these different scenarios.}
\label{fig:Algebra_classification}
\end{figure}

If the DLA is not full rank, then the system is said to be \textit{uncontrollable}. In this case $\liea$ is  a \textit{proper} subalgebra of $\mathfrak{s u}(d)$, and only a \textit{proper} subgroup $\lieg\subset\SC\UC(d)$ is available to the control system, meaning that the set of reachable states $ \{ U(\thv)\ket{\psi}, \quad \forall\, U(\thv) \in \lieg\}$ is not the whole state space. As depicted in Fig.~\ref{fig:Algebra_classification},   there are two sources of uncontrollability \cite{polack2009uncontrollable}. 

On one hand, if the generators in $\GC$ share one or more common symmetries, i.e., there is at least one Hermitian operator $\Sigma$ that commutes with every element in $\GC$, then every $H\in \liea$ is block diagonal in the eigenbasis of $\Sigma$. This causes the state space to break into subspaces that are invariant under the action of $\liea$, in which case controllability is clearly disrupted. Here, the DLA is a \textit{reducible} representation of some Lie algebra. On the other hand, even in the absence of symmetries, that is, when the DLA is \textit{irreducible}, uncontrollability can arise simply because the Lie algebra is a proper subalgebra of $\mf{su}(d)$.

Let us finally remark that, as shown in Fig.~\ref{fig:Algebra_classification}, even though a reducible system cannot be controllable on the entire Hilbert space it may still be controllable on some (or all) of the invariant subspaces. Given a DLA that is a direct sum of irreducible representations, i.e., $\liea=\bigoplus_j \liea_j$, then the Hilbert space can be expressed as $\HC =\bigoplus_j \HC_j$, with $\H_j$ being invariant under the action of $\liea$. A system is said to be \textit{subspace controllable} on subspace $\H_j$ if $\liea_j$ is full rank, i.e., $\liea_j=\mf{u}(d_j)$, where $d_j=\dim(\H_j)$) and $\mf{u}(d_j)$ denotes the unitary algebra of degree $d_j$.

\section{MAIN RESULTS}\label{sec:main-results}

As previously discussed, VQAs and QOC can be considered as two formulations of a common variational optimization problem that optimizes parameters controlling the dynamical evolution of a quantum system. In this section we present our main results, where we basically leverage tools from QOC to analyze the trainability and the existence of barren plateaus in VQAs. Specifically, we organize our results in term of the different controllability settings shown in Fig.~\ref{fig:Algebra_classification}. In all cases, the proofs are presented in the Appendix. The main idea behind our results is that, given a PSA $U(\thv)$ as in Eq.~\eqref{eq:ansatz}, the study of the DLA of the ansatz can diagnose the presence (or absence) of barren plateaus in the VQA landscapes. 

\subsection{Controllable systems}

First, let us consider controllable systems. It is well known that the distribution of unitaries generated by controllable systems converges to a $2$-design in the long-time (i.e., for sufficiently deep circuits)~\cite{banchi2017driven}. However, the rate of convergence actually depends on the specific choice of generators. Hence, our first result  analyzes the depth at which the expressibility $\Vert\AC^{(2)}_{U(\thv)}\Vert_{\infty}$ of a controllable system is $\varepsilon$ small.
\begin{theorem}\label{theo:depth}
Consider a controllable system.  Then, the PSA $U(\thv)$ will form an $\varepsilon$-approximate $2$-design, i.e.  $\Vert\AC^{(2)}_{U(\thv)}\Vert_{\infty}=\varepsilon$ with $\epsilon>0$, when the number of layers $L$ in the circuit is
\begin{equation}\label{eq:depth}
    L=\frac{\log(1/\varepsilon)}{\log\left(1/\Vert\AC^{(2)}_{U_1(\thv)}\Vert_{\infty}\right)}\,.
\end{equation}
Here $\Vert\AC^{(2)}_{U_1(\thv)}\Vert_{\infty}$ denotes the expressibility of a single layer  $U_1(\thv_1)$ of the ansatz according to Eqs.~\eqref{eq:ansatz} and~\eqref{eq:superop}.
\end{theorem}
\noindent See Appendix~\ref{app:proof_theo_convergence} for a proof of Theorem~\ref{theo:depth}.

We note that Theorem~\ref{theo:depth} arises from the following expression that  connects the expressibility of an $L$-layered PSA to the expressibility of a single layer of the ansatz to the $L$-th power as
\begin{equation}
    \Vert\AC^{(2)}_{U(\thv)}\Vert_{\infty} =\left( \Vert\AC^{(2)}_{U_1(\thv)}\Vert_{\infty} \right)^L\,.
\end{equation}
Here we can see that  $\Vert\AC^{(2)}_{U(\thv)}\Vert_{\infty}=0$ if and only if $\Vert\AC^{(2)}_{U_1(\thv)}\Vert_{\infty}=0$.
Hence, as expected, PSAs that have more expressible layers require less depth to have an $\varepsilon$-expressibility (to be $\varepsilon$-approximate $2$-designs). Conversely, one can also see that ansatzes with less expressible layers require more depth to become $\varepsilon$-approximate $2$-design.

The following corollary analyzes the scaling of $L$.
\begin{corollary}\label{cor:Lscaling}
Let the single layer expressibility  of a controllable system be $\Vert\AC^{(2)}_{U_1(\thv)}\Vert_{\infty}=1-\delta(n)$, with $\delta(n)$ being at most polynomially vanishing with $n$, i.e., with $\delta(n)\in\Omega(1/\poly(n))$. Then, if $L(n)\in\Omega(n/\delta(n))$,  $U(\thv)$ will be no worse than an $\varepsilon(n)$-approximate $2$-design (i.e., $\Vert\AC^{(2)}_{U(\thv)}\Vert_{\infty}\leq\varepsilon(n)$) with $\varepsilon(n)\in\OC(1/2^n)$, where we have added the $n$-dependence in $L$ and $\varepsilon$ for clarity.
\end{corollary}
\noindent See Appendix~\ref{app:proof_corollary_convergence} for a proof of Corollary~\ref{cor:Lscaling}.

From Corollary~\ref{cor:Lscaling} we have that when the single layer expressibility is (at most) polynomially vanishing with $n$, then a polynomial number of layers suffice to make the PSA $U(\thv)$ exponentially close to being a $2$-design. We note, however, that in the case where the single layer expressibility is exponentially close to $1$, one requires an exponential number of layers to form an $\varepsilon$-approximate $2$-design. In all the aforementioned cases it is worth remarking that an exponential number of layers will always lead to $\varepsilon$-approximate $2$-designs with $\varepsilon\in\OC(1/2^n)$, independently of the value of $\Vert\AC^{(2)}_{U_1(\thv)}\Vert_{\infty}$.

Once the depth of the ansatz is sufficient for the controllable system to be an $\varepsilon$-approximate $2$-design, then a barren plateau will arise. Hence, one can prove the following proposition from Theorem~\ref{theo:depth} and Corollary~\ref{cor:Lscaling}.
\begin{proposition}[Controllable]\label{prop:controllable}
There exists a scaling of the depth for which controllable systems form $\varepsilon$-approximate $2$-designs with $\varepsilon\in\OC(1/2^n)$, and hence the system exhibits a barren plateau according to Definition~\ref{def:BP}.
\end{proposition}
\noindent See Appendix~\ref{app:proof_prop_controllable} for a proof of Proposition~\ref{prop:controllable}.

Proposition~\ref{prop:controllable}  rephrases the well known barren plateau results of~\cite{mcclean2018barren,holmes2021connecting} in terms of controllability. Specifically, it has been shown that when an ansatz forms a $2$-design,  such randomness leads to a barren plateau. Hence, the proof of Proposition~\ref{prop:controllable} simply follows the proof in~\cite{mcclean2018barren}, with the addition  that the convergence to a $2$-design comes from the fact that the system is controllable. 

Evidently, it becomes relevant to determine systems that are controllable as these can exhibit barren plateaus. In this work we prove that two relevant sets of generators lead to full rank DLAs, and hence to controllable systems.
\begin{proposition}\label{prop:controllable-gen}
The following two sets of generators generate full rank DLAs, and concomitantly lead to controllable systems:
\begin{itemize}
    \item $\GC_{\rm HEA}=\Big\{X_i,Y_i\Big\}_{i=1}^n \bigcup \left\{\sum_{i=1}^{n-1} Z_i Z_{i+1}\right\}$,
    \item $\GC_{\rm SG}=\left\{\sum_{i=1}^n X_i, \sum_{i<j}  \left(h_i Z_i + J_{ij} Z_iZ_j\right)\right\}$, with  $h_i,J_{ij}\in\Rbb$ sampled from a Gaussian distribution.
\end{itemize}
\end{proposition}
\noindent See Appendix~\ref{app:proof_of_hea_sg} for a proof of Proposition~\ref{prop:controllable-gen}.

The first case in Proposition~\ref{prop:controllable-gen} corresponds to the generators of PSA layered Hardware Efficient Ansatz~\cite{kandala2017hardware}, and hence Proposition~\ref{prop:controllable} indicates that this system can exhibit barren plateaus. While it is known that the layered Hardware Efficient Ansatz converges to a $2$-design for sufficient depth~\cite{harrow2009random,brandao2016local,harrow2018approximate,mcclean2018barren,cerezo2020cost}, the proof of existence of barren plateaus for this ansatz presented here is novel in that we show that the system is controllable. 

The  second result in Proposition~\ref{prop:controllable-gen} pertains to determining the ground state energy of quantum spin glasses (usually configured to encode solution to combinatorial optimization problems)~\cite{lucas2014ising,streif2020training}  with a PSA generated by $\GC_{\rm SG}$. Hence, since the system is controllable, according to Proposition~\ref{prop:controllable}, such an ansatz will also exhibit a barren plateau. This provides a no-go theorem for determining the ground state of certain spin glasses with Hamiltonians  using deep PSA variational ansatzes
.

\subsection{Subspace controllable systems}

Let us now consider the case of reducible DLAs, i.e., control systems with symmetries. Here we recall that in this case the DLA is a direct sum of the form $\liea=\bigoplus_j \liea_j$, such that any unitary $U(\thv)$ in the dynamical group $\lieg$ preserves the subspaces $U(\thv)\HC_j\subset \HC_j$.    Then, similarly to the fully controllable case, if a system is subspace controllable in a given subspace there exists a depth at which the unitaries $U(\thv)$ form $2$-designs in that subspace. In such a case, we can derive the following theorem for the variance of the cost function partial derivative with respect to a parameter $\theta_\mu\equiv\theta_{pq}$ associated to layer $p$ and generator $H_q$ (see Eq.~\eqref{eq:ansatz} for a definition of the ansatz). In the following, will slightly abuse notation and denote $H_\mu=H_q$.

\begin{theorem}[Subspace controllable]\label{Theo:subspace}
Consider a system that is reducible, i.e. so that the Hilbert space is $\HC=\bigoplus_j \HC_j$ with each $\HC_j$ invariant under the action of the dynamical Lie group $\lieg$ (see Def.~\ref{def:dynamical_lie_group}), and controllable on some $\HC_k$ of dimension $d_k$ (i.e. $\liea_k=\mf{u}(d_k)$ or $\mf{s}\mf{u}(d_k)$). Consider a cost function $C(\thv)$ in the form of Eq.~\eqref{eq:cost} and suppose that the number of layers $L$ in the circuit is enough to allow the distribution of unitaries $U(\thv)$ to be $\varepsilon$ close to a 2-design in $\HC_k$. Then, if the initial state is such that $\rho\in\HC_k$, the variance of the cost function partial derivative with respect to parameter $\th_\mu$ is given by
\small
\begin{equation}\label{Eq_th2}
\Var_{\vec{\theta}}[\partial_{\mu} C(\thv)] =\frac{2d_k}{(d_k^2-1)^2} \Delta(H_{\mu}\k)\Delta(O\k)\Delta(\rho\k)\,,
\end{equation}
\normalsize
where $O$ is the operator whose expectation value is being minimized and $H_\mu$ is the generator of the corresponding gate. Here $\Delta(A) = D_{HS}\left(A,\Tr[A]\frac{\id_d}{d}\right)$, with  $D_{HS}\left(A,B\right)=\Tr[(A-B)^2]$ the Hilbert-Schmidt distance, and where we defined
$A\k$ as the reduction of operator $A$ onto the subspace of $\H_k$. Explicitly, $A\k=Q_k A Q_k\ad$ with $Q_k$ a matrix (of dimension $d_k\times 2^n$) with columns corresponding to a basis of $\H_k$.
\end{theorem}
\noindent See Appendix~\ref{app:proof_theo_subspace} for a proof of Theorem \ref{Theo:subspace}.

Theorem~\ref{Theo:subspace}  shows that the input state $\rho$ can actually play a crucial role in determining the gradient scaling of the cost function. Specifically, if $\rho$ belongs to an invariant subspace where the system is controllable, then the scaling of the cost function partial derivative variance is determined  by the dimension of the invariant subspace rather than by the dimension $d=2^n$ of the Hilbert space. Hence, the cost function $C(\thv)$ might exhibit barren plateaus in some  subspaces but not in others. This is formalized in the following corollary.

\begin{corollary}\label{cor:HXXZ}
Consider an ansatz of the form in \eqref{eq:ansatz} giving rise to a reducible DLA, and let $\rho\in\HC_k$, with $\HC_k$ some invariant subspace that is controllable (i.e. the DLA reduced to such subspace is full rank). The following bound holds

\begin{equation}
    \Var_{\thv}[\partial_{\mu} C(\thv)]\leq\frac{4d_k^2}{(d_k^2-1)^2}\sqrt{\Tr[H_\mu^4]} \sqrt{\Tr[O^4]} \,.
\end{equation}
That is, provided $\Tr[(H_\mu)^4],\Tr[O^4]\in \OC(2^n)$, the cost function will exhibit a barren plateau for any subspace such that $d_k\in\OC(2^n)$.
\end{corollary}

\noindent See Appendix~\ref{app:proof_corollary_subspace} for a proof of Corollary~\ref{cor:HXXZ}. In addition, in this appendix, we also note relevant cases for which $\Tr[(H_\mu)^4],\Tr[O^4]\in \OC(2^n)$ holds.

\begin{figure}[t]
\includegraphics[width=.9\columnwidth]{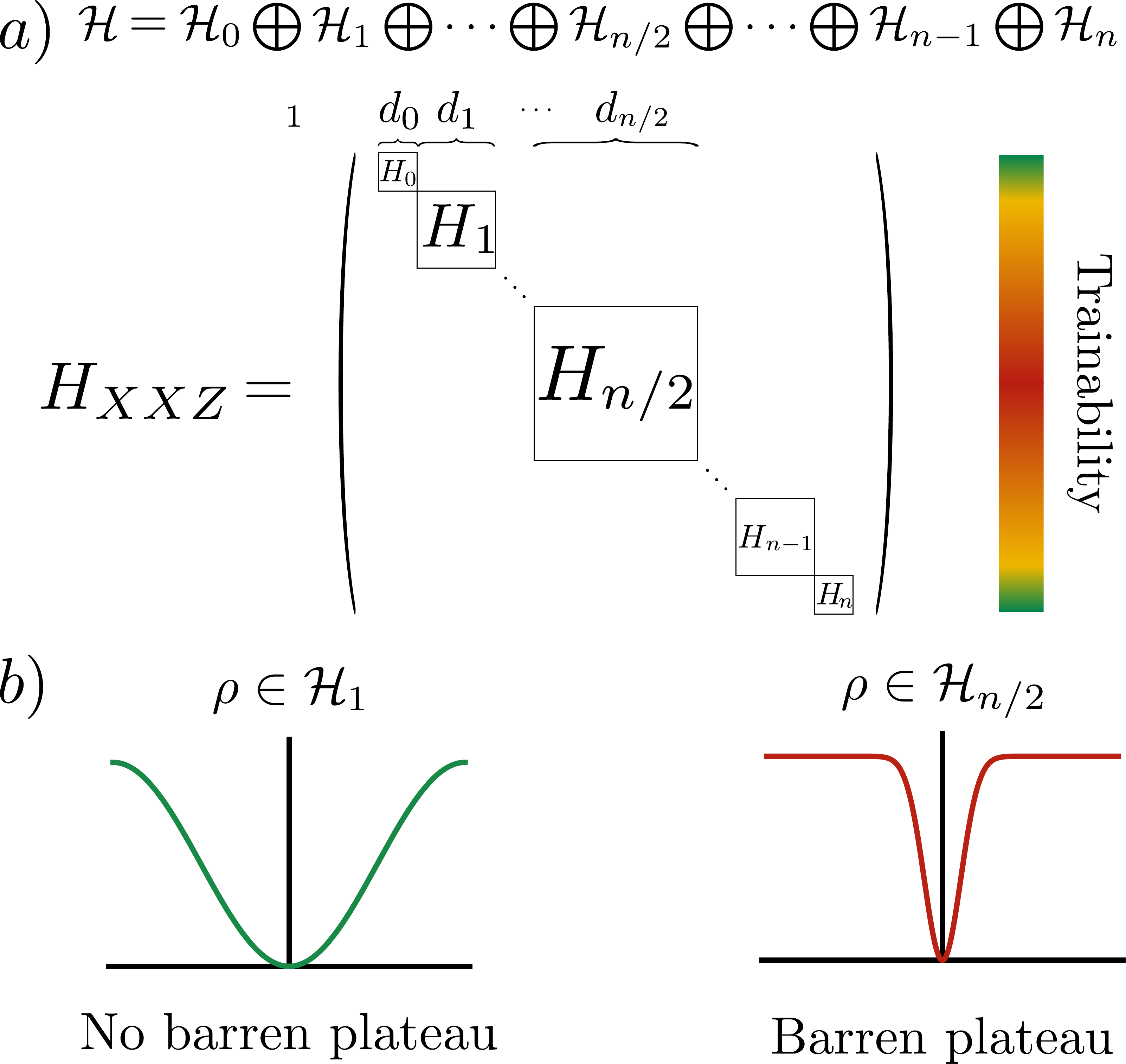}
\caption{\textbf{Trainability analysis for a PSA used to prepare the ground state of a Heisenberg $XXZ$ Hamiltonian.} a) Due to the symmetries in the $XXZ$ Hamiltonian of Eq.~\eqref{eq:HXXZ}, the Hilbert space spanned by the PSA  can be expressed as a direct sum of invariant subspaces $\HC_k$ (of dimension $d_k=\binom{n}{k}$) composed of states with $k$ excitations. Hence both $H_{XXZ}$ and $U(\thv)$  generated by Eq.~\eqref{eq:genHVAXXZ} will be block diagonal.  Since the system is subspace controllable in each invariant subspace, the gradient scaling of the PSA can be analyzed via Theorem~\ref{Theo:subspace}. b) The existence or absence of barren plateaus is directly determined by the dimension of the invariant subspace to which the input state $\rho$ belongs. For instance, the cost can be trainable for $\rho\in\HC_1$, but will exhibit a barren plateau if $\rho\in\HC_{n/2}$, as in the latter case the dimension $d_{n/2}$ is exponentially large.}
\label{fig:HXXZ}
\end{figure}

For example, let us consider the problem of preparing the ground state of a Heisenberg $XXZ$ spin chain
\small
\begin{equation}\label{eq:HXXZ}
    H_{XXZ}=\sum_{i=1}^{n-1} \left(X_iX_{i+1}+ Y_i Y_{i+1} + J Z_iZ_{i+1}\right)
\end{equation}
\normalsize
with a PSA generated by $\GC=\GC_{XXZ_U} \bigcup \{Z_1\}$. Here, the XXZ generators
\small
\begin{equation}\label{eq:genHVAXXZ}
    \GC_{XXZ_U}= \left\{\sum_{i} X_iX_{i+1}+ Y_i Y_{i+1},\sum_{i} Z_iZ_{i+1}\right\}_{i:even,odd}\!\!\!\!\!
\end{equation}
\normalsize
are accompanied by a \textit{control} generator $Z_1$, which is introduced precisely to make the system (subspace) controllable. We remark that we here employ a $U$ subindex in $\GC_{XXZ_U}$ to indicate that this set of generators is uncontrollable. Since all elements in $\GC_{XXZ_U}$ commute with $M_z=\sum_{i=1}^n Z_i$, the Hilbert space fractures into $n+1$ invariant subspaces of fixed excitation~\footnote{We recall that a state  $\ket{\psi}$ has $m$ excitations if it can be expressed as a linear combination of computational basis states with Hamming weight $m$} $\HC=\bigoplus_{m=0}^n \HC_m$,  of dimension $d_m=\dim(\HC_m)=\binom{n}{m}$ \cite{cerezo2017factorization}. 

Because the example set $\GC$ has a DLA that is full rank on every subspace~\cite{wang2016subspace}, we can analyze the trainability of such a PSA using Theorem~\ref{Theo:subspace}. 
The implications of Corollary~\ref{cor:HXXZ} for such a VQA are schematically shown in Fig.~\ref{fig:HXXZ}. Here we find that the presence, or absence, of barren plateaus for the PSA $U(\thv)$ generated by Eq.~\eqref{eq:genHVAXXZ} is completely determined by the scaling of the invariant subspace to which the input state belongs. For instance, the cost may not exhibit a barren plateau if $\rho$ has a number of excitations that does not scale with $n$, while it will have a barren plateau for $k=n/2$ (as in this case the dimension $d_{n/2}$ scales exponentially with the number of qubits).

\subsection{Uncontrollable and reducible systems}

Analyzing the scaling of the gradients in the case of uncontrollable systems becomes much more intricate than in the controllable or subspace controllable cases, mainly because integrating over the Haar measure of proper subgroups of the unitary group is not so straightforward~\cite{collins2006integration}. As shown in this (and the next) section, one can still obtain a few analytical results for these cases. In particular, one can derive an upper bound for the variance of partial derivatives in terms of the degree of expressibility on the invariant subspaces, in a spirit similar to that of~\cite{holmes2021connecting}.

Before presenting our main results for  uncontrollable and reducible systems, it is convenient to introduce some notation. We will use $U_B$ and $U_A$, respectively, to address the portions of the circuit that come before and after a given parameter $\theta_\mu\equiv\theta_{pq}\in\thv$. That is,
\begin{equation}
    U_{B} = \prod_{m=0}^q e^{-i H_m \theta_{pm}} \left(\prod_{l=1}^{p-1} \prod_{k=0}^K e^{-i H_k \theta_{lk}}\right)\,,
\label{eq:UB}
\end{equation}
and
\begin{equation}
    U_{A} = \prod_{l=p+1}^L \prod_{k=0}^K e^{-i H_k \theta_{pk}}\left(\prod_{k=q+1}^K e^{-i H_k \theta_{pk}}\right)\,,
\label{eq:UA}
\end{equation}
where we have omitted the $\thv$ dependency for simplicity.

Then, the following theorem holds.
\begin{theorem}\label{theo:express-subspace2}
Consider a system that is reducible and let $\rho\in\HC_k$ with $\HC_k$ an invariant subspace of dimension $d_k$. Then, the variance of the cost function partial derivative is upper bounded by
\small
\begin{equation}\label{eq:expressibility}
\Var_\theta[ \partial_{\mu} C(\thv)]  \leq  \min\{G_A(\rho\k),G_B(O\k)\}\,,
\end{equation}
\normalsize
with
\footnotesize
\begin{align}
  G_B(\rho\k) &= \left(\left\Vert \mathcal{A}_{U_B\k}\left((\rho\k)^{\otimes 2}\right)\right\Vert_2 -\frac{\Delta(\rho\k)}{d_k^2-1}\right)\Tr\left[\left\langle X^2\right\rangle_{U_A\k}\right]\nonumber\\
  G_A(O\k) &= \left(\left\Vert \mathcal{A}_{U_A\k}\left((O\k)^{\otimes 2}\right)\right\Vert_2 -\frac{\Delta(O\k)}{d_k^2-1}\right)\Tr\left[\left\langle Y^2\right\rangle_{U_B\k}\right].\nonumber
\end{align}
 \normalsize
Here we defined $X= [H_\mu\k, (U_A\k\,)^\dagger O\, U_A\k]$ and $ Y = [H_\mu\k, (U_B\k) \rho\k (U_B\k)\ad]$. For simplicity, we here employed the short-hand notation   $\langle\cdot \rangle_{U_x\k}$ (with $x=A,B$) indicates the expectation value over the distribution of unitaries obtained from $U_x\k$ in the $k$-th subsystem. Finally, $\Vert M\Vert_2=\sqrt{\Tr[M\ad M]}$ is the Frobenius norm, and $\Delta(\cdot)$ was defined in Theorem~\ref{Theo:subspace}.
\end{theorem}
\noindent See Appendix~\ref{app:proof_theo_express} for a proof of Theorem~\ref{theo:express-subspace2}.

Theorem~\ref{theo:express-subspace2} generalizes the expressibility results in~\cite{holmes2021connecting} to invariant subspaces. More specifically, Theorem~\ref{theo:express-subspace2}  provides a bound on the variance of the cost function partial derivative $ \partial_{\mu} C(\thv)$ as a function of the ansatz expressibility on the relevant invariant subspace. Hence, similar to the results observed in~\cite{holmes2021connecting}, the more expressible an ansatz is in a subspace, the smaller the gradients will be. Moreover, ansatzes that are very expressible in subspaces with exponentially large dimensions can exhibit barren plateaus as the variance of the cost function partial derivative will vanish exponentially, according to Eq.~\eqref{eq:expressibility}.

\subsection{Uncontrollable and irreducible systems}\label{sec:Irre}

Here, we analyze a case where the DLA is an irreducible representation of some proper subalgebra of $\mf{su}(d)$. Specifically, we consider a toy model ansatz $U(\thv) =\prod_{l=1}^L e^{-i\theta_{lx} S_y}e^{-i\theta_{ly} S_x}$ with generators $\GC=\{S_x,S_y\}$, where $\liea=\{iS_x,iS_y,iS_z\}$ is the spin $S=(d-1)/2$ irreducible representation of $\mf{su}$(2). That is, $[S_j,S_k]=2i\epsilon_{jkl} S_l$ with $\epsilon$ the Levi-Civita symbol and $j,k,l\in \{x,y,z\}$.  We address the task of minimizing a cost function of the form
\begin{equation}\label{eq:costsu2}
    C(\thv)=\bramatket{m}{U\ad(\thv)(S_x+S_y+S_z)U(\thv)}{m}\,,
\end{equation}
where $\ket{m}$ is an eigenstate of $S_z$, i.e., $S_z\ket{m}=m\ket{m}$ with $m\in \{-S,-S+1,\ldots,S-1,S\}$. 

Let us analyze the variance of partial derivative with respect to parameter $\theta_\mu=\theta_{jx}$, i.e., the one corresponding to the generator $S_x$ on the $j$-th layer. Assuming a depth $p$ such that that the distributions $U_A(\thv)$ and $U_B(\thv)$ converge to $\varepsilon$-approximate $2$-designs on the dynamical Lie group $\lieg$ (which in this case is the $d$-dimensional irreducible representation of $\SC\UC(2)$), we are able to explicitly integrate over the Haar measure on $\lieg$ and find the following proposition to hold.
\begin{proposition}\label{prop:varsu2}
Consider the cost function of Eq.~\eqref{eq:costsu2}. Let $\th_{\mu}=\th_{j,x}$, and let us assume that the circuit is deep enough to allow for the distribution of unitaries $U_A$ and $U_B$ to converge to $2$-designs on $\lieg=\SC\UC$(2). Then variance of the cost function partial derivative $\partial_{\mu} C(\thv)=\partial C(\thv)/\partial\theta_\mu$ is
\begin{equation}\label{eq:varsu2}
    \Var_{\thv}[\partial_{\mu} C(\thv)]=\frac{2m^2}{3}\,.
\end{equation}
\end{proposition}
\noindent See Appendix~\ref{app:proof_su2} for a proof of Proposition~\ref{prop:varsu2}.

Proposition~\ref{prop:varsu2}  shows that the variance of the cost function again depends on the input state $\ket{m}$, which is a similar result to the one obtained in Theorem~\ref{Theo:subspace}. Moreover, here $\Var_{\thv}[\partial_{\mu} C(\thv)]$ can in fact be as large as $d^2$. This is due to the fact that the ``size'' (the difference between maxima and minima) of the landscape also grows with $d$. One can get rid of this effect by considering a normalized cost instead, $\widetilde{C}(\thv)=C(\thv)/S$, where the \textit{ad-hoc} factor $1/S$ guarantees that the landscape is $|C(\thv)|\leq 1$ for all values of $d$. The variance of such normalized landscape is $\Var_{\thv}[\partial_{\mu} \widetilde{C}(\thv)]=\frac{2}{3}\frac{m^2}{S^2}=\frac{8}{3}\frac{m^2}{(d-1)^2}$, that is, vanishes exponentially for initial states with $|m|\in\OC(\poly(\log(d)))$. Similar to the subspace controllable results in Corollary~\ref{cor:HXXZ}, here the choice of initial state is again crucial as it can lead to the cost function exhibiting barren plateaus. 

\subsection{General case: linking gradient scaling to the dimension of the Lie algebra}\label{sec:conjecture}

In this section we note that the dimension of the DLA can be linked to the scaling of the variance of the cost function partial derivatives. This opens up the possibility of diagnosing the existence of barren plateaus of uncontrollable systems by analyzing the scaling of their DLAs.
First, let us remark that a key aspect of the toy model in Section~\ref{sec:Irre} is that the dimension of the DLA is $\dim(\liea)=3$. This is independent of the dimension $d$ of the Hilbert space it acts on. Moreover, as shown in Eq.~\eqref{eq:varsu2} the variance is also independent of $d$ as it does not present the typical dimensional-dependent factor in the denominator that one usually obtains when integrating over unitary 2-designs (see Eq.~\eqref{Eq_th2} in Theorem~\ref{Theo:subspace}). 

For instance, when the system is controllable, $\dim(\liea)=d^2-1=2^{2n}-1$, and thus the dimension of the DLA is exponentially growing with the system size $n$. Concomitantly, one finds that $\Var_{\vec{\theta}}[\partial_{\mu} C(\thv)] =\frac{2d}{(d^2-1)^2} \,\,g(H_{\mu},O,\rho)$~\cite{mcclean2018barren,cerezo2020cost}, and hence the variance is exponentially vanishing with the system size. A similar result is obtained in the subspace controllable case (see Theorem~\ref{Theo:subspace}) where the variance is of the form $\frac{2d_k}{(d_k^2-1)^2}g(H_{\mu}^{(k)},O^{(k)},\rho\k)$. 

These facts have led us to conjecture that the dimension of the DLA plays a key role in determining the presence or absence of barren plateaus in the cost function landscape. More specifically, for PSAs with sufficient depth (i.e., with a depth such that the distribution of unitaries generated by $U(\thv)$ has converged to the Haar measure in the Lie group $\lieg$), we have noted that the following conjecture appears to hold.
\begin{conjecture}\label{conjecture}
Let the state $\rho$ belong to a subspace $\HC_k$ associated with a subspace DLA $\liea_k$ (or sub-DLA, the subrepresentation in $\liea$ where $\rho$ has support on). Then,  the scaling of the variance of the cost function partial derivative is inversely proportional to the scaling of the dimension of the DLA, i.e.
\begin{equation}\label{Eq:conjecture}
    \Var_{\vec{\theta}}[\partial_{\mu} C(\thv)]\in\OC\left(\frac{1}{\poly(\dim(\liea_k))}\right)\,.
\end{equation}
\end{conjecture}

The implications of Conjecture~\ref{conjecture} are as follows. First, it means that systems with a sub-DLA $\liea_k$~\footnote{Let us note that one should look at the sub-DLA instead of the full DLA, since when there are symmetries and the initial state belongs to one or multiple invariant subspaces, the dynamics is contained in those.} that is polynomially growing with the system size can exhibit gradients that vanish only polynomially, and hence may not exhibit barren plateaus. Conversely, systems with a sub-DLA that is exponentially growing with the system size would exhibit gradients that vanish exponentially with the system size, hence exhibiting barren plateaus. Here, we remark that systems with sub-DLAs that are not exponentially growing may still have barren plateaus which are not related to the dimension of the DLA. For instance, if the cost function is global, the system can still exhibit barren plateaus even with trivial ansatzes that do not a have exponentially growing dimension of the DLA~\cite{cerezo2020cost}.

Using Conjecture~\ref{conjecture}, one could diagnose gradient scalings by determining the size of the Lie algebra of a given ansatz $U(\thv)$. This comes at the cost of taking the set of generators $\GC$ and computing the DLA. While numerical methods (as in Algorithm~\ref{alg:lie}) can prove valuable insights for small system sizes, these algorithms will generally scale poorly in the number of qubits. Hence, performing a theoretical analysis of the DLA (similar to the one performed in Proposition~\ref{prop:controllable-gen})  is a preferable method. 

Here we remark that there are simple (yet pathological) cases that show that Eq.~\eqref{Eq:conjecture} does not preclude the possibility that systems with algebras that grow polynomial with the system size may still exhibit barren plateaus. For instance, consider Eq.~\eqref{Eq_th2}, where $\rho$ belongs to a subspace with polynomially growing algebra: $d_k\in\OC(\poly(n))$. Then, note that if the input state is exponentially close to being maximally mixed on $\HC_k$ (i.e., if $\Delta(\rho\k)\in\OC(1/2^n)$)  one can easily verify that  the system will exhibit a barren plateau according to Definition~\ref{def:BP} as the cost function partial derivative will be exponentially vanishing. Here, the barren plateau arises not from the dimension of the DLA being exponentially large but rather from trying to train a VQA on an input state that is exponentially close to being maximally mixed. A similar result can be found if $H_{\mu}^{(k)}$ is exponentially close to the identity. Hence, we remark that  Conjecture~\ref{conjecture} does not imply that systems with polynomially growing algebras are exempt from having barren plateaus, as cases where $\rho$ ($H_{\mu}^{(k)}$) is exponentially close  to being maximally mixed (the identity) will naturally be hard to train from the definition of the cost function in Eq.~\eqref{eq:cost}.

We finally note that to further support the claim in Conjecture~\ref{conjecture}, we present in the following section  results obtained from numerically computing the scaling of the variance of the cost function partial derivatives for systems with DLAs having several different dependencies on the number of qubits. As discussed in Section~\ref{sec:results}, we see that the the result in Conjecture~\ref{conjecture} holds true for all cases considered, as in these cases the scaling of the variance of the cost function partial derivative is inversely proportional to the scaling of the dimension of the DLA. In addition, based on our conjecture one can accurately make predictions regarding whether a given modification to an ansatz (adding a new generator to $\GC$ by introducing a new unitary in each layer) might improve or be detrimental to the trainability of the parameters.

\section{NUMERICAL SIMULATIONS}\label{sec:results}

In this section we present results obtained by numerically computing the variance of the cost function partial derivatives for systems with different PSAs, and with DLAs of dimensions with different scaling. In particular, we consider systems that are controllable, subspace controllable, and subspace uncontrollable. As we show, in all cases Conjecture~\ref{conjecture} is verified. Finally, we refer the reader to Appendix \ref{app:proof_su2} for a numerical study of the toy model in Section~\ref{sec:Irre}, where  $\liea$ and $\lieg$ are, respectively, the $d$-dimensional irreducible representations of $\mf{su}$(2) and $\SC\UC$(2).

\begin{figure}[t]
\includegraphics[width=.9\columnwidth]{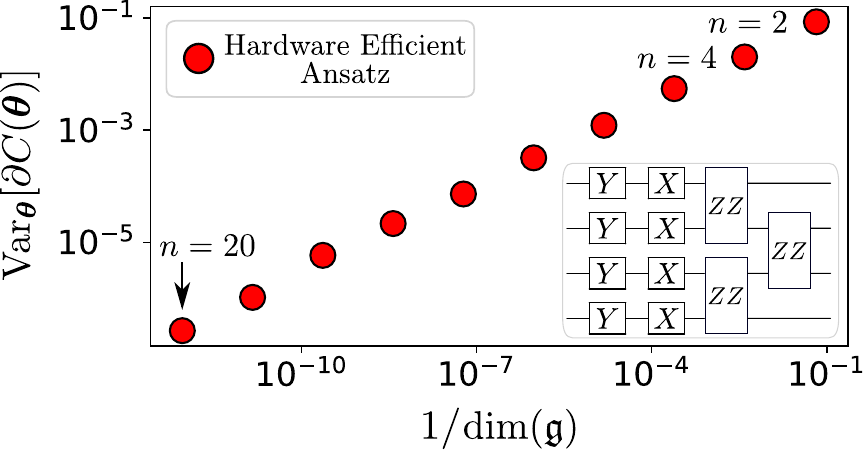}
\caption{\textbf{Variance of cost function partial derivatives versus inverse of the DLA dimension for a controllable system.} The layered Hardware Efficient Ansatz (shown in the inset for $n=4$) is a controllable system with generators given  in Proposition~\ref{prop:controllable-gen}. Then, as shown in Proposition~\ref{prop:controllable}, the cost function of Eq.~\eqref{eq:cost_HEA} exhibits a barren plateau and hence $\Var_{\thv}[\partial_{\mu} C(\thv)]\in\OC(1/2^n)$. Moreover,  since the system is controllable one finds that $\dim(\liea)=4^{n}-1$. Hence, as shown in the plot, Conjecture~\ref{conjecture} holds for controllable systems, since  the dependence of $\Var_{\thv}[\partial_{\mu} C(\thv)]$ versus $1/\dim(\liea)$ is linear on a log-log scale. }
\label{fig:HEA}
\end{figure}

\subsection{Controllable systems}

\begin{figure*}[th]
\includegraphics[width=1\linewidth]{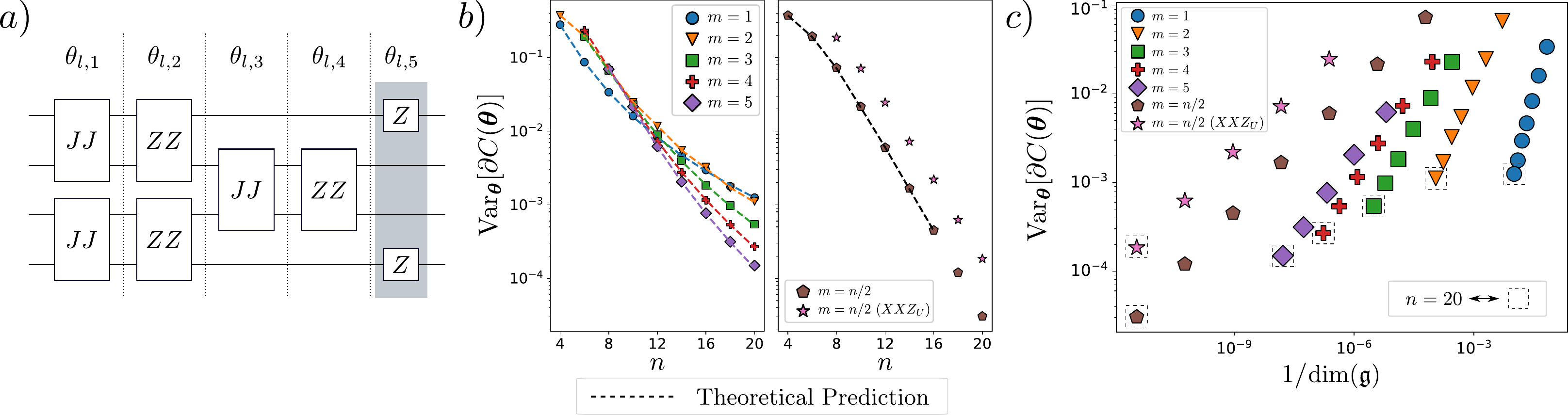}
\caption{\textbf{Numerical results for the $XXZ$ model.} a) Schematic illustration of a single layer of the subspace controllable periodic structure  ansatz for the set of generators $\GC_{XXZ}$ in Eq.~\eqref{eq:gcXXZ} for $n=4$ qubits.  Removing the unitary in the shaded area leads to the uncontrollable periodic structure HVA generated by $\GC_{XXZ_U}$.  b) Left panel: Variance of the partial derivative of the cost in Eq.~\eqref{eq:costxxz} versus the number of qubits $n$. The different markers correspond to initial states with a number of excitations $m=1,2,\ldots,5$ (left) with an ansatz generated by $\GC_{XXZ}$, and to  $m=n/2$ (right) with an ansatz generated by $\GC_{XXZ}$ and by $\GC_{XXZ_U}$. Here we recall that a state   $\ket{\psi}$ has $m$ excitations if it can be expressed as a linear combination of computational basis states with Hamming weight $m$. The dashed lines represent the theoretical prediction of Eq.~\eqref{Eq_th2}. In both cases the plot is shown in a log-linear scale. For $m=1,2,\ldots,5$ the variance is polynomially vanishing with $n$, while for $m=n/2$ the variance is exponentially vanishing with $n$.  c) Variance of the cost function partial derivative versus  $1/\dim(\liea)$. The plot is shown in a log-log scale. For each value of $m$ we see a linear dependence, which verifies Conjecture~\ref{conjecture}.}
\label{fig:dipXXZ}
\end{figure*}

First, let us remind that when the system is controllable, $U(\thv)$ forms a $2$-design (see Proposition~\ref{prop:controllable}). In this case the scaling of $\Var_{\thv}[\partial_{\mu} C(\thv)]$ has been widely analyzed in the literature (see for instance~\cite{mcclean2018barren,holmes2021connecting}). Controllable systems, as previously discussed, satisfy Conjecture~\ref{conjecture}. In Fig.~\ref{fig:HEA} we show the variance of cost function partial derivatives as a function of $1/\dim(\liea)$ for the cost function 
\begin{equation}\label{eq:cost_HEA}
    C(\thv)=\bramatket{\vec{0}}{U\ad(\thv)(Z_1\otimes Z_2)U(\thv)}{\vec{0}}\,.
\end{equation}
Here, $U(\thv)$ is  a layered Hardware Efficient ansatz (see the circuit in the inset of  Fig.~\ref{fig:HEA}) with $200$ layers and where $\ket{\vec{0}}=\ket{0}^{\otimes n}$. For each value of $n=2,4,\ldots,20$, the variance was computed by randomly initializing $1000$ sets of parameters.  Since this system is controllable (as proved in Proposition~\ref{prop:controllable}), then $\dim(\liea)=d^2-1=4^{n}-1$. In Fig.~\ref{fig:HEA} we see that, as expected, the variance is a polynomial function of $1/\dim(\liea)$ (indicated by a straight line in a log-log scale).

\subsection{Reducible systems}

\subsubsection{The XXZ model}\label{section:XXZ}

Let us first consider the task of finding the ground state energy of the XXZ Hamiltonian $H_{XXZ}$ of Eq.~\eqref{eq:HXXZ}. First, let us notice that $\GC_{XXZ_U}$, the uncontrollable set of generators of Eq.~\eqref{eq:genHVAXXZ}, has two symmetries: magnetization and parity. Hence, the DLA is reducible, i.e. a sum of irreducible sub-representations $\liea_{XXZ_U} = \bigoplus_{\substack{m=0\\ \sg=\pm}}^n \liea_{m,\sigma} \subseteq \mf{u}(d_{m,\sigma})$, where the indices $m$ and $\sg$ indicate number of excitations and parity, respectively (see Appendix \ref{App:XXZ} for details). Notably, the system can be rendered subspace controllable (while preserving the invariant subspace structure) by introducing an additional generator consisting of local fields at the ends of the chain \cite{poggi2016optimal,larocca2021krylov}
\begin{equation}\label{eq:gcXXZ}
     \GC_{XXZ}=\GC_{XXZ_U}\cup \left\{\Z_1+\Z_N\right\}\,.
\end{equation}
The new set $\GC_{XXZ}$ generates a DLA that is full rank on each of the invariant subspaces, i.e.  $\liea_{\rm XXZ} = \bigoplus_{\substack{m=0\\ \sg=\pm}}^n \mf{u}(d_{m,\sigma})$. In Figure~\ref{fig:dipXXZ}(a), we sketch a single layer of the ansatz generated by $\GC_{XXZ}$. Note that upon the removal of the unitary generated by $Z_1+Z_N$ (indicated by a shaded area), one recovers the HVA ansatz with generators $\GC_{XXZ_U}$ proposed in Ref.~\cite{wiersema2020exploring}.

Figure~\ref{fig:dipXXZ}(b) shows numerical results obtained by computing $\Var_{\thv}[\partial_{\mu} C(\thv)]$ for the cost function
\begin{equation}\label{eq:costxxz}
    C(\thv)=\bramatket{\psi_{m,+}}{U\ad(\thv)H_{XXZ}U(\thv)}{\psi_{m,+}}/n\,,
\end{equation}
with $J=1$. Here, $U(\thv)$ is the HVA ansatz generated by $\GC_{XXZ}$ (see Eq.~\eqref{eq:gcXXZ}) with $L=6n$ layers, and $\ket{\psi_{m,+}}$ is an initial state with $m$ excitations and even parity $\sigma=+$ (see Appendix \ref{App:numerical} for details).  For each system system size $n=2,4,\ldots,20$, and for each value of $m$, we computed the variance with respect to $\theta_{\frac{L}{2},2}$  by randomly initializing each $\thv_{pq}\in[0,2\pi]$ and averaging over $9500$ sets of parameters (for $n=20$ we averaged over $2700$ sets of parameters). 

In Figure~\ref{fig:dipXXZ}(b, left) we see that for $m=1,2,\ldots, 5$ the variance of the cost function partial derivative is polynomially decreasing with $n$, indicating that the cost function does not exhibit a barren plateau for initial states with fixed number of excitations. However, in the case $m=n/2$ (see Figure~\ref{fig:dipXXZ}(b, right)), one can observe that $\Var_{\thv}[\partial_{\mu} C(\thv)]$ vanishes exponentially. In addition, in Figure~\ref{fig:dipXXZ}(b) we also show the curves for $\Var_{\thv}[\partial_{\mu} C(\thv)]$ obtained from the analytical result in Eq.~\eqref{Eq_th2} of Theorem~\ref{Theo:subspace}. The agreement between theoretical and numerical results indicates that, already for the linear depths used in the experiments, the ansatz is well converged to a $2$-design. Hence, the results in Theorem~\ref{Theo:subspace} suggest that  the system  will exhibit a barren plateau when initialized on any subspace where $d_m\in\OC(2^n)$, for example, in the case of $m=n/2$ excitations. 

In addition, Figure~\ref{fig:dipXXZ}(b, right) shows the scaling of the variance for the PSA generated by $\GC_{XXZ_U}$, with an initial state with $m=n/2$. As previously noted, this case is not controllable and hence Theorem~\ref{Theo:subspace} does not hold. However, the gradient scaling of the cost function can still be diagnosed using the expressibility result of Theorem~\ref{theo:express-subspace2}. First, we note that the variance values for the uncontrollable case are larger than the ones for the controllable case. This result is in accordance with the fact that the smallest variances are reached with the higher expressibilities. Still, despite the system not being controllable, we find that the cost function still exhibits a barren plateau as the cost vanishes exponentially with $n$.

In Figure~\ref{fig:dipXXZ}(c) we show that  for all subspace controllable cases considered, Conjecture~\ref{conjecture} holds. Specifically, we have shown $\Var_{\thv}[\partial_{\mu} C(\thv)]$ as a function of $  1/\dim(\liea)$, and we see a linear dependence in a log-log scale. This is true both for the exponentially growing algebras ($m=n/2$) as well as for the polynomially growing algebras ($m=1,2,\ldots,5$). Moreover, we see that the Conjecture is verified on the subspace uncontrollable case of $\GC_{XXZ_U}$ (pink stars), where the dimension of the DLA is exponentially growing, and concomitantly, the variance of the cost function partial derivative is exponentially suppressed.

\subsubsection{The Ising Model}\label{section:Ising}

In this section we present results obtained for numerically simulating the use of a PSA to find the ground state of the Ising model. Specifically, consider the Hamiltonian of the one-dimensional Transverse Field Ising Model (TFIM) 
\begin{equation}\label{eq:TFIM}
    H_{\TFIM}=\sum_{i=1}^{n_f} Z_iZ_{i+1}+h_x\sum_{i=1}^{n} X_i\,,
\end{equation}
where $n_f=n-1$ in the case of open boundary conditions, and $n_f=n$ in the periodic boundary conditions case (where $Z_{n+1}\equiv Z_1$). Then, as shown in Figure~\ref{fig:Ising}(a, left), the ansatz is generated by the set
\begin{equation}\label{eq:genTFIM}
    \GC_{\TFIM}= \left\{\sum_{i=1}^{n_f} Z_iZ_{i+1},\sum_{i=1}^{n} X_i\right\}\,.
\end{equation}
Note that the PSA generated by $\GC_{\TFIM}$ is in fact the QAOA employed for solving the MAXCUT problem on a 2-regular graph~\cite{farhi2014quantum,hadfield2019quantum}.

As discussed in Appendix~\ref{App:XXZ}, the generators in $\GC_{\TFIM}$ (with open boundary conditions) have two symmetries: parity symmetry $\Pi$, and the so-called $\Zbb_2$ symmetry $\Pi_{\Zbb_2}$ (representing an invariance under a global flip in the qubits). The Hilbert space is broken into four invariant subspaces, $\HC=\bigoplus_{\sigma,\sigma'} \HC_{\sigma,\sigma'}$, where $\sigma,\sigma'=\pm1$ respectively spanning the eigenvalues of $\Pi$ and $\Pi_{\Zbb_2}$, and where $\dim(\HC_{\sigma,\sigma'})$ is exponentially growing, i.e., $\dim(\HC_{\sigma,\sigma'})\in\OC(2^n)$. In turn, the DLA decomposes as $\liea_{\rm TFIM} = \bigoplus_{\substack{\sigma,\sigma'}}\liea_{\sigma,\sigma'} \subseteq \mf{u}(d_{\sigma,\sigma'})$. However, employing Algorithm~\ref{alg:lie} we computed the dimension of the DLA generated by $\GC_{\TFIM}$ and we found that it only grows polynomially with $n$. That is, we obtain that
\begin{equation}\label{eq:dimTIM}
    \dim(\liea_{\rm TFIM})=n^2\,.
\end{equation}
Clearly, this implies that $\dim(\liea_{\sigma,\sigma'})\leq n^2$ for all $\sigma,\sigma'$.

Note that the set $\{ \id,\Pi\}$ constitutes a representation of $S_2$, the symmetric group of two elements, under which the open-boundary-condition TFIM generators are invariant. Instead, the TFIM generators with closed boundary conditions are invariant under a representation of $C_n$, the cyclic group of $n$ elements. As discussed in Appendix~\ref{App:XXZ}, the dimension of the DLA now grows linearly instead of quadratically.

\begin{figure}[t]
\centering
\includegraphics[width=1\columnwidth]{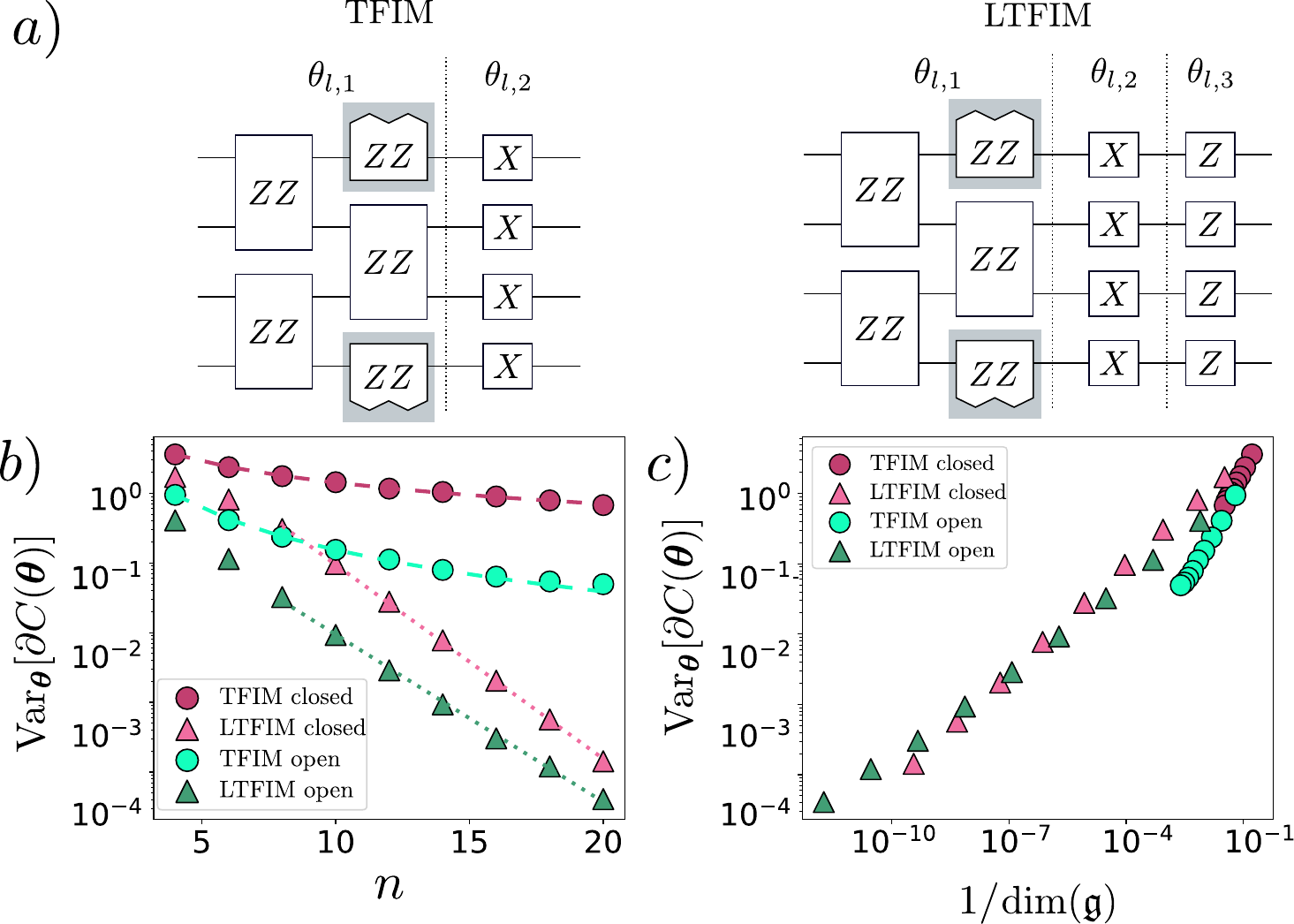}
\caption{\textbf{Numerical results for the TFIM and LTFIM models.} a) Schematic illustration of a single layer of the PSA for the sets of generators $\GC_{\TFIM}$ in Eq.~\eqref{eq:genTFIM} (left), and $\GC_{\LTFIM}$ in Eq.~\eqref{eq:genLTFIM} (right) for $n=4$ qubits. By adding (removing) the gates in the shaded one obtains the ansatz in Eq.~\eqref{eq:genTFIM} with periodic (open) boundary conditions.   b) Variance of the cost function partial derivative of the cost function in Eq.~\eqref{eq:costTFIM} versus the number of qubits $n$ for each ansatz. The dashed (dotted) lines indicate the best polynomial (exponential) fit. The plot is shown in a log-linear scale. c) Variance of the cost function partial derivative versus $1/\dim(\liea)$. The plot is shown in a log-log scale. }
\label{fig:Ising}
\end{figure}

Similarly to what happened in the $XXZ$ case, we can turn the TFIM model subspace controllable upon the introduction of an extra generator. Consider the set
\begin{equation}\label{eq:genLTFIM}
    \GC_{\rm LTFIM} = \GC_{\rm TFIM}\bigcup\left\{\sum_{i=1}^{n} Z_i\right\},
\end{equation}
leading to the PSA in Figure~\ref{fig:Ising}(a, right). The set $\GC_{\rm LTFIM}$ can also be regarded as being constituted by the individual terms in the one-dimensional Longitudinal and Transverse Field Ising Model (LTFIM) Hamiltonian
\begin{equation}\label{eq:LTFIM}
    H_{\rm LTFIM}=\sum_{i=1}^{n_f} Z_iZ_{i+1}+h_x\sum_{i=1}^{n} X_i+h_z\sum_{i=1}^{n} Z_i\,.
\end{equation}
In addition, the ansatz generated by $\GC_{\rm LTFIM} $ is also a QAOA-type ansatz where an additional mixer has been added. 

In the case of open boundary conditions, the $\sum_{i=1}^{n} Z_i$ term breaks the $\Zbb_2$ symmetry, and thus the set $\GC_{\rm LTFIM}$ only conserves the parity symmetry, $\liea_{\rm LTFIM} = \bigoplus_{\substack{\sigma}}\liea_{\sigma}$. Using Algorithm~\ref{alg:lie} we find that the DLA is full rank on both $\sigma=\pm1$ parity subspaces, and hence
\begin{equation}
    \dim(\liea_{\sigma})\in\OC(2^{2n})\,.
\end{equation}
Similarly, in the closed boundary condition case, one can also find that the dimension of the DLA grows exponentially with $n$. This is an example where we show how a simple modification to the ansatz (adding a layer generated by $\sum_{i=1}^{n} Z_i$)  can greatly change the dimension of the DLA, and, as discussed below, such a small change can greatly affect the trainability of the cost function. 

In Figure~\ref{fig:Ising}(b) we show results for numerically computing $\Var_{\thv}[\partial_{\mu} C(\thv)]$ for the cost function
\begin{equation}\label{eq:costTFIM}
    C(\thv)=\bra{+}^{\otimes n}{U\ad(\thv)H_{\TFIM} U(\thv)}\ket{+}^{\otimes n}/n\,,
\end{equation}
where $U(\thv)$ is the PSA generated by the set  $\GC_{\TFIM}$ of Eq.~\eqref{eq:genTFIM} with $L=12n$ layers for open boundary conditions, and $L=6n$ for closed boundary conditions. For each value of $n=4,6,\ldots,18$ we computed the variance by picking $4400$ random sets of parameters, while for $n=20$ we picked $1000$ random intializations. In all cases the partial derivative was taken with respect to $\theta_{\frac{L}{2},2}$.  We see from Figure~\ref{fig:Ising}(b) that the variance of the cost partial derivative vanishes polynomially with $n$ for both open and closed boundary conditions, and hence the system does not exhibit a barren plateau. Then, as shown in Figure~\ref{fig:Ising}(c), once again, Conjecture~\ref{conjecture} holds for both open and closed boundary conditions: $\Var_{\thv}[\partial_{\mu} C(\thv)]$ and $\dim(\liea_{\TFIM})$ respectively vanish, and grow, polynomially with $n$. 

Moreover, in Figure~\ref{fig:Ising}(b) we also depict results obtained by computing $\Var_{\thv}[\partial_{\mu} C(\thv)]$ for the  LTFIM ansatz, using the same cost function of Eq.~\eqref{eq:costTFIM}. Now,  $U(\thv)$ is the PSA generated by the set  $\GC_{\LTFIM}$ in Eq.~\eqref{eq:genTFIM} with $L=6n$ layers. Using the same number of samples than for the TFIM case, we find that $\Var_{\thv}[\partial_{\mu} C(\thv)]$  vanishes exponentially with $n$ for both open and closed boundary conditions, and hence the cost exhibits a barren plateau. We see in Figure~\ref{fig:Ising}(c) that  Conjecture~\ref{conjecture} also holds for the LTFIM ansatzes with open and closed boundary conditions, as this time $\Var_{\thv}[\partial_{\mu} C(\thv)]$ vanishes exponentially with $n$, while $\dim(\liea_{\TFIM})$ grows exponentially with $n$. 

It is worth noting that, as discussed before, and  as shown in Figure~\ref{fig:Ising}(a), the difference between the TFIM and the LTFIM ansatz is given by an additional unitary in each layer (parametrized by a single angle). However, despite this simple difference, we find the variance of the cost function have different scaling, as one cost exhibits a  barren plateau while the other one does not exhibit a barren plateau. 

\subsubsection{ Erd\"{o}s–R\'{e}nyi model}

\begin{figure}[t]
\centering
\includegraphics[width=1\columnwidth]{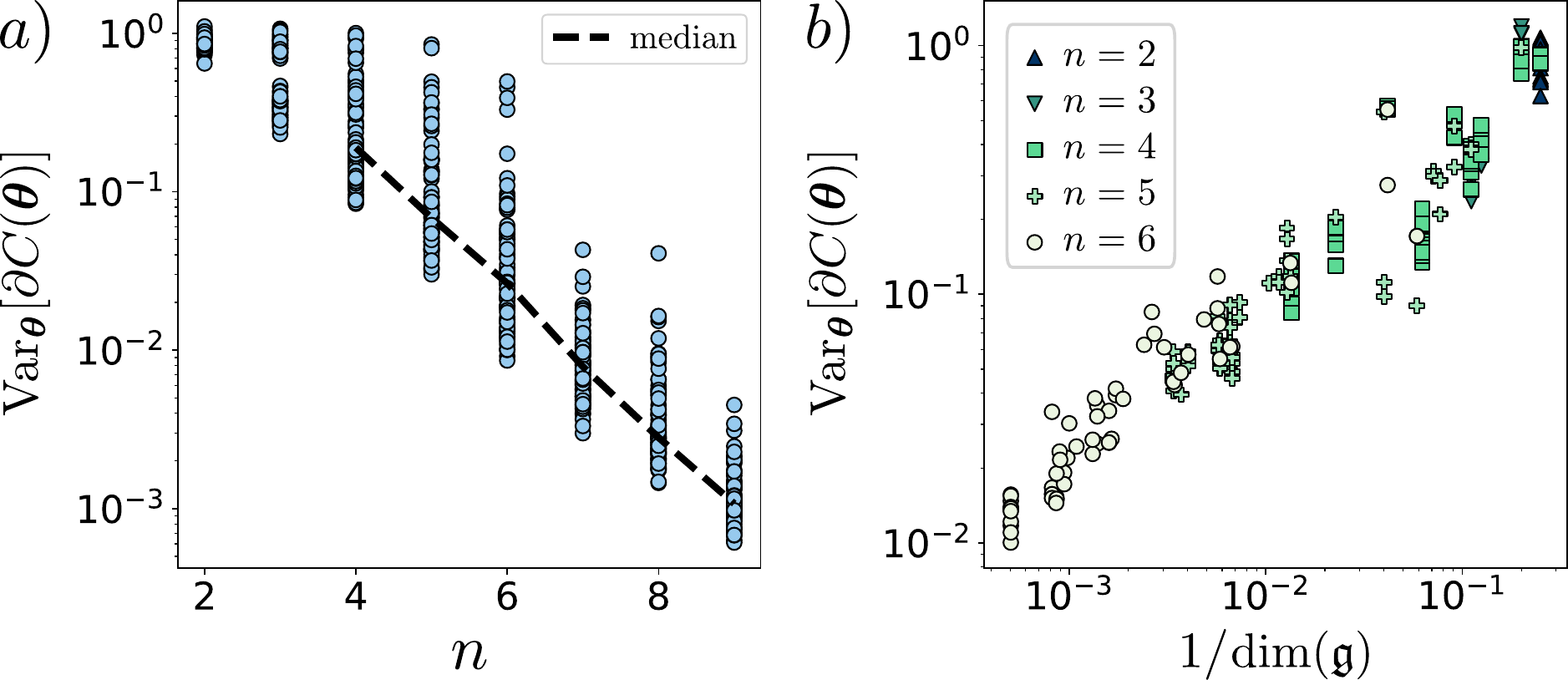}
\caption{\textbf{Numerical results for the  Erd\"{o}s–R\'{e}nyi model.} a) Variance of the cost function partial derivative of the cost function in Eq.~\eqref{eq:costErdos} versus the number of qubits $n$. The dashed  line indicates the medians across graphs computed for each value of $n$. The plot is shown in a log-linear scale. b) Variance of the cost function partial derivative versus $1/\dim(\liea)$. The plot is shown in a log-log scale. }
\label{fig:Erdos}
\end{figure}

Let us now consider the task of solving MAXCUT problems with a QAOA ansatz. Here, we recall that MAXCUT is specified by a graph $G=(V,E)$ of nodes $V$ and edges $E$, such that one seeks to determine a partition of the nodes of $G$ into two sets that maximize the number of edges connecting nodes between sets. The MAXCUT  Hamiltonian is given by

\begin{equation}
    H_{\rm ER}=-\frac{1}{2}\sum_{ij \in E}(\id-Z_i Z_j)\,,
\end{equation}
and we consider the standard QAOA ansatz generated by
\begin{equation}\label{eq:genErdos}
    \GC = \left\{\sum_{i=1}^{n} X_i,H_{\rm ER}\right\}.
\end{equation}
Let us analyze the variance of the partial derivative of the cost 
\begin{equation}\label{eq:costErdos}
    C(\thv)=\bra{+}^{\otimes n}{U\ad(\thv)H_{\rm ER} U(\thv)}\ket{+}^{\otimes n}/|E|\,,
\end{equation}
where we use $|E|$ (the number of edges in the graph) to normalize the cost function. For each value of $n=2,3,\ldots,9$ we generated 90 graphs according  to the Erd\"{o}s–R\'{e}nyi model~\cite{erdds1959random}. That is, each graph $G$ was chosen uniformly at random from the set of all graphs of $n$ nodes. Then, for each graph we sampled $3000$ random initializations with $L=12 n$ layers and we took the partial derivative with respect to the angle in the $L/2$-th layer associated to the (mixer) Hamiltonian $\sum_{i=1}^{n} X_i$. In Fig.~\ref{fig:Erdos}(a) we show results of  $\Var_{\thv}[\partial_{\mu} C(\thv)]$ versus the number of qubits. Here we can see that, as expected, even for fixed $n$ different graphs will have different value of the variance. However, by computing the median variance for each system size we found that the scaling of the median is exponentially decaying with the system size. While this result does not preclude the possibility of  generating graphs that will not have a barren plateau, it suggests that  uniform sampling of graphs from the Erd\"{o}s–R\'{e}nyi model will lead to the landscape for a typical graph having a barren plateau. 
Then, as shown in Fig.~\ref{fig:Erdos}(b), we compute the dimension of the DLA for each graph, and find that Conjecture~\ref{conjecture} is confirmed, as the relation between $\Var_{\thv}[\partial_{\mu} C(\thv)]$  and $\dim(\liea)$ is linear in a log-log-scale.

\section{DISCUSSION}\label{sec:discussions}

In this work, we have explored a fundamental connection between VQAs and the theory of QOC with the purpose of analyzing the existence of barren plateaus in a family of periodic-structured ansatzes which contain, as special cases, the QAOA and the HVA, among other widely used ansatzes in variational quantum algorithms and quantum machine learning. Our results show that one can diagnose the presence of barren plateaus in the cost function landscape by analyzing the degree of controllability of the system, characterized by the dimension of the dynamical Lie algebra (DLA) obtained from the set of generators of the ansatz.

Our main results are the following. First, we show that if the DLA is full rank, i.e. if  the system is controllable, then the cost function exhibits a barren plateau. This follows from the fact that, as we show, controllable systems converge to $2$-designs. Here, we also derive an expression relating the depth required for a given ansatz to become an $\varepsilon$-approximate two-design with the expressibility of one of its layers.

We then consider systems with symmetries, where the Hilbert space partitions into invariant subspaces associated with the different eigenspaces. In this context, we show that when the system is subspace controllable, the existence of barren plateaus crucially depends on the input state to the VQA. For example, the cost might be trainable for certain input states, but might exhibit a barren plateau for others. Specifically, our results connect the scaling of the variance of cost function partial derivatives to that of the dimension of the subspace in which the input state has support on. Instead, when the system is subspace uncontrollable, we show that one can still upper bound the variance of the cost function partial derivative using the expressibility of the ansatz in the relevant subspace. This indicates that larger subspace expressibilities leads to smaller gradients.

Finally, we present an conjecture that shows that one can directly study the scaling of the cost function partial derivative variance by computing the dimension of the subspace DLA to which the input state belongs. This conjecture implies that ansatzes with polynomially growing DLAs can exhibit polynomially vanishing gradients, while ansatzes with exponentially growing DLAs should exhibit exponentially vanishing gradients.

In addition, we performed numerical simulations of VQAs with the hardware efficient ansatz, QAOA, and HVA, for problems such as preparing ground states of the $XXZ$ model and of the Ising model, or solving MAXCUT problems on graphs generated from the Erd\"{o}s–R\'{e}nyi model. The numerical results match our theoretical predictions and hence verify our analytical results for controllable and subspace controllable systems. Moreover, in all cases considered we verify that our conjecture holds, further providing evidence that the scaling of the cost function partial derivative variance may be directly linked to the dimension of the subspace DLA. 

\subsubsection*{Implications of our results to ansatz design}

The broader implication of our results is that the framework introduced here can be used to design ansatzes, as one could potentially predict if an ansatz, or a modification to the ansatz, will lead to the cost function exhibiting a barren plateau. Hence, our work can be considered as paving the way towards trainability-aware ansatz design.

For instance, we have shown how a simple change in the ansatz structure, such as adding an additional parametrized unitary per layer, can greatly affect the gradient scaling of the cost by changing the controllability of the system. This means that one should be careful when employing schemes such as the Adaptive QAOA or quantum optimal control ansatz as the addition of an operator $H$ to the set  $\GC$ of generators of the ansatz can lead to barren plateaus if the system becomes controllable (or subspace controllable in an exponentially growing subspace). In particular, if $H$ does not commute with the elements in $\GC$, one should analyze how the DLA changes by such addition before proceeding to change the ansatz.

Here, we crucially remark that one of the main advantages of the aforementioned theoretical analysis is that it can be performed classically (either analytically or numerically) as it just requires the evaluation of the DLA. Hence, our methods save precious quantum resources as one does not need to run the quantum algorithm, or even access a quantum computer, to test the trainability of the ansatzes.

Finally, we remark that if our conjecture holds more generally, then one can use this additional tool to directly study the trainability of an ansatz by estimating the scaling of the variance of the cost function partial derivative through the scaling of the dimension of the DLA. For example, such results can be used to show that certain ansatzes might not have exponentially vanishing gradients. For instance,  when considering a QAOA ansatz for solving MAXCUT on $2$-regular graphs, a straightforward computation of the DLA reveals its scaling is only linear in $n$. Hence, we expect (and we find) no barren plateaus. Similarly, one can use our conjecture to analyze ansatz proposals in the literature. For example, Ref.~\cite{lee2021towards} recently proposed an ansatz generated by the set of products up to $K$-body Pauli $X$ operators, i.e., $\GC=\{ X_i\}_i \cup \{ X_iX_{j}\}_{i>j}\cup \{X_iX_jX_k\}_{i>j>k}\cup \cdots$. Since the ansatz is abelian, the dimension of DLA is just the number of generators. Thus, we expect that when using a poly number of layers the ansatz should be rid barren plateaus.

\subsubsection*{Outlook}

In the present work, we have established a novel framework for diagnosing the presence of barren plateaus in VQAs. While here we mainly focus on the trainability of ansatzes for near-term quantum computing, our results  should also be considered as useful in the broader context of QOC. For instance, while the barren plateau phenomenon has been recently widely studied in VQAs, it is clear from our manuscript that barren plateaus can (and will) also arise in QOC schemes (see also~\cite{arenz2020drawing}). Hence, we leave for future work to study how some of the results derived for the trainability of VQAs be used to analyze the trainability of QOC control pulses. 

In addition, we note that since our work studies the trainability of certain families of ansatzes, we also leave for future work to show how the tools here presented can be employed to study more general ansatzes (e.g., ansatzes for quantum machine learning applications) which do not necessarily have a periodic structure. In addition, we leave as an open question how the results in our conjecture can be generalized and formally proved.

\section{ACKNOWLEDGMENTS}

We thank Marco Farinati and Robert Zeier for useful discussions on Lie-algebras, and we also thank Zoe Holmes and Pablo Poggi for helpful discussions. ML acknowledges partial support by CONICET (PIP 112201 50100493CO), UBACyT (20020130100406BA) and ANPCyT (PICT-2016-1056)).  ML was also supported by the U.S. Department of Energy (DOE), Office of Science, Office of Advanced Scientific Computing Research, under the Quantum Computing Applications Team (QCAT) program and also under the Accelerated Research in Quantum Computing (ARQC) program. Piotr C. was supported by the Laboratory Directed Research and Development (LDRD) program of Los Alamos National Laboratory (LANL) under project numbers 20190659PRD4. KS acknowledges support from the National Science Foundation under Grant No. 2014010. GM was supported by the US Department of Energy, Office of Advanced Scientific Computing Research, under the ARQC program and by the LANL ASC Beyond Moore's Law project. PJC acknowledges initial support from the LANL ASC Beyond Moore's Law project. MC acknowledges initial support from the Center for Nonlinear Studies at LANL. KS, PJC, and MC also acknowledge support from LDRD program of LANL under project number 20190065DR.

\bibliographystyle{unsrtnat}
\bibliography{quantum}

\cleardoublepage

\onecolumngrid

\appendix

\begin{center}
	{\Large \bf Appendices} 
\end{center}

\setcounter{section}{0}
\setcounter{proposition}{0}
\setcounter{figure}{0}
\setcounter{corollary}{0}
\setcounter{theorem}{0}
\setcounter{definition}{0}

\bigskip

In the following appendices we present additional information and derive proofs for the main results in the manuscript. In Appendix~\ref{sec:app:preliminaries} we introduce preliminary notation and definitions that will be  relevant for the rest of the appendices. Then, in Appendix~\ref{sec:app:ansatzes} we provide additional details on different widely known ansatzes that are Periodic Structure Ansatz (PSA). In Appendix~\ref{sec:app:barren plateaus} we provide a brief review of barren plateaus. Appendices~\ref{app:proof_theo_convergence}--\ref{app:proof_su2} contain the proofs of our main Theorems, Corollaries and Propositions. Finally, in Appendix~\ref{App:XXZ} we discuss the symmetries in the $XXZ$ and Ising spin models considered in the main text, and in Appendix~\ref{App:numerical} we provide additional details on the initial state used for the numerical simulations of the $XXZ$ model.

\section{Preliminaries}\label{sec:app:preliminaries}
Let us first review some definitions and prior results that will be relevant for the rest of the appendices.

\medskip

 \textbf{Properties of the Haar measure.} Let $\mathcal{U}(d)$ denote the group of $d\times d$ unitary matrices. Let $d\mu_H(V) = d\mu(V)$ be the volume element of the Haar measure, where $V\in \mathcal{U}(d)$. Then, the Haar measure has the following properties: (1) The volume of the Haar measure is finite: $\int_{\mathcal{U}(d)} d\mu(V) < \infty$. (2) The Haar measure is uniquely defined up to a multiplicative constant factor. (3) Let $d\zeta(V)$ be an invariant measure. Then there exists a constant $c$ such that $d\zeta(V) = c\cdot d\mu(V)$. (4) The Haar measure is left- and right-invariant under the action of the unitary group of degree $d$, i.e., for any integrable function  $g(V)$, the following holds: 
 \begin{align}
\int_{\mathcal{U}(d)} d\mu(V) g(W V) = \int_{\mathcal{U}(d)}d\mu(V) g(VW) = \int_{\mathcal{U}(d)} d\mu(V) g(V), 
\end{align}
where $W\in \mathcal{U}(d)$.

\medskip

\textbf{Symbolic integration.}
Let us present formulas that allow for the symbolical integration with respect to the Haar measure on a unitary group~\cite{puchala2017symbolic}. For any $V\in \mathcal{U}(d)$ the following expressions are valid for the first two moments:  
\small
\begin{equation}\label{eq:delta}
\begin{aligned}
    \int_{\mathcal{U}(d)} d\mu(U)u_{\vec{i}\vec{j}}u_{\vec{p}\vec{k}}^*&=\frac{\delta_{\vec{i}\vec{p}}\delta_{\vec{j}\vec{k}}}{d}\,,   \\
\int_{\mathcal{U}(d)} d\mu(U)u_{\vec{i}_1\vec{j}_1}u_{\vec{i}_2\vec{j}_2}u_{\vec{i}_1'\vec{j}_1'}^{*}u_{\vec{i}_2'\vec{j}_2'}^{*}&=\frac{\delta_{\vec{i}_1\vec{i}_1'}\delta_{\vec{i}_2\vec{i}_2'}\delta_{\vec{j}_1\vec{j}_1'}\delta_{\vec{j}_2\vec{j}_2'}+\delta_{\vec{i}_1\vec{i}_2'}\delta_{\vec{i}_2\vec{i}_1'}\delta_{\vec{j}_1\vec{j}_2'}\delta_{\vec{j}_2\vec{j}_1'}}{d^2-1}
-\frac{\delta_{\vec{i}_1\vec{i}_1'}\delta_{\vec{i}_2\vec{i}_2'}\delta_{\vec{j}_1\vec{j}_2'}\delta_{\vec{j}_2\vec{j}_1'}+\delta_{\vec{i}_1\vec{i}_2'}\delta_{\vec{i}_2\vec{i}_1'}\delta_{\vec{j}_1\vec{j}_1'}\delta_{\vec{j}_2\vec{j}_2'}}{d(d^2-1)}\,,
\end{aligned}
\end{equation}
\normalsize
where $u_{\vec{i}\vec{j}}$ are the matrix elements of $U$. Assuming $d=2^n$, we use the notation $\vec{i} = (i_1, \dots i_n)$ to denote a bitstring of length $n$ such that $i_1,i_2,\dotsc,i_{n}\in\{0,1\}$. 

\medskip

\textbf{Useful Identities.} 
We introduce the following identities, which can be derived using Eq.~\eqref{eq:delta} (see~\cite{cerezo2020cost} for a review):
\footnotesize
\begin{align}
    \int_{\UC(d)} d \mu(U) \operatorname{Tr}\left[U A U^{\dagger} B\right]&=\frac{\operatorname{Tr}[A] \operatorname{Tr}[B]}{d},\label{Eq_lemma1}\\
    \int_{\UC(d)} d \mu(U) \operatorname{Tr}\left[U A U^{\dagger} B U C U^{\dagger} D\right] &=\frac{\operatorname{Tr}[A] \operatorname{Tr}[C] \operatorname{Tr}[B D]+\operatorname{Tr}[A C] \operatorname{Tr}[B] \operatorname{Tr}[D]}{d^{2}-1}  -\frac{\operatorname{Tr}[A C] \operatorname{Tr}[B D]+\operatorname{Tr}[A] \operatorname{Tr}[B] \operatorname{Tr}[C] \operatorname{Tr}[D]}{d\left(d^{2}-1\right)},\label{Eq_lemma2}\\
        \int_{\UC(d)} d \mu(U) \operatorname{Tr}\left[U A U^{\dagger} B\right] \operatorname{Tr}\left[U C U^{\dagger} D\right]& = \frac{\operatorname{Tr}[A] \operatorname{Tr}[B] \operatorname{Tr}[C] \operatorname{Tr}[D]+\operatorname{Tr}[A C] +\operatorname{Tr}[B D]}{d^{2}-1}-\frac{\operatorname{\Tr}[A C] \operatorname{Tr}[B] \operatorname{Tr}[D]+\operatorname{Tr}[A] \operatorname{Tr}[C] \operatorname{Tr}[B D]}{d\left(d^{2}-1\right)},\label{Eq_lemma3}
\end{align}
\normalsize
where $A, B, C$, and $D$ are linear operators on a $d$-dimensional Hilbert space. 

\medskip

\textbf{Integration over parameter space}: In the next sections we will derive analytical expressions for the variance of the  partial derivatives of cost functions $C(\thv)$ over parametrized circuits $U(\thv)$. In such derivations, we will have to deal with integration over the parameter space. A key step in the following analysis will be to relate the integration over parameters with integration over the ensemble of unitaries arising from different parameter choices. In this sense, we recall that given a set of parameters $\{\thv\}$ one can obtain an associated set of unitaries generated by the quantum circuit $\{U(\thv)\}$. Then,  consider the integration of some function $f(U(\thv))$ over $\thv$. Defining $\Ubb$ as the distribution of unitaries generated by $U(\thv)$, the following identity holds
\begin{equation}\label{eq:paramstounit}
    \int_{\thv} d\thv f(U(\thv)) = \int_{\Ubb} dU f(U)\,.
\end{equation}

In addition, if the distribution of unitaries $\Ubb$ can be shown to converge to a $2$-design, the integration over the distibution can be further converted into an integration over the Haar measure
\begin{equation}
    \int_{\Ubb} dU \xrightarrow[]{} \int_{\UC(d)} d\mu(U)
\label{eq:tohaar}
\end{equation}
allowing the use of identities \eqref{Eq_lemma1}, \eqref{Eq_lemma2}, and~\eqref{Eq_lemma3}.

\section{Ansatzes}\label{sec:app:ansatzes}
In general, ansatzes for parametrized quantum circuits can be divided into two primary categories: problem-agnostic and problem-inspired ansatzes. In a problem agnostic ansatz one does not have any information about the problem, or its solution, that one can encode in the ansatz. Such is the case for instance in a task of estimating the spectrum of an unknown density operator~\cite{cerezo2020variational}.  On the other hand, problem-inspired ansatzes employ prior information about a given problem or task. For example, for the problem of estimating the ground state energy of a particular Hamiltonian, one can design ansatzes that preserve the symmetry of the problem Hamiltonian~\cite{gard2020efficient}. 

Here we remark that several well known problem-agnostic and problem-inspired ansatzes in the literature are PSAs of the form in Eq.~\eqref{eq:ansatz}. In particular, our framework allows us to study the hardware-efficient ansatz (HEA) \cite{kandala2017hardware}, quantum alternating operator ansatz (QAOA) \cite{farhi2014quantum}, Adaptive QAOA \cite{hadfield2019quantum,zhu2020adaptive}, Hamiltonian variational ansatz (HVA) \cite{wecker2015progress}, and Quantum Optimal Control Ansatz (QAOC) \cite{choquette2020quantum}.

 Below we provide several examples of problem-inspired and problem-agnostic PSA and highlight their advantages in different problems.

\medskip

\textbf{Hardware efficient ansatz.} The Hardware Efficient Ansatz (HEA) is a problem-agnostic ansatz, which relies on gates native to a quantum hardware. In particular, an ansatz can be designed based on a gate alphabet, which depends on the architecture and the connectivity of a given quantum hardware. This procedure helps in avoiding the overhead associated with transpiling an arbitrary unitary into a sequence of native gates. For example, one can consider native gates, such as single qubit rotations $e^{-i\theta/2 Z}e^{-i\gamma/2 Y}$ and CNOTs, where $Y$ and $Z$ denote Pauli matrices, and a CNOT between the control qubit $i$ and the target qubit $j$ is given by: $e^{-i\pi/2 (|1\rangle \langle 1|_i\otimes(X_j-\id_j) )}$. Then an ansatz of the form in \eqref{eq:ansatz} can be generated as follows: one layer consists of parametrized single qubit rotations on each qubit, followed by unparametrized CNOTs acting on neighboring qubits. 

The HEA has been employed to prepare the ground state of molecules \cite{kandala2017hardware}, to study  Hamiltonians that are similar to the device's interactions \cite{kokail2019self}, and in several other variational quantum algorithms \cite{khatri2019quantum,sharma2019noise,tkachenko2020correlation,cerezo2020cost}. The HEA is also suitable in the near-term implementations of VQAs due to its low-depth structure which results into a lower-noise circuit in comparison to other ansatze~\cite{tkachenko2020correlation,wang2020noise}.

\medskip

\textbf{Quantum alternating operator ansatz.} The Quantum Alternating Operator Ansatz (QAOA) is a problem-inspired ansatz that simulates the discretized adiabatic transformations \cite{farhi2014quantum}. Consider a goal of preparing the ground state of a problem Hamiltonian $H_P$. Let $H_M$ denote a mixer Hamiltonian, with corresponding ground state $|\psi\rangle$. Then the QAOA maps $\ket{\psi}$ to the ground state of $H_P$ by sequentially applying the problem unitary $e^{-i \gamma_l H_P}$, followed by the mixer unitary $e^{-i \beta_l H_M}$. Let $\vec{\theta} = (\vec{\gamma}, \vec{\beta})$. Then the QAOA is given by $U(\vec{\theta}) = \prod_{l=1}^{L} e^{-i \beta_{l} H_{M}} e^{-i \gamma_{l} H_{P}}$, which follows the general form of the ansatz defined in \eqref{eq:ansatz}. Here, $p$ is the order of the discretized adiabatic transformation and it determines the precision of the solution \cite{farhi2014quantum}. The QAOA was originally introduced for finding approximate solutions to combinatorial optimization problems \cite{farhi2014quantum}.  The QAOA has been generalized as a standalone ansatz \cite{hadfield2019quantum} and its performance has been investigated in several tasks, including the task of learning a unitary \cite{kiani2020learning}. Moreover, the QAOA has been shown to be computationally universal \cite{lloyd2018quantum,morales2020universality}, and the choice of optimal mixer is still an open debate~\cite{wang2020x,bartschi2020grover}. 

\medskip 

\textbf{Adaptive QAOA.} As a consequence of the adiabatic theorem, the QAOA should lead to good solutions for high values of $p$ \cite{farhi2014quantum}. However, for small values of $p$, the QAOA is an \textit{ad-hoc} ansatz, which is not necessarily an optimal strategy to approximate the ground state of the problem Hamiltonian. A way to improve such an \textit{ad-hoc} ansatz is to employ a variable mixer instead of a fixed mixer at each layer \cite{hadfield2019quantum}. Let $\{G_k\}_{k=1}^q$ denote a set of mixer Hamiltonians. Then an adaptive QAOA can be defined as follows: 
$U(\vec{\theta}) = \prod_{l=1}^L e^{-i \beta_l G_l}e^{-i \gamma_l H_p}$, where each $G_l$ can be adaptively picked from $\{G_k\}_{k=1}^q$. 

One particular adaptive approach was introduced in \cite{zhu2020adaptive}, where at each layer, $G_l$ is picked based on the largest gradient of the cost function among all $\{G_l\}$. Moreover, \cite{zhu2020adaptive} observed that adaptive entangling mixers can improve performance and reduce the number of parameters and CNOTs to achieve a desired accuracy in comparison to the non-adaptive QAOA. We note that the adaptive QAOA follows the form in \eqref{eq:ansatz} if adaptive mixers are learned up to a fixed layer and then the whole structure is repeated. That is, $U(\vec{\theta}) = \prod_{m=1}^p \prod_{l=1}^r e^{-i \beta_{l,m}G_{l,m}}e^{-i \gamma_{l,m} H_P}$, where $G_{l,1}$ are learned adaptively for each $l \in \{1, \dots, r\}$, and $G_{l,1} = G_{l,m}$ for all $m \in \{1, \dots, m\}$.

\medskip 

\textbf{Hamiltonian variational ansatz.} The Hamiltonian variational ansatz is another problem-inspired ansatz, which implements time evolution under problem Hamiltonian via Trotterization \cite{wecker2015progress}. It can be understood as a generalization of the QAOA to more than two non-commuting Hamiltonians. Let $H_P = \sum_{l} H_l$ denote a problem Hamiltonian, such that $[H_l, H_{l^{'}}]\neq 0$. Then the HVA of order $p$ is given by $U(\vec{\theta}) = \prod_{k=1}^p \prod_{l} e^{-i \theta_{l,k} H_l}$, which is in the form of \eqref{eq:ansatz}. The HVA has been investigated in studying one- and  two-dimensional  quantum  many-body models \cite{ho2019efficient,cade2020strategies}. 

A simple example where the HVA can be employed is the XXZ model to study magnetism. For a one-dimensional chain, the Hamiltonian for the XXZ model is given by $H_{XXZ} = - \sum_{l=1}^n X_l X_{l+1} +  Y_lY_{l+1} + g Z_lZ_{l+1}$, where $g$ determines the phase of magnetisation. Let $H_A = \sum_{l=1}^n A_l A_{l+1}$, where $A \in \{X, Y, Z\}$.  Then, one way to parametrize a HVA of order $p$ is as follows: $U(\vec{\theta}) = \prod_{l=1}^p e^{- i\beta_l H_X} e^{-i\gamma_l H_Y} e^{-i\delta_l H_Z}$ for $g=1$, and where $\vec{\theta} = (\vec{\beta}, \vec{\gamma}, \vec{\delta})$.
Another way to parametrize a HVA is as follows: 
$U(\vec{\theta}) = \prod_{l=1}^p e^{- i\beta_l (H_X+H_Y)} e^{-i\delta_l H_Z}$, where we redefined $\vec{\theta}=(\vec{\beta}, \vec{\delta})$.

\medskip 

\textbf{Quantum optimal control ansatz.} The HVA discussed above helps constraining the variational search to a relevant symmetric subspace of the the total Hilbert space. In general, this approach might require high values of $p$ to achieve a desired accuracy in approximating the ground state of many-body Hamiltonians. One way to avoid high values of $p$ is to introduce drive terms in addition to the problem Hamiltonian, which break the symmetry of the problem Hamiltonian $H_p$. This approach falls under the framework of quantum optimal control \cite{choquette2020quantum}. In particular, let $\{\hat{H}_k\}$ denote a set of drive terms. Then the update time-dependent Hamiltonian is given by $\bar{H} (t)= H_P + \sum_k c_k(t) {H}_k$, where drive terms ${H}_k$ are picked such that $[H_P, {H}_k]\neq 0$ for all $k$. Here, $c_k(t)$ are time-dependent control parameters. Let $H_P = \sum_q H_q$ and let $\vec{\theta} = (\vec{\gamma}, \vec{\beta})$. Then the Quantum Optimal Control Ansatz  (QOCA) of order $L$ is given by $U(\vec{\theta}) = \prod_{l=1}^L \prod_q e^{-i \beta_{l,q} H_q} \prod_k e^{-i \gamma_{l,k} {H}_k}$, where $\gamma_{l,k}$ denote the discrete drive amplitudes of the control parameter $c_k(t)$. 

In general, finding an optimal drive Hamiltonian terms $\{{H}_k\}$ is a computationally challenging problem. One can employ an adaptive approach to pick drive Hamiltonians from a fixed set of Hamiltonian, similar to the adaptive QAOA \cite{zhu2020adaptive}. In \cite{choquette2020quantum}, the QOCA was shown to outperform other ansatze, including the HEA and the HVA for the task of preparation of the ground state of the half-filled Fermi Hubbard model.

\section{Barren Plateaus}\label{sec:app:barren plateaus}

As mentioned in the main text, the barren plateau phenomenon  has been recognized as one of the most important challenges to overcome to guarantee the success of VQAs. When a cost function exhibits a barren plateau, its gradients are exponentially suppressed (in average) across the optimization landscape. Consider the following mathematical definition.
\begin{definition}[Barren Plateau]\label{def:BPSM}
A cost function $C(\thv)$ as in Eq.~\eqref{eq:cost} is said to have a barren plateau when training $\theta_\mu\equiv\theta_{pq}\in\thv$,  if the cost function partial derivative $\partial C(\thv)/\partial \theta_{\mu}\equiv \partial_{\mu} C(\thv)$ is such that
\begin{equation}\label{eq:BPSM}
    \Var_{\thv}[\partial_{\mu} C(\thv)]\leq F(n)\,, \quad \text{with} \quad F(n)\in \OC\left(\frac{1}{b^n}\right)\,,
\end{equation}
for some $b>1$. Here the variance is taken with respect to the set of parameters $\thv$.
\end{definition}
From Chebyshev inequality we know that $\Var_{\thv}[\partial_{\mu} C(\thv)]$ bounds the probability that $\partial_{\mu} C(\thv)$ diverges from its average (of zero) as $P(|(\partial_{\mu} C(\thv)|>c)\leq \Var_{\thv}[\partial_{\mu} C(\thv)]/c^2$ for any $c>0$. 

Equation~\eqref{eq:BPSM} implies that one requires a precision (i.e., a number of shots) that grows exponentially with $n$  to navigate trough the flat landscape and determine a cost minimizing direction when optimizing the cost function. Moreover, as shown in~\cite{cerezo2020impact,arrasmith2020effect}, barren plateaus affect both gradient-based and gradient-free methods meaning that simply changing the optimization strategy does not mitigate or solve the barren plateau issues. Since the goal of VQAs is to have computational complexities that scale polynomial with $n$, such exponential scaling in the required precision  destroys the hope of achieving a computational advantage with the VQA over classical methods (which usually scale exponentially with $n$). 

The first result for barren plateaus was obtained in~\cite{mcclean2018barren}, where it was shown that deep unstructured ansatz that form $2$-designs have barren plateaus. This phenomenon  was then generalized to layered Hardware Efficient Ansatzes in~\cite{cerezo2020cost} were it was proven that the locality of the cost function is connected to the existence of barren plateaus. That is, global cost functions (i.e., cost functions where $O$ in~\eqref{eq:cost} acts non-trivially in all qubits) exhibit barren plateaus even for shallow depths, whereas local cost functions (i.e., cost functions where $O$ in~\eqref{eq:cost} acts non-trivially in a small number of neighboring qubits) do not exhibit barren plateaus for short-depth ansatzes.

The barren plateaus phenomenon  has also been studied in the context of  quantum neural networks~\cite{sharma2020trainability,pesah2020absence,zhao2021analyzing}, and to the problem of learning scramblers~\cite{holmes2020barren}. In addition, it has been shown that circuits that generate large amounts of entanglement~\cite{sharma2020trainability,patti2020entanglement,marrero2020entanglement} are prone to suffer from barren plateaus. To circumvent or mitigate the effect of barren plateaus, several strategies have been developed~\cite{verdon2019learning,volkoff2021large,skolik2020layerwise,grant2019initialization,pesah2020absence,zhang2020toward,bharti2020quantum,cerezo2020variational,sauvage2021flip,liao2021quantum}. 

\section{Quantum Optimal Control}\label{app:qoc}

Here we recall for convenience that in a standard QOC setting one is interested in controlling the dynamical evolution of a quantum state $\kp$ in a$d$-dimensional Hilbert space $\H=\Cbb^d$ (where $d=2^n$)~\cite{dalessandro2010introduction}. Here, the system dynamics are determined by a Hamiltonian
\begin{equation}
    H(\{f_k(t)\}) = H_0+\sum_{k=1}^K f_k(t)H_k
\label{eq:QOC_Hamiltonian2}
\end{equation}
that is tunable through some time-dependent control fields functions $\{f_k(t)\}$.

At its core, the problem in QOC is to determine how to shape the control fields such that the system evolves in a desired manner. A specific set of optimal fields is usually constructed by imposing a parametrization on the functions and applying standard numerical optimization routines. The success of such optimization process depends on the structure of the underling optimization spaces, the so-called quantum control landscapes \cite{chakrabarti2007quantum,larocca2018quantum,larocca2020exploiting}.

For instance, a common choice is to consider piece-wise constant fields where the protocol duration $T$ is divided in $L$ intervals $\Delta t_j=t_j-t_{j-1}$ (such that $T=\sum_{j=1}^L \Delta t_j$) at each of which the fields take a constant value, e.g., $f_k(t) =f_{k,j}$ if $t_{j-1}<t<t_j$. In this case, the propagator factorizes into a product of individual sub-propagators, each of which is generated by a constant Hamiltonian and thus leads to the simple matrix exponential form
\begin{equation}
U(\{f_k(t)\}) \!=\! \prod_{l=1}^L e^{-i \tilde{H}_l \Delta t_l},\,\,\, \tilde{H}_l\!=\!H_0+\sum_{k=1}^K f_{l,k} H_k\,.
\label{eq:QOC_propagator}
\end{equation}
By Trotterizing Eq. \eqref{eq:QOC_propagator} one finds
\begin{equation}
U(\{f_k(t)\}) \approx \prod_{l=1}^L\prod_{k=0}^K e^{-i H_k \Delta t_l f_{l,k}}\,,
\label{eq:QOC_propagatorTrot}
\end{equation}
where $f_{l,0}=1$ for all $l$. Note that~\eqref{eq:QOC_propagatorTrot} is a PSA of the form of Eq.~\eqref{eq:ansatz}, through the identification 
\begin{equation}
        \theta_{l,0}=\Delta t_0\,, \quad \quad
        \theta_{l,k}=\Delta t_l f_{l,k} , \quad k\geq 1\,.
\label{eq:correspondence}
\end{equation}
We remark that in the limit $\Delta t_l \xrightarrow[]{} 0$, Eq.~\eqref{eq:QOC_propagatorTrot} becomes exact. In the general case, the exact and Trotterized ansatzes approximately coincide, and a nontrivial correction of Eq.~\eqref{eq:correspondence} is needed to make the correspondence exact. In any case, Eq.~\eqref{eq:correspondence} allows us to henceforth use the notation $U(\{f_k(t)\})= U(\thv)$ and indicate with $\thv$ the trainable parameters in a QOC setting.

\section{Dynamical Lie Algebra computation}\label{app:alg}

First, let us recall that a Lie algebra is a vector space $\mf{g}$ together with an operation $[\cdot,\cdot]:\mf{g}\rightarrow{}\mf{g}$ called \textit{Lie bracket} that is bilinear, \textit{alternating} (the output is zero if the inputs are linearly dependent) and satisfies the identity $[x,[y,z]]+[y,[z,x]] + [z,[x,y]] = 0$, known as the Jacobi identity. Lie algebras are vector spaces that are \textit{closed} under such Lie bracket, i.e. $[x,y]\in\mf{g}$ for all $x,y\in\mf{g}$. Note that this operation is not necessarily associative, i.e. $[x[y,z] \neq [[x,y],z]$. In fact, this is precisely what the Jacobi identity captures: how the order of evaluation affects the result of the operation.

In our quantum context, Lie algebras manifest as matrix Lie algebras. For example, the space of quantum observables $\mf{u}(d)$ is a subspace of the vector space of $d\times d$ complex matrices that is closed under matrix commutator (playing the role of a Lie bracket). 
More generally, we will encounter ourselves with Lie algebras that form Lie subalgebras of $\mf{u}(d)$. A subalgebra is a subspace of an algebra that is itself closed under the Lie bracket. For example, in a $n$ qubit quantum system, the subspace $\Omega=\spn\{ X,Y,Z \}$, where $X=\sum_i X_i$, $Y=\sum_i Y_i$ and $Z=\sum_i Z_i$ is closed under commutation and thus constitutes a 3-dimensional subalgebra of the $4^n$-dimensional operator space.

Specifically, we will be interested in the so-called \textit{dynamical Lie algebra} (DLA). This Lie algebra is the subspace of $\mf{u}(d)$ generated by the Lie closure of the generators of a paramterized quantum circuit (see Definition~\ref{def:dynamical_lie_algebra}). In general, computing the DLA is a highly nontrivial task. One possible approach to the DLA is direct construction, i.e. start with a set of generators defining a subspace of operator space (but not a subalgebra), and start commuting them, finding new elements until one obtains a basis of the DLA (see Algorithm~\ref{alg:lie}). The complexity of such approach is, in general, $\OC(\rm poly(d))$ with $d=2^n$, that is, exponential in $n$ the number of qubits. For example, a naive approach (representing operators as dense $d\times d$ matrices) yields roughly $\OC(d^2d^6)$,  since, in general, one has to check linear independence $\OC(d^2)$, for example by implementing LU or QR decompositions (whose cost is $\OC(N^3)$ for $N\times N$ matrices) on such $d\times d$ matrices. Although such complexity can be reduced, for example, using more intelligent representations of operators, in essence direct construction is attempting to build a basis for a subalgebra of $\mf{su}(d^2)$, i.e. a basis with potentially as many as $d^2$ elements, and therefore it cannot generally avoid exponentiality.

Despite being exponential, direct construction of the DLA (either numerically or analitically) on small system sizes can constitute a remarkably useful tool to later extrapolate or prove the scaling (e.g. by induction) of the DLA beyond those small 'afforable' system sizes. Moreover, note that in many cases one may only be interested in checking whether the dimension of the DLA is above a certain threshold, a task with complexity linear in the size of such threshold. Of course, this is neglecting the complexity of computing new DLA elements, which, as mentioned above, can be substantially diminished by choosing efficient representations for those  operators.

\begin{algorithm}[t]
    \DontPrintSemicolon
    \KwIn{Set of generators $\GC$ of the ansatz.}
    \KwOut{Basis $S$ of the algebra $\liea$ obtained from $\GC$.}
     \kwInit{  $S \gets \GC$, $S_{prev} \gets \GC$, $S_{new}\gets\{ \}$, $C\gets \id$, $\texttt{new}\gets 1$.  }
    \While{\texttt{new}>0}{
        $S_{new}\gets\{ \}$.\;
        \For{$h_0 \in \GC$}{
            \For{$h \in S_{prev}$}{
                $C\gets[h_0,h]$\;
                \If{$S\cup\{C\}$ is linearly independent\;}{
                    $S\gets S\cup\{C\}$.\;
                    $S_{new}\gets S_{new}\cup\{C\}$.\;
                }
            }
        }
        $\texttt{new}\gets |S_{new}|$.\;
        $S_{prev} \gets S_{new}$.\;
    }
    \Return{$S$}
    \caption{Basis for the Dynamical Lie Algebra (DLA).}
    \label{alg:lie}
 \end{algorithm}

\section{Proof of Theorem 1: Convergence of controllable systems to 2-designs}\label{app:proof_theo_convergence}

In the following we provide a proof for Theorem \ref{theo:depth}, which we recall for convenience.
\begin{theorem}\label{theo:depthSM}
Consider a controllable system.  Then, the PSA $U(\thv)$ will form an $\varepsilon$-approximate $2$-design, i.e.  $\Vert\AC^{(2)}_{U(\thv)}\Vert_{\infty}=\varepsilon$ with $\epsilon>0$, when the number of layers $L$ in the circuit is
\begin{equation}
    L=\frac{\log(1/\varepsilon)}{\log\left(1/\Vert\AC^{(2)}_{U_1(\thv)}\Vert_{\infty}\right)}\,.
\end{equation}
Here $\Vert\AC^{(2)}_{U_1(\thv)}\Vert_{\infty}$ denotes the expressibility of a single layer  $U_1(\thv_1)$ of the ansatz according to Eqs.~\eqref{eq:ansatz} and~\eqref{eq:superop}.
\end{theorem}

\begin{proof}

To study the convergence of the PSA $U(\vec{\theta})$ to an approximate $2$-design we employ the tools of Harmonic analysis. The following arguments are based on Ref.~\cite{brown2010random}. This is similar to using Fourier analysis to study the convergence of a probability distribution on a real line to the normal distribution, which is also known as the central limit theorem.

The second moment operator corresponding to the distribution $\Ubb$ over unitaries $U$ can be defined as follows
\begin{align}
     M^{(2)}_{U} = \int_{\Ubb} dU  U\otimes U\otimes U^{*}\otimes U^{*}\,.
\end{align}
Let $M^{(2)}_{U_H}$ denote the second moment operator corresponding to the Haar distribution. Our goal is to evaluate the difference between $M^{(2)}_{U}$ and $M^{(2)}_{U_H}$.  Let us review some properties of $M^{(2)}_{U_H}$ before we calculate this distance operator. An important property of $M^{(2)}_{U_H}$ is that it is a projector onto a two-dimensional subspace, that is, $M^{(2)}_{U_H}$  has eigenvalues $0$ or $1$. We show this by noting that the following equations hold:
\begin{align}
    ( M^{(2)}_{U_H})^2 &= \left(\int_{\UC(d)} d\mu(U)  U\otimes U\otimes U^{*}\otimes U^{*}\right) \left(\int_{\UC(d)} d\mu(V)  V\otimes V\otimes V^{*}\otimes V^{*}\right)\nonumber \\
    &= \int_{\UC(d)} d\mu(U) \int_{\UC(d)} d\mu(V) \left(UV \otimes UV \otimes (UV)^* \otimes (UV)^*\right) \label{eq:invariance}\\
    &= \int_{\UC(d)} dU  dW  \left(W \otimes W\otimes W^{*}\otimes W^* \right) \nonumber\\
    &= \int_{\UC(d)} dW  W\otimes W\otimes W^{*}\otimes W^* \nonumber\\
    &=  M^{(2)}_{U_H}\,,
\end{align}
and
\begin{align}
    \Tr\left[ M^{(2)}_{U_H}\right] &= \Tr\left[\int_{\UC(d)} d\mu(U)  U\otimes U\otimes U^{*}\otimes U^{*}\right] \nonumber\\
    &= \sum\limits_{i,j,k,l=1}^d \int_{\UC(d)} d\mu(U)  U_{i,i} U_{j,j} U^*_{k,k} U^*_{l,l} \label{eq:wein}\\
    &= \frac{2}{d^2-1} \sum_{i,j=1}^d 1 - \frac{2}{d(d^2-1)} \sum_{i=1}^d 1 \nonumber \\
    &= 2\,.
\end{align}
In Eq.~\eqref{eq:invariance}, we used the left invariance of the Haar measure, and in~\eqref{eq:wein} we used the Weingarten function to explicitly evaluate the integral. The first property puts in evidence that $M^{(2)}_{U_H}$ is a projector and the second property shows that the eigenspace with eigenvalue 1 is a two-dimensional subspace.

Let $V=U_L \cdots U_2 U_1$, be an $L$-layered PSA, where each unitary $U_j$ is sampled from the same distribution $d\mu=P(U) dU$. Then, the probability distribution and moment operator of $V$ are respectively given by

\begin{align}
    \mu_L[V]&= \int_{\UC(d)} dU_1 \int_{\UC(d)} dU_2 \cdots \int_{\UC(d)} dU_L P(U_1)P(U_2)\cdots P(U_L) \delta(V-U_L\cdots U_2U_1)\,, \\
    M^{(2)}_{V}&= \int_{\UC(d)} d\mu_L[V] \quad V\otimes V\otimes V^{*}\otimes V^* \nonumber\\
    &= \int_{\UC(d)} dV  \int_{\UC(d)} dU_1 \int_{\UC(d)} dU_2 \cdots \int_{\UC(d)} dU_L  P(U_1)P(U_2)\cdots P(U_L) \delta(V-U_L\cdots U_2U_1) \quad V\otimes V\otimes V^{*}\otimes V^* \nonumber\\
    &=\int_{\UC(d)} dU_1 \int_{\UC(d)} dU_2 \cdots \int_{\UC(d)} dU_L  P(U_1)P(U_2)\cdots P(U_L) \quad (\prod_{l=1}^L U_l)\otimes (\prod_{l=1}^L U_l)\otimes (\prod_{l=1}^L U_l^*)\otimes (\prod_{l=1}^L U_l^*)\nonumber \\
    &=\prod_{l=1}^L \int_{\UC(d)} dU_l P(U_l) U_l\otimes U_l\otimes U_l^* \otimes U_l^* \nonumber\\
    &= \left(M^{(2)}_{U_1}\right)^L\,.\label{eq:momentV}
\end{align}
Equation~\eqref{eq:momentV} shows that the moment operator of an $L$-layered ansatz is equal to the $L$-th power of the moment operator of a single layer. We can also calculate this formally. In our case each $U_l$ is given by $U_l=\prod\limits_{k=1}^K e^{-i \theta_{lk} H_k}$, where $\GC=\{H_k\}_{k=1}^K$ are the set of generators, and where the $\theta_{lk}$ are sampled from the uniform distribution. Then, let us note that 
\small
\begin{align}
   M^{(2)}_{U_1} &= \left(\frac{1}{2\pi}\right)^n\int \prod_{k=1}^K d\theta_k \prod_{k=1}^K e^{-i \theta_k H_k} \otimes \prod_{k=1}^K e^{-i \theta_k H_k} \otimes \prod_{k=1}^K e^{i \theta_k H^*_k} \otimes \prod_{k=1}^K e^{i \theta_k H^*_k} \nonumber\\
    &= \left(\frac{1}{2\pi}\right)^n\prod_{k=1}^K \int d\theta_k e^{-i \theta_k H_k} \otimes e^{-i \theta_k H_k} \otimes  e^{i \theta_k H^*_k} \otimes e^{i \theta_k H^*_k}\nonumber\\
    &= \left(\frac{1}{2\pi}\right)^n \prod_{k=1}^K \left(V_k^{\otimes 2} \otimes V_k^{*\otimes 2}\right)\int d\theta_k \!\!\!\sum_{\substack{l_{1,k},l_{2,k}\\l_{3,k},l_{4,k}}}\!\!\! e^{-i\theta_k(l_{1,k} +l_{2,k} -l_{3,k} -l_{4,k} )} \ketbraq{l_{1,k},l_{2,k},l_{2,k},l_{4,k}} \left(W_k^{\otimes 2} \otimes W_k^{*\otimes 2}\right)^{\dagger} \nonumber\\
    &= \prod_{k=1}^K \left(W_k^{\otimes 2} \otimes W_k^{*\otimes 2}\right) \sum_{\substack{l_{1,k},l_{2,k}\\l_{3,k},l_{4,k}}} \delta(l_{1,k} +l_{2,k} -l_{3,k} -l_{4,k} ) \ketbraq{l_{1,k},l_{2,k},l_{3,k},l_{4,k}} \left(W_k^{\otimes 2} \otimes W_k^{*\otimes 2}\right)^{\dagger}\,.
\end{align}
\normalsize

Here $H_k\ket{l_{i,k}}=l_{i,k}\ket{l_{i,k}}$, and $W_k$ is the unitary matrix that diagonalizes $H_k$. To calculate the distance to a  $2$-design we need to prove some properties of eigenvalues and eigenvectors of $M^{(2)}_{U_1}$. Let $\ket{\phi}$ be an eigenvector of $M^{(2)}_{U_1}$ with associated eigenvalue $\lambda$. Then,
\begin{align}
    |\bra{\phi}M^{(2)}_{U_1}\ket{\phi}| = |\lambda| &= \Big| \int_{\UC(d)} dU_1 \ P(U_1) \  \bra{\phi}U_1 \otimes U_1\otimes U_1^* \otimes U_1^*\ket{\phi} \Big| \\
   & \leq  \int_{\UC(d)} dU_1 \ P(U_1) \  \Big|\bra{\phi}U_1 \otimes U_1\otimes U_1^* \otimes U_1^*\ket{\phi}\Big| \\
   &\leq 1\,.
\end{align}
The equality holds if and only if $\ket{\phi}$ is an eigenvector of $U \otimes U_1\otimes U_1^* \otimes U_1^*$ $\forall U_1$, such that $P(U_1)\neq 0$. For the specific case of Haar measure, this means that $\ket{\phi}$ is an eigenvector of $U \otimes U\otimes U^* \otimes U^*$ $\forall U \in \SC\UC(d)$. We already showed that there are two such eigenvectors for Haar measure with eigenvalue 1. 

Now, using the following argument we hope to show that those two are also the only eigenvectors of $M^{(2)}_{U_1}$ with eigenvalue 1, given that the set $\GC$ is controllable. Let $\ket{\phi}$ be an eigenvector of $M^{(2)}_{U_1}$ with eigenvalue 1. One can now also see that $\ket{\phi}$ is also an eigenvector of $M^{(2)}_{V}$ (since $M^{(2)}_{V}=(M^{(2)}_{U_1})^L$) $\forall L\in \Zbb^+$. Then, $\GC$ being a controllable set implies that for all $U\in \SC\UC(d)$, there exists an $L$ for which $V=\prod\limits_{l=1}^L\prod\limits_{k=1}^K e^{-i \theta_{lk} H_k}=U$  so that $P(U)\neq 0$. That, is one can obtain any unitary in $\UC(d)$ by tuning the parameters in $V$. It also implies that $\ket{\phi}$ has to be an eigenvector of $U \otimes U\otimes U^* \otimes U^*$ with eigenvalue 1. But this means that $\ket{\phi}$ is an eigenvector of $U \otimes U\otimes U^* \otimes U^*$ with eigenvalue 1 $\forall U\in \SC\UC(d)$. So $\ket{\phi}$ is also an eigenvector of $M_2[\mu_H]$ with eigenvalue 1, and there are two such eigenvectors. Let us call them $\ket{\phi_1}$, and $\ket{\phi_2}$.  Hence, we have
\begin{align}
    M^{(2)}_{U_H} &= \ketbraq{\phi_1} + \ketbraq{\phi_2} \\
     M^{(2)}_{U_1} &= \ketbraq{\phi_1} + \ketbraq{\phi_2} + \sum\limits_{i=1}^{d^4-2} \lambda_i \ketbraq{\psi_i} \\
      M^{(2)}_{V} &= \ketbraq{\phi_1} + \ketbraq{\phi_2} + \sum\limits_{i=1}^{d^4-2} \lambda_i^L \ketbraq{\psi_i} \,.
\end{align}
Here $\{\lambda_i\}$ is the set of the remaining eigenvalues such that $0\leq|\lambda_i|<1$. Let $\lambda_{\text{max}}$ be the eigenvalue with maximum modulus. Thus, we can now show that   $\Vert\AC^{(2)}_{U(\thv)} \Vert_\infty=\Vert M^{(2)}_{U_H}-M^{(2)}_{V} \Vert_\infty = \Vert\sum\limits_{i=1}^{N^4-2} \lambda_i^k \ketbraq{\psi_i}\Vert_\infty = |\lambda^L_{\text{max}}|$. Then, recalling that  $|\lambda_{\text{max}}|=\Vert\AC^{(2)}_{U_1(\thv)} \Vert_\infty$ is also the expressibility of one layer, one can find that 
\begin{equation}\label{eq:explayers}
    \Vert\AC^{(2)}_{U(\thv)}\Vert_{\infty} =\left( \Vert\AC^{(2)}_{U_1(\thv)}\Vert_{\infty} \right)^L\,.
\end{equation}
Solving for $L$ and denoting $\varepsilon=\Vert\AC^{(2)}_{U(\thv)}\Vert_{\infty}$ leads to 
\begin{equation}\label{eq:Ldepth}
    L=\frac{\log(1/\varepsilon)}{\log\left(1/\Vert\AC^{(2)}_{U_1(\thv)}\Vert_{\infty}\right)}\,.
\end{equation}

\end{proof}

\section{Proof of Corollary \ref{cor:Lscaling}: Rate of convergence of controllable systems to 2-designs}\label{app:proof_corollary_convergence}
Here we prove Corollary~\ref{cor:Lscaling}.
\begin{corollary}\label{cor:LscalingSM}
Let the single layer expressibility  of a controllable system be $\Vert\AC^{(2)}_{U_1(\thv)}\Vert_{\infty}=1-\delta(n)$, with $\delta(n)$ being at most polynomially vanishing with $n$, i.e., with $\delta(n)\in\Omega(1/\poly(n))$. Then, if $L(n)\in\Omega(n/\delta(n))$,  $U(\thv)$ will be no worse than an $\varepsilon(n)$-approximate $2$-design (i.e., $\Vert\AC^{(2)}_{U(\thv)}\Vert_{\infty}\leq\varepsilon(n)$) with $\varepsilon(n)\in\OC(1/2^n)$, where we have added the $n$-dependence in $L$ and $\varepsilon$ for clarity.
\end{corollary}

\begin{proof}
Let us recall from the main text, and from Eq.~\eqref{eq:explayers}, that
\begin{equation}
    \Vert\AC^{(2)}_{U(\thv)}\Vert_{\infty} =\left( \Vert\AC^{(2)}_{U_1(\thv)}\Vert_{\infty} \right)^L\,.
\end{equation}
Replacing $\Vert\AC^{(2)}_{U_1(\thv)}\Vert_{\infty}=1-\delta(n)$, one finds 
\begin{equation}
    \Vert\AC^{(2)}_{U(\thv)}\Vert_{\infty} =\left( 1-\delta(n) \right)^L\,.
\end{equation}
Then, assuming that  $L(n)\in\Omega(n/\delta(n))$ we have, by definition, that there exists $c\geq 0$ and $n_0$ such that $L(n)\geq cn/\delta(n)$ for all $n\geq n_0$. From the previous, we find that for $n\geq n_0$
\begin{equation}
    \Vert\AC^{(2)}_{U(\thv)}\Vert_{\infty} =\left( 1-\delta(n) \right)^L \leq \left( 1-\delta(n) \right)^{\frac{cn}{\delta(n)}}\,,
\end{equation}
where we have used the fact that $\left( 1-\delta(n) \right)\leq 1$. Recalling that $\delta(n)\in\Omega(1/\poly(n))$, i.e., $\delta(n)$ vanishes no faster than a polynomial function of $n$, we find to first order that 
\begin{equation}
    \left( 1-\delta(n) \right)^{\frac{c'n}{\delta(n)}}\approx e^{-n+\OC(\frac{1}{n})}\,.
\end{equation}
Then, defining $\varepsilon(n)=\left( 1-\delta(n) \right)^{\frac{cn}{\delta(n)}}$ we prove the result in  Corollary~\ref{cor:Lscaling} as
\begin{equation}
     \Vert\AC^{(2)}_{U(\thv)}\Vert_{\infty}\leq \varepsilon(n)\,, \quad \text{with}\quad \varepsilon(n)\in\OC(\frac{1}{2^n}).
\end{equation}

\end{proof}

\section{Proof of Proposition \ref{prop:controllable}: Controllability leads to barren plateaus}\label{app:proof_prop_controllable}

Let us now provide a proof for Proposition~\ref{prop:controllable}.

\begin{proposition}[Controllable]\label{prop:controllableSM}
There exists a scaling of the depth for which controllable systems form $\varepsilon$-approximate $2$-designs with $\varepsilon\in\OC(1/2^n)$, and hence the system exhibits a barren plateau according to Definition~\ref{def:BP}.
\end{proposition}

\begin{proof}
Let us start by noting that for all controllable systems one can form a $\varepsilon$-approximate $2$-designs with $\varepsilon\in\OC(1/2^n)$ with a depth scaling obtained from Eq.~\eqref{eq:Ldepth}.  Then, let us  recall that we have defined the expressibility superoperator as
\begin{align}~\label{eq:superopSM}
    \AC^{(t)}_{U(\thv)}=  M^{(t)}_{U_H}- M^{(t)}_{U(\thv)} \,.
\end{align}
with its ordinary action given by 
\begin{align}~\label{eq:superopSM2}
    \AC^{(t)}_{U(\thv)}(\cdot)=  M^{(t)}_{U_H}(\cdot)- M^{(t)}_{U(\thv)}(\cdot) \,.
\end{align}
We then have that for any quantum state $\rho$
\begin{equation}\label{eq:ineqexpress}
    \Vert \AC^{(t)}_{U(\thv)}(\rho)\Vert_{\infty}\leq \Vert \AC^{(t)}_{U(\thv)}\Vert_{\infty}\,,
\end{equation}
which follows from the normalization of $\rho$. Hence, if $\Vert \AC^{(t)}_{U(\thv)}\Vert_{\infty}\leq \varepsilon$ we find from Eq.~\eqref{eq:ineqexpress} that for any quantum state $\rho$, the inequality $\Vert \AC^{(t)}_{U(\thv)}(\rho)\Vert_{\infty}\leq \varepsilon$ holds, which is precisely the definition of an $\varepsilon$-approximate state $2$-design. Finally, we can use the results from~\cite{mcclean2018barren}, which imply that since  $\varepsilon\in\OC(1/2^n)$, then the variance of the cost function partial derivative is given by
\begin{equation}
    \Var_{\thv} [\partial_{\mu} C(\thv)]= \widehat{F}(n)\,, \quad \text{with} \quad \widehat{F}(n)\in\OC(1/2^n)\,.
\end{equation}
Thus, the cost function  exhibit a barren plateau.

\end{proof}

\section{Proof of Proposition \ref{prop:controllable-gen}: Controllability of the HEA and the Spin Glass model}\label{app:proof_of_hea_sg}

Let us now prove Proposition~\ref{prop:controllable-gen}.
\begin{proposition}\label{prop:controllable-genSM}
The following two sets of generators generate full rank DLAs, and concomitantly lead to controllable systems:
\begin{itemize}
    \item $\GC_{\rm HEA}=\Big\{X_i,Y_i\Big\}_{i=1}^n \bigcup \left\{\sum_{i=1}^{n-1} Z_i Z_{i+1}\right\}$,
    \item $\GC_{\rm SG}=\left\{\sum_{i=1}^n X_i, \sum_{i<j}  \left(h_i Z_i + J_{ij} Z_iZ_j\right)\right\}$, with  $h_i,J_{ij}\in\Rbb$ sampled from a Gaussian distribution.
\end{itemize}
\end{proposition}

In the following we find that by repeated nested commutators between the elements of the sets $\GC$ in Proposition~\ref{prop:controllable-gen}. one can obtain all $2^{2n}-1$ Pauli strings, and hence, that the DLA $\liea$ obtained is full rank.

\begin{proof}
We divide the proof into HEA and GS models.

\subsubsection{Generators of the Hardware Efficient Ansatz (HEA)}
We first start with the set $\GC_{\rm HEA}$ corresponding to the set of generators of a HEA. 
\begin{equation}\label{eq:gen-hea-supp}
\GC_{\rm HEA}=\Big\{X_i,Y_i\Big\}_{i=1}^n \bigcup \left\{\sum_{i=1}^{n-1} Z_i Z_{i+1}\right\}~.
\end{equation}

First, let us note that from the commutation of $X_i$ with $Y_i$, we get every $Z_i$. Meaning that one can already obtain all single qubit Pauli operators. Then, the commutation of $X_i$ and $Y_i$ with $\sum_{i=1}^{n-1} Z_i Z_{i+1}$, respectively, gives
\begin{align}
   A_i =Y_i (Z_{i+1}+Z_{i-1})\,,\quad \text{and}\quad   B_i   = X_i (Z_{i+1}+Z_{i-1})\,.
\end{align}
It then follows that the commutator of $X_{i+1}$ with $B_i$ is 
\begin{equation}
    C_i = X_i Y_{i+1},~\forall i \in \{1, \dots, n-1\}~.
\end{equation}
Then by computing the commutators of $\{X_i\}_{i=1}^n$, $\{Y_i\}_{i=1}^n$, and $\{Z_i\}_{i=1}^n$ with $\{C_i\}_{i=1}^{n-1}$ we get all nearest-neighbour two-body Pauli operators 
\begin{equation}
   D_i =  \{X_iX_{i+1}, Y_iY_{i+1}, Z_iZ_{i+1}, X_iY_{i+1}, Y_iX_{i+1}, X_iZ_{i+1}, Z_iX_{i+1}, Y_iZ_{i+1}, Z_iY_{i+1}\},~ \forall i \in \{1, \dots, n-1\}~.
\end{equation}
Similarly, it can be readily verified that the commutators between $\{X_i\}_{i=1}^n$, $\{Y_i\}_{i=1}^n$, $\{Z_i\}_{i=1}^n$, and operators in $\{D_i\}_{i=1}^{n-1}$ yield all ``nearest-neighbour'' three-body Pauli operators.

Now, let us show that $\liea$ also contains the remaining non-nearest-neighbour two body operators. Consider the commutator between the three-body nearest-neighbour operators

\begin{equation}
    [M_i N_{i+1} O_{i+2}, P_i N_{i+1} Q_{i+2}] = [M_i,P_i][O_{i+2},Q_{i+2}]
\end{equation}
where $N,M,O,P,Q\in\{X,Y,Z\}$. Clearly, the different choices of $M,O,P$ and $Q$ will generate all next-nearest-neighbour two-body operators. Iterating this procedure we obtain all $9\binom{n}{2}$ two-body terms. Then, once we have all two-bodies we can use one-bodies $\{X_i\}_{i=1}^n$, $\{Y_i\}_{i=1}^n$ and $\{Z_i\}_{i=1}^n$ to get all three-bodies. Three-bodies with one-bodies will give four-bodies, and so on. We will get all n-body operators. Thus the DLA of the HEA is full rank, which implies that the HEA is a controllable ansatz.

\bigskip 

\subsubsection{Generators of the Spin Glass (SG)}
The set of generators for a spin glass system is given by 
\begin{equation}
    \GC_{\rm SG}=\left\{\sum_{i=1}^n X_i, \sum_{i<j}  \left(h_i Z_i + J_{ij} Z_iZ_j\right)\right\}~, 
\end{equation}
with  $h_i,J_{ij}\in\Rbb$. For convenience, we define the following two operators
\begin{equation}
    H_p=\sum_{ij} h_i Z_i + J_{ij} Z_iZ_j,\quad H_m=\sum_i X_i~.
\end{equation}
The commutator of $H_p$ and $H_m$ gives
\begin{equation}
 A_0 =   [H_p,H_m]= \sum_{ij} h_i Y_i + J_{ij} Y_iZ_j~. 
\end{equation}

We then compute the commutator of $A_0$ and $H_m$: 
\begin{equation}
[H_m,A_0] =-H_p+\sum_{ij} J_{ij}Y_iY_j~.
\end{equation}
Combining $H_p$ and $[H_m,A_0]$ gives
\begin{align}
    A_1 = \sum_{ij} J_{ij}Y_iY_j~.
\end{align}

We now compute the commutator between $A_1$ and $H_m$ as follows:
\begin{equation}
    A_3:= [A_1, H_m] = \sum_{rs} J_{rs} Z_r Y_s~.
\end{equation}

Combining $A_0$ and $A_3$, we get 
\begin{equation}
    A_4 := \sum_{i} h_i Y_i~.
\end{equation}

Similarly, combining $H_m$ and $A_4$ gives
\begin{equation}
    A_5 := \sum_i h_i Z_i~.
\end{equation}

Finally, combining $A_5$ and $H_p$ gives 
\begin{equation}
    A_6 := \sum_{ij} J_{ij}Z_iZ_j~.
\end{equation}

From the commutator of $A_4$ and $A_5$ we get 
\begin{equation}
  A_7:= [A_4, A_5] = \sum_i h_i^2 X_i~.
\end{equation} 
Moreover, the commutators of $A_7$ with $A_4$ and $A_5$ lead to 
\begin{equation}
    A_8 := [A_7, A_4] = \sum_i h_i^3 Z_i~,\quad
    A_9 := [A_5, A_7] = \sum_i h_i^3 Y_i~.
\end{equation}
By repeating this procedure, we get that the set
\begin{equation}
     \bar{\SC}=\{ \sum_i h_i^{2m} X_i,~ \sum_i h_i^m Y_i~, \sum_i h_i^m Z_i\}_{m=1}^n
\end{equation}
also belongs to the Lie algebra. Now, because the $h_i'$s are sampled from a Gaussian distribution, we can safely assume them to be non-zero and different from each other. Then, using a Vandermonde determinant type of argument one can show that the $3n$ elements in $\bar{\SC}$ are linearly independent and span the same subspace as $\SC=\{X_i,Y_i,Z_i\}_{i=1}^n$. Thus $\SC$ belongs to $\liea_{\rm{SG}}$. Combining this with $A_6$, we essentially get the generators of $\GC_{\rm{HEA}}$ and hence one can again generate all $n$-body Pauli operators. Thus the DLA of the spin-glass system is also of full rank, which implies that the spin-glass system is controllable.

\end{proof}

\section{Proof of Theorem 2: Variance in subspace controllable systems}\label{app:proof_theo_subspace}

In the following, we provide a proof for Theorem \ref{Theo:subspace} by explicitly computing the variance of the cost function partial derivative in a subspace controllable setting. Consider a set of generators that share a symmetry (for simplicity we assume only one symmetry, although generalization to multiple symmetries is straightforward), i.e., there is a Hermitian operator $\Sigma$ such that $[\Sigma,g]=0 \quad \forall g\in\liea $. Assuming $\Sigma$ has $N$ distinct eigenvalues, the DLA has the form $\liea=\bigoplus_{m=1}^N \liea_m$. This imposes a partition of $\HC=\bigoplus_{m=1}^N \HC_m$ where each subspace $\HC_m$ of dimension $d_m$ is invariant under $\liea$.  

Let us introduce some notation. Consider the $d\times d_m$ matrix that results from horizontally stacking the eigenvectors of $\Sigma$ associated with the $m$-th eigenvalue (of degeneracy $g_m$)
\begin{equation}
Q_m\ad = \begin{bmatrix} \vdots &  \vdots  && \vdots \\ \ket{v_1}, & \ket{v_2},&  &, \ket{v_{g_m}} \\ \vdots &  \vdots  && \vdots \end{bmatrix}, 
\end{equation}
such that $Q_m$ maps vectors from $\H$ to $\H_m$. These satisfy 

\begin{equation}
Q_m Q_n\ad = \id_{d_m} \delta_{mn},\quad  Q_m\ad Q_m  = \mathbb{P}_m\,,
\end{equation}\label{eq:proj2}
where $\mathbb{P}_m$ are projectors onto the subspace, such that $\sum_{m=1}^N \mathbb{P}_m=\id_d$. Let us now use the notation
\begin{equation}\label{Eq_red}
    \ket{\psi}^{(m)} = Q_m \ket{\psi}, \quad   A^{(m)} = Q_m A Q_m\ad\,,
\end{equation}
to denote the $d_m$-dimensional reduced states and operators, respectively. Recall that, since any unitary $U\in\lieg$ produced by such a system is block diagonal, we can write $U= \sum_m \P_m U \P_m$. Also, let us note that if $A=A\ad$ then $(A\k)\ad=A\k$.

We are ready to prove of Theorem \ref{Theo:subspace}, which we here recall for convenience.

\begin{theorem}[Subspace controllable]\label{Theo:subspaceSM}
Consider a system that is reducible (so that the Hilbert space is $\HC=\bigoplus_j \HC_j$ with each $\HC_j$ invariant under $\lieg$), and controllable on some $\HC_k$ of dimension $d_k$. Then, if the initial state is such that $\rho\in\HC_k$, the variance of the cost function  partial derivative is given by
\small
\begin{equation}\label{Eq_th2SM}
\Var_{\vec{\theta}}[\partial_{\mu} C(\thv)] =\frac{2d_k}{(d_k^2-1)^2} \Delta(H_{\mu}\k)\Delta(O\k)\Delta(\rho\k)\,.
\end{equation}
\normalsize
Here 
\begin{equation}\label{eq:hs}
\Delta(A) = D_{HS}\left(A,\Tr[A]\frac{\id_d}{d}\right),\quad \text{with }D_{HS}\left(A,B\right)=\Tr[(A-B)^2]
\end{equation}
is the Hilbert-Schmidt distance, and $A\k$ the reduction of operator $A$ onto the subspace of $\H_k$ as defined in Eq \eqref{Eq_red}.
\end{theorem}

\begin{proof}

Consider the partial derivative of the cost function $C(\thv)$ with respect to the parameter $\th_{pq}$ ($= \theta_\mu$), i.e. the one associated with layer $p$ and generator $H_{q}$ ($=H_\mu$). We have

\begin{equation}
    \partial_{\mu} C(\thv) = i\Tr[ U_B \rho U_B\ad [H_\mu,O_A]]
\end{equation}

\begin{equation}
\begin{split}
    \partial_{pq} C(\thv)&= \partial_{pq}\left(\Tr[U(\thv) \rho U(\thv)\ad O]\right) \\
    &=\partial_{pq}\left(\Tr[U_B \rho U_B \ad O_A]\right)\\
    &=i\Tr[ U_B \rho U_B\ad [H_q,O_A]]
\end{split}
\end{equation}
where in the second line we have expanded $U(\thv)=U_AU_B$, with
\begin{equation}
U_{B} = \prod_{m=0}^q e^{-i H_m \theta_{pm}} \left(\prod_{l=1}^{p-1} \prod_{k=0}^K e^{-i H_k \theta_{lk}}\right)\,\quad\text{and}\quad
U_A = \prod_{l=p+1}^L \prod_{k=0}^K e^{-i H_k \theta_{pk}}\left(\prod_{k=q+1}^K e^{-i H_k \theta_{pk}}\right)\,,
\end{equation}
corresponding to the unitaries before and after the parameter $\th_{pq}$,
and $O_A=U_A\ad O U_A$. Then, the variance of the partial derivative is
\begin{equation}
    \Var_{\thv}[ \partial_{pq} C] = - \int_{\Ubb_A} dU_A I(U_A)\,,\quad \text{with } \quad I(U_A) = \int_{\Ubb_B} dU_B \Big(F(U_B,U_A)\Big)^2\,,
\end{equation}
where $\Ubb_A$ and $\Ubb_B$ denote the distribution of before and after unitaries, respectively,
and $F(U_B,U_A) = \Tr[ U_B \rho U_B\ad [H_q,U_A\ad O U_A]]$.
Now, expanding the block-diagonal unitaries in this equation and assuming that the initial state belongs to a particular invariant subspace, i.e., $\rho\k=\P_k \rho \P_k=\rho$, we find
\begin{equation}
    F(U_B,U_A) = \sum_{mn} \Tr[ \P_m U_B \P_m \rho \P_n U_B\ad \P_n X ]
    = \Tr[U^{(k)}_B \rho^{(k)} U^{\dagger (k)}_B X]
\label{Eq_F}
\end{equation}
where we defined $X=[H_q\k,(U_A\k)\ad O\k U_A\k]$, and thus

\begin{equation}
    I(U_A) = \int_{\Ubb_B} dU_B \left(\Tr[ U_B \rho U_B\ad [H_q,U_A\ad O U_A]\right)^2= \int_{\Ubb_B\k} dU_B^{(k)} \left(\Tr[U^{(k)}_B \rho^{(k)} U^{\dagger (k)}_B X]\right)^2\,.
\label{Eq_Ik}
\end{equation}
Note that here we have used the fact that, owing to Eq. \eqref{Eq_F}, $F(U_B,U_A)$ actually depends only on the action of $U_B\k$, i.e., on the action of  $U_B$ projected onto the $k$-invariant subspace, and  not on the action of the entire unitary $U_B$,. Hence, the integration over $\Ubb_B$ can be replaced by an integration over $\Ubb_B\k$.

At this point, we introduce the assumption of subspace controllability on $\HC_k$. By virtue of Theorem \ref{theo:depth} and Corollary~\ref{cor:Lscaling}, we know that the distribution of unitaries produced by a subspace controllable PSA constitute a $\varepsilon$-approximate 2-design in the subspace. Therefore, we can use Eq.~\eqref{Eq_lemma2} to integrate and get
\begin{equation}
    I(U_A) = \int_{\UC(d_k)} dU\k \left(\Tr[U\k \rho^{(k)} (U\k)\ad X]\right)^2 = \frac{\Tr[XX] \,\Delta(\rho\k)}{d_k^2-1}\,,
\label{Eq_Iready}
\end{equation}
where $\Delta(A)$ is the Hilbert-Schmidt distance defined in Theorem \ref{Theo:subspaceSM}. Here, we used that $\Tr[X]=0$. This is grounded in the fact that, because the generator $V$ shares the symmetry, the commutator $X$ in the subspace is still a commutator, i.e., $X=[H_q\k,\tilde{O}\k]$ with $\tilde{O}=U_A\ad O U_A$. If the initial state was spread across two (or more) subspaces, neither $\rho\k$ would be a density matrix, nor $X$ would be a commutator and one would have to be more careful in the derivation.

Finally, we proceed to integrate Eq.~\eqref{Eq_Iready} over $\Ubb_A$. Notice again that $\tH\k=\P_k \tH \P_k = U_A^{\dagger (k)} O\k U_A\k$ so $X$ is actually only a function of $U_A\k$ (and not of the entire $U_A$). Consequently we can integrate over the reduced distribution $\Ubb_A\k$, that, according to the subspace controllability assumption, forms a $\varepsilon$-approximate $2$-design over $\UC(d_k)$. This leads to
\begin{align}
    \Var_{A,B} [\partial_{pq} C]&= -\frac{\Delta(\rho\k)}{d_k^2-1} \int_{\UC(d_k)} d\mu(U\k)\, \Tr[X X] \\
    &= -2\frac{\Delta(\rho\k)}{d_k^2-1} \int_{\UC(d_k)} d\mu(U\k) \Big( \Tr[ H_q\k \tilde{O}\k H_q\k \tilde{O}\k] - \Tr[ H_q\k H_q\k \tilde{O}\k\tilde{O}\k] \Big) \label{eq10}\\
    &=-2\frac{\Delta(\rho\k)}{d_k^2-1}\Big( \frac{(\Tr[H_q\k])^2\Tr[O\k O\k]+\Tr[H_q\k H_q\k](\Tr[O\k])^2}{d_k^2-1}\nonumber\\
    &\quad - \frac{\Tr[H_q\k H_q\k]\Tr[O\k O\k]+(\Tr[H_q\k]\Tr[O\k])^2}{d_k(d_k^2-1)}-\frac{\Tr[H_q\k H_q\k]\Tr[O\k O\k]}{d_k} \Big) \nonumber\\
    &=\frac{2d_k\Delta(\rho\k)}{d_k(d_k+1)(d_k^2-1)}\left(\Tr[H_q\k H_q\k]- \frac{1}{d_k}\Tr[H_q\k]^2\right)\left(\Tr[O\k O\k] -\frac{1}{d_k}\Tr[O\k])^2\right)  \label{eq11}\\
     &=\frac{2d_k}{(d_k^2-1)^2}\Delta(H_q\k)\Delta(O\k) \Delta(\rho\k)\,. \label{eq:th2_SM}
\end{align}
Here, we used the notation $\Var_{\thv}\xrightarrow[]{} \Var_{A,B}$ to make explicit the assumption that both $\Ubb_A\k$ and $\Ubb_B\k$ form $2$-designs in $\UC(d_k)$. Note that we first expanded $\Tr[X X]$ in Eq.~\eqref{eq10}, and then used identities \eqref{Eq_lemma2} and \eqref{Eq_lemma3} on each integrand respectively to arrive at Eq.~\eqref{eq11}. Note that, even though $\Tr[H_q]=0$ (and, in most applications, $\Tr[O]=0$), this is not necessarily true for their reduced analogs $H_q\k$ and $O\k$.

\end{proof}

\section{Proof of Corollary \ref{cor:HXXZ}: Exponentially growing subspaces have barren plateaus}\label{app:proof_corollary_subspace}

Let us prove Corollary~\ref{cor:HXXZ}, which we here recall.

\begin{corollary}\label{cor:HXXZSM}
Consider a PSA of the form in \eqref{eq:ansatz} giving rise to a reducible DLA, and let $\rho\in\HC_k$, with $\HC_k$ some invariant subspace that is controllable (i.e. the DLA reduced to such subspace is full rank). If, $\Tr[(H_\mu)^4],\Tr[O^4]\in \OC(2^n)$, the cost function will exhibit a barren plateau for any subspace such that $d_k\in\OC(2^n)$.
\end{corollary}
\begin{proof}

We consider the situation in which we have a reducible system and an initial state $\rho$ that belongs to an invariant subspace $\HC_k$ which, by assumption, is controllable.  Notice that, using Eqs.~\eqref{eq:proj2} and~\eqref{Eq_red}, we find
\begin{equation}
    \Tr[A\k A\k] =\Tr[A\P_k A\P_k]\,.
\end{equation}
Then, since  the operator $M=A\P_k A$ is a positive semi-definite operator, and since $\P_k\leq \id$, one  can write
\begin{equation}
    \Tr[A\k A\k] \leq \Tr[A^2\P_k]\,.
\end{equation}
Using this expression, the following bound on the Hilbert-Schmidt distance holds
\begin{align}
    \Delta(A\k) &=\Tr[A\k A\k]-\frac{1}{d_k}(\Tr[A\k])^2\\
    &\leq \Tr[A^2\P_k]\\
    &\leq \sqrt{\Tr[A^4]\Tr[\P_k^2]}\\
    &\leq \sqrt{d_k}\sqrt{\Tr[A^4]}\,.\label{eq:ineqA}
\end{align}
Using Eq.~\eqref{eq:ineqA} and the equation for the variance in ~\eqref{eq:th2_SM} we find
\begin{equation}
    \Var_{\thv}[\partial_{\mu} C(\thv)]\leq\frac{4d_k^2}{(d_k^2-1)^2}\sqrt{\Tr[H_\mu^4]} \sqrt{\Tr[O^4]} \,,
\end{equation}
where we have additionally used the fact that $\Delta(\rho\k)\leq 2$ $\forall \rho$. Recalling that we are interested in the case when $d_k\in\OC(2^n)$ we find that, assuming $\Tr[V^4],\Tr[O^4]\in \OC(2^n)$, then the function
\begin{equation}
    T(n)=\frac{4d_k^2}{(d_k^2-1)^2}\sqrt{\Tr[H_\mu^4]} \sqrt{\Tr[O^4]}
\end{equation}
is such that 
\begin{equation}
    T(n)\in\OC(1/2^n)\,.
\end{equation}
Hence, we finally find that $\Var_{\thv}[\partial_{\mu} C(\thv)]\leq T(n)$, and the cost exhibits a barren plateau from the fact that the variance of the cost function partial derivative is upper  bounded by a function that  vanishes exponentially with $n$.

\end{proof}

Let us finally denote that the proof of Corollary~\ref{cor:HXXZ} follow from the fact that $\Tr[(H_\mu)^4],\Tr[O^4]\in \OC(2^n)$. We here note that the following relevant cases satisfy this assumption:
\begin{itemize}
    \item $H_\mu$, $O$ are projectors of arbitrary rank.
    \item $H_\mu$, $O$ have a decomposition in the Pauli string basis of the form $\sum_i c_i \vec{\sigma}_i$ (with $c_i$ real coefficients, and $\vec{\sigma}_i\in\{\id,X,Y,Z\}^{\otimes n}$) with up to $\OC(\poly(n))$ terms such that $\sum_{i} c_i^4\in\OC(\poly(n)) $. 
\end{itemize}

\def\r{\rho}
\def\A{\AC}

\section{Proof of Theorem \ref{theo:express-subspace2}: Expressibility in the subspace}\label{app:proof_theo_express}

Here we prove Theorem~\ref{theo:express-subspace2}, which we now recall.

\begin{theorem}\label{theo:express-subspace2SM}
Consider a system that is reducible and let $\rho\in\HC_k$ with $\HC_k$ an invariant subspace of dimension $d_k$. Then, the variance of the cost function partial derivative is upper bounded by
\begin{equation}\label{eq:expressibilitySM}
\Var_\theta[ \partial_{\mu} C(\thv)]  \leq  \min\{G_A(\rho\k),G_B(O\k)\}\,,
\end{equation}
with
\begin{align}
  G_B(\rho\k) &= \left(\left\Vert \mathcal{A}_{U_B\k}\left((\rho\k)^{\otimes 2}\right)\right\Vert_2 -\frac{\Delta(\rho\k)}{d_k^2-1}\right)\Tr\left[\left\langle X^2\right\rangle_{U_A\k}\right]\label{eq:bound1}\\
  G_A(O\k) &= \left(\left\Vert \mathcal{A}_{U_A\k}\left((O\k)^{\otimes 2}\right)\right\Vert_2 -\frac{\Delta(O\k)}{d_k^2-1}\right)\Tr\left[\left\langle Y^2\right\rangle_{U_B\k}\right].\label{eq:bound2}
\end{align}
Here we define $X=[H_\mu\k, (U_A\k\,)^\dagger O\, U_A\k]$ and $ Y = [H_\mu\k, U_B\k \rho\k (U_B\k)\ad]$. For simplicity, we here employed the short-hand notation   $\langle\cdot \rangle_{U_x\k}$ (with $x=A,B$) to indicate the expectation value over the distribution of unitaries obtained from $U_x\k$ in the $k$-th subsystem. Finally, $\Vert M\Vert_2=\sqrt{\Tr[M\ad M]}$ is the Frobenius norm, and $\Delta(\cdot)$ was defined in Theorem~\ref{Theo:subspace}.
\end{theorem}

\begin{proof}

Let us first note that the variance of the cost function partial derivative can be expressed as
\begin{align}
    \Var_\theta[ \partial_{\mu} C(\thv)] &= - \int_{\mathbb{U}_A\k} dU_A\k \int_{\mathbb{U}_B\k} d U_B\k \Tr[U_B^{\otimes 2} (\rho\k)^{\otimes 2} (U_B^\dagger)^{\otimes 2} X^{\otimes 2}]\label{eq:variance-first-step-rho}\\
   &= \int_{\mathbb{U}_A} d U_A \int_{\mathbb{U}_B} d U_B \Tr[(U_A\ad)^{\otimes 2} (O\k)^{\otimes 2} {U_A}^{\otimes 2} Y^{\otimes 2}]\label{eq:variance-first-step-O}
\end{align}
where we defined
\begin{align}\label{eq:xlk}
X = [H_\mu\k, (U_A\k)^\dagger O\k U_A\k]\,, \quad \text{and} \quad Y = [H_\mu\k, (U_B\k) \rho\k (U_B\k)\ad]\,.
\end{align} 
Then, we recall that $\AC_x(\cdot)$ denotes the expressibility superoperator for the second moment in the $k$-th subspace 
\begin{align}
    \AC_\omega(\cdot)= \int_{\UC(d_k)} d\mu(U) \, U\ts(\cdot)(U\ad)\ts- \int_{\Ubb_\omega\k} d U \, (U)^{\otimes t}  (\cdot) (U^\dagger)^{\otimes t}\,,
\end{align}
where $\omega=A,B$ indicates that one evaluates the expressibility of $U_A\k(\thv)$ or $U_B\k(\thv)$ in the $k$-th subspace.

First, let us derive Eq.~\eqref{eq:bound1}. Replacing $\mathcal{A}_B((\rho\k)^{\otimes 2})$ into~\eqref{eq:variance-first-step-rho} leads to  
\begin{equation}
\begin{aligned}
    \Var_\theta[ \partial_{\mu} C(\thv)] &= - \int_{\mathbb{U}_A\k} d U_A\k \int_{\mathcal{U}(d_k)} d\mu(U) \Tr[U_{H}^{\otimes 2} (\rho\k)^{\otimes 2} U_{H}^{\dagger \otimes 2} X^{\otimes 2}] + \int_{\mathbb{U}_A\k} d U_A\k  \Tr[\mathcal{A}_B((\rho\k)^{\otimes 2}) X^{\otimes 2}] \\ 
    &= \Var_B[ \partial_{\mu} C(\thv)] +  \int_{\mathbb{U}_A\k} d U_A\k \Tr[\mathcal{A}_B((\rho\k)^{\otimes 2}) X^{\otimes 2}]\, .
\end{aligned}
\end{equation}
Which leads to  
\begin{equation}
    | \Var_\theta[ \partial_{\mu} C(\thv)]- \Var_B[ \partial_{\mu} C(\thv)] | \leq  \bigg| \int_{\mathbb{U}_A\k} d U_A\k \Tr[\mathcal{A}_B((\rho\k)^{\otimes 2}) X^{\otimes 2}] \bigg| \,.
\end{equation}
Using the triangle inequality and then the Cauchy-Schwarz inequality one finds
\begin{align}\label{eq:DiffVarR}
        | \Var_\theta[ \partial_{\mu} C(\thv)]- \Var_B[ \partial_{\mu} C(\thv)] |  &\leq \int_{\mathbb{U}_A\k} d U_A\k |  \Tr[\mathcal{A}_B((\rho\k)^{\otimes 2}) X^{\otimes 2}] | \\
        &\leq \Vert \mathcal{A}_B((\rho\k)^{\otimes 2}) \Vert_2 \int_{\mathbb{U}_A\k} d U_A\k || X^{\otimes 2}||_2 \\
         &\leq \Vert \mathcal{A}_B((\rho\k)^{\otimes 2})\Vert_2 \left\langle|| X^{\otimes 2}||_2\right\rangle_{U_A\k} \, .
\end{align}
Then, using the fact that
\begin{equation}\label{eq:xlk-simplification}
    || X^{\otimes 2}||_2 = \sqrt{\Tr[X^{\otimes 2}X^{\otimes 2}]} = \sqrt{\Tr[X^{2} \otimes X^{2}]} = \vert \Tr[X^2]\vert = \vert \Tr[[H_\mu\k, (U_A\k)^\dagger O\k U_A\k]^2] \vert\, ,
\end{equation}
one finds
\begin{align}\label{eq:DiffVarR2}
        | \Var_\theta[ \partial_{\mu} C(\thv)]- \Var_B[ \partial_{\mu} C(\thv)] |  \leq  \Vert \mathcal{A}_B((\rho\k)^{\otimes 2})\Vert_2 \Tr\left[\left\langle X^2\right\rangle_{U_A\k}\right] \, .
\end{align}
Note that the term $\Var_B[ \partial_{\mu} C(\thv)]$ was explicitly computed in Eq.~\eqref{Eq_Iready}, i.e., in Eq.~\eqref{Eq_Iready} one evaluates the variance when $U_B\k(\thv)$ forms a $2$-design. Hence, replacing $\Var_B[ \partial_{\mu} C(\thv)]= -\frac{\Delta(\rho\k)}{d_k^2-1} \Tr\left[\left\langle X^2\right\rangle_{\mathbb{U}_A\k}\right]$ in~\eqref{eq:DiffVarR2} we finally obtain
\begin{align}\label{eq:DiffVarR3}
         \Var_\theta[ \partial_{\mu} C(\thv)]  \leq  \left(\Vert \mathcal{A}_B((\rho\k)^{\otimes 2})\Vert_2 -\frac{\Delta(\rho\k)}{d_k^2-1}\right)\Tr\left[\left\langle X^2\right\rangle_{U_A\k}\right] \, .
\end{align}

Similarly, one can derive Eq.~\eqref{eq:bound2} by replacing $\mathcal{A}_A((O\k)^{\otimes 2})$ into~\eqref{eq:variance-first-step-rho}, which leads to  
\begin{equation}
\begin{aligned}
    \Var_\theta[ \partial_{\mu} C(\thv)] &= - \int_{U_B\k} d U_B\k \int_{\mathcal{U}(d_k)} d\mu(U) \Tr[U_{H}^{\otimes 2} (O\k)^{\otimes 2} U_{H}^{\dagger \otimes 2} Y^{\otimes 2}] + \int_{U_B\k} d U_B\k  \Tr[\mathcal{A}_A((O\k)^{\otimes 2}) Y^{\otimes 2}]\,.
\end{aligned}
\end{equation}
Following a similar derivation to the one used in obtaining~\eqref{eq:DiffVarR}, one finds
\begin{align}\label{eq:DiffVarR5}
         \Var_\theta[ \partial_{\mu} C(\thv)]  \leq  \left(\Vert \mathcal{A}_A((O\k)^{\otimes 2})\Vert_2 -\frac{\Delta(O\k)}{d_k^2-1}\right)\Tr\left[\left\langle Y^2\right\rangle_{U_B\k}\right] \, .
\end{align}

\end{proof}

\def\l{\lambda}
\def\d{\delta}
\def\a{\alpha}
\def\aa{\tilde{\alpha}}
\def\gam{\gamma}
\def\gg{\tilde{\gamma}}
\def\bb{\tilde{\beta}}
\def\b{\beta}
\def\bb{\tilde{\beta}}

\section{Proof of Proposition \ref{prop:varsu2}: Variance on the irreducible representations of $\SC\UC$(2)}\label{app:proof_su2}

In this section we derive a proof for Proposition~\ref{prop:varsu2} for the variance of the irreducible representations of $\SC\UC$(2).  Proposition~\ref{prop:varsu2SM}  reads.

\begin{proposition}\label{prop:varsu2SM}
Consider the cost function of Eq.~\eqref{eq:costsu2}. Let $\th_{\mu}=\th_{j,x}$, and let us assume that the circuit is deep enough to allow for the distribution of unitaries $U_A$ and $U_B$ to converge to $2$-designs on $\lieg=\SC\UC$(2). Then variance of the cost function partial derivative $\partial_{\mu} C(\thv)=\partial C(\thv)/\partial\theta_\mu$ is
\begin{equation}\label{eq:varsu2SM}
    \Var_{\thv}[\partial_{\mu} C(\thv)]=\frac{2m^2}{3}\,.
\end{equation}
\end{proposition}

\begin{proof}

First, let us bring the reader into context. We consider a toy model ansatz $U(\thv) =\prod_{l=1}^L e^{-i\theta_{lx} S_y}e^{-i\theta_{ly} S_x}$ with generators $\GC=\{S_x,S_y\}$, where $\{iS_x,iS_y,iS_z\}$, with  $S_\nu \in \Cbb^{d\times d}$ ($\nu=x,y,z$), form a basis of the spin $S=(d-1)/2$ irreducible representation of $\mf{su}$(2). Here, it is convenient to use as basis of the Hilbert space the set $\{\ket{m}\}$,  $m=-S,-S+1,\ldots,S-1,S$, of eigenvectors of $S_z$. That is, we have
\begin{equation}
    S_z\ket{m}=m\ket{m}\,,\,\,\,\quad S_{\pm}\ket{m}=\l_{\pm} (m)\ket{m\pm1}
\end{equation}
with $\l_{\pm}=\sqrt{S(S+1)-m(m\pm1)}$, and where $S_+$ and $S_-$ are the spin ladder operators such that $S_x=\frac{1}{2}(S_++S_-)$ and $S_-=\frac{1}{2i}(S_+-S_-)$.

The dynamical group $\lieg$ is the $d$-dimensional representation of $\SC\UC$(2) and we will be interested in the partial derivative w.r.t $\th_{\mu}=\th_{j,x}$, i.e., the parameter associated with generator $V=S_x$ on the $j$-th layer of the circuit. We will assume that the circuit is deep enough to allow for the distribution of unitaries $U_A$ and $U_B$ to converge to $2$-designs on $\lieg=\SC\UC$(2). In consequence, when computing the variance of the different gradient components, we will be allowed to replace the integration over the angles in the PSA with an integration over the Haar measure on the group $\SC\UC$(2). Moreover, because irreducible representations of $\SC\UC$(2) and $\SC\OC$(3) are isomorphic, we can choose to integrate over the latter
\begin{equation}
    \int_{\thv} dU(\thv) \xrightarrow[]{} \int_{\SC\UC(2)} d\mu(U)\,  \xrightarrow[]{} \int_{\SC\OC(3)} d\mu(U)\,.
\end{equation}
That is, we choose a parameterization of $\SC\OC$(3), e.g. in terms of Euler angles
\begin{equation}
    U(\a,\b,\gam) = e^{-i\a S_z}e^{-i\b S_y}e^{-i\gam S_z}\,,
\label{Eq_unitary_euler}
\end{equation}
and compute the mean value of a function of the form $F_M(\thv)=\Tr[U(\thv)\rho U(\thv)\ad M]$, with $M$ Hermitian, over the cost landscape as
\begin{equation}
\begin{split}
    \langle F \rangle &= \int_{\thv} dU(\thv) \Tr[U(\thv)\rho U(\thv)\ad M] \\
    &=\int_{\SC\OC (3)} d\mu(U) \Tr[U\rho U\ad M]
    \\&=\frac{1}{8\pi^2} \int_0^{2\pi}d\a\int_0^{\pi} d\b \sin{\b}\int_0^{2\pi} d\gam\, \Tr[ U(\a,\b,\gam) \rho U(\a,\b,\gam)\ad M]\,.
\end{split}
\end{equation}
where $\frac{1}{8\pi^2} d\a\, d\b \sin{\b}\,d\gam$ is the normalized Haar measure for this parametrization of $\SC\OC$(3). Considering an initial state that is an eigenstate of $J_z$, i.e., $\rho = \ket{m} \bra{m}$, its evolution can be conveniently expressed in terms of Wigner's small $d$-matrices
\begin{equation}
    U(\a.\b,\gam) \rho U(\a.\b,\gam)\ad =\sum_{m'm''} e^{-i\a(m'-m'')} d_{m'm}(\b)d_{mm''}(-\b) \ket{m'}\bra{m''}
\end{equation}
and using that $d^j_{rk}(\b)=d^j_{kr}(-\b)$ and the orthonormality relations
\begin{equation}
    \frac{1}{2\pi} \int_0^{2\pi} d \theta e^{-i\theta(m-m')}=\delta_{mm'}\,,\quad\quad
    \frac{2S+1}{2}\int_0^\pi d\b\, \sin{\b}\, d^S_{rk}(\b)d^{S'}_{rk}(\b)=\d_{S,S'}
\end{equation}
we arrive at

\begin{equation}
    \langle F \rangle = \frac{\Tr[M]}{2S+1}\,.
\end{equation}
\subsubsection{An unbiased gradient}

Now, the computation of the mean value of $\partial_{\mu} C(\thv)$ amounts to making the choice $M=X$, with $X=[H_\mu,U_A\ad O U_A]$. We readily find that
\begin{equation}
    \langle \partial_{\mu} C \rangle = \int_{\Ubb_A} dU_A W(U_A)   =0
\end{equation}
since 
\begin{equation}
    W(U_A)=\int_{\Ubb_B} dU_B \Tr[U_B\rho U_B\ad X(U_A)] = \int_{\SC\OC (3)} d\mu(U) \Tr[U\rho U\ad X] = \frac{\Tr[X(U_A)]}{2S+1} = 0\,,
\end{equation}
which follows from the fact that commutators are traceless.

\subsubsection{The variance}
Let us now  compute the variance of $\partial_{\mu} C(\thv)$. Having shown that $\langle \partial_{\mu} C \rangle=0$, the variance is
\begin{equation}
    \Var[\partial_{\mu} C] = \langle (\partial_{\mu} C)^2  \rangle =- \int_{\Ubb_A} dU_A W(U_A)\,,
\end{equation}
where
\begin{equation}
 W(U_A)=\int_{\Ubb_B} dU_B \Tr[U_B\ts \rho\ts (U_B\ts)\ad X\ts]=  \frac{1}{8\pi^2} \int_0^{2\pi}d\a\int_0^{\pi} d\b \sin{\b}\int_0^{2\pi} d\gam\, Y(\a,\b,\gam)\,,
\end{equation}
and 
\begin{equation}
\begin{split}
    Y(\a,\b,\gam)&=\Tr[ U\ts(\a,\b,\gam) \rho\ts (U\ts(\a,\b,\gam))\ad X]\\
    &=\sum_{m'm''n'n''}  e^{-i\a(m'-m'')}  e^{-i\a(n'-n'')}  d_{m'm}(\b)d_{m''m}(\b) \bra{m''}X\ket{m'} d_{n'n}(\b)d_{n''n}(\b) \bra{n''}X\ket{n'}
\end{split}
\end{equation}
Here and henceforth, $d_{m,m'}\equiv d^S_{m,m'}$. Now, again, integration over $\gam$ is trivial, and integrating over $\a$ we get $\delta_{n'',m'+n'-m''}$, so that
\begin{equation}
    W(U_A)= \sum_{m'm''n'}  \bra{m''}X\ket{m'} \bra{m'+n'-m''}X\ket{n'} 
    \frac{1}{2}\int_0^{\pi}d\b\, \sin\b\, d_{m'm}(\b)\, d_{m'' m}(\b)\,  d_{n'n}(\b)\, d_{m'+n'-m'', n}(\b)\,.
\label{Eq_A}
\end{equation}

In the following we will consider that the cost function minimize the expectation value of an operator $O$ that is an element of the dynamical algebra, i.e.,  $O\in \liea$. It is easy to verify that any unitary $U(\a,\b,\gam)$ of the form of Eq. \eqref{Eq_unitary_euler}, associated with some $3$-dimensional Euler rotation matrix $R(\a,\b,\gam)$,  transforms an arbitrary element $S_{\hat{n}}=(S_x,S_y,S_z)\cdot\hat{n} $ (with $\hat{n}$ a unit vector $\in \Rbb^3$)  of the algebra in the following manner
\begin{equation}
    U(\a,\b,\gam) S_{\hat{n}} U\ad(\a,\b,\gam) = S_{\hat{n}'},\quad \text{with  } \hat{n}'=R(\a,\b,\gam)\hat{n}
\end{equation}
Using the shorthand $\sin(x)\xrightarrow{}s_x$ and $\cos(x)\xrightarrow{}c_x$, the rotation matrix reads
\begin{equation}
R(\a,\b,\gam)=\left[\begin{array}{ccc}
c_{\a} c_{\b} c_{\gam}-s_{\a} s_{\gam} & -c_{\gam} s_{\a}-c_{\a} c_{\b} s_{\gam} & c_{\a} s_{\b} \\
c_{\a} s_{\gam}+c_{\b} c_{\gam} s_{\a} & c_{\a} c_{\gam}-c_{\b} s_{\a} s_{\gam} & s_{\a} s_{\b} \\
-c_{\gam} s_{\b} & s_{\b} s_{\gam} & c_{\b}
\end{array}\right]\,.
\end{equation}
Hence, considering $O=S_x+S_y+S_z$ and, calling $X^{(\iota)}=[J_x,U S_{\nu} U\ad]$, we get

\begin{equation}
\begin{split}
    &X^{(x)}=i s_\b c_\gam S_y +i(c_\a s_\gam +s_\a c_\b c_\gam) S_z \\
    &X^{(y)}=-i s_\b s_\gam S_y +i(c_\a c_\gam -s_\a c_\b s_\gam) S_z \\
    &X^{(z)}=-i c_\b S_y +is_\a s_\b S_z \,.
\end{split}
\end{equation}
Thus, denoting $X=\sum_{\nu=\{x,y,z\}} X^{(\nu)}= (a_x,a_y,a_z)\cdot (S_x,S_y,S_z)^T$ we find
\begin{equation}
    \begin{split}
        &a_x=0\\
        &a_y=-i(s_\b(s_\gam-c_\a)+c_\b)\\
        &a_z=i(s_\a(s_\b +c_\b(c_\gam-s_\gam)) +c_\a(s_\gam+c_\gam))\,.
    \end{split}
\end{equation}

For the moment, let us neglect the explicit dependence of $a_i$ on the Euler angles and let us compute the matrix elements of $X$
\begin{equation}
    \begin{split}
    \bra{m''}X\ket{m'}&=a_+\bra{m''}S_+\ket{m'} +a_-\bra{m''}S_-\ket{m'}+ a_z\bra{m''}S_z\ket{m'}  \\
    &=a_+\l_+(m')\d_{m'',m'+1}+a_-\l_-(m')\d_{m'',m'-1}+a_z m'\delta_{m'',m'}
    \end{split}
\end{equation}
where $a_{\pm}=\frac{a_x\mp ia_y}{2}$, and where we here define $S_{\pm}\ket{m}=\l_{\pm} (m)\ket{m\pm1}$. Similarly, 
\begin{equation}
    \begin{split}
    \bra{m'+n'-m''}X\ket{n'}&=a_+\l_+(n')\d_{m'+n'-m'',n'+1} +a_-\l_-(n')\d_{m'+n'-m'',n'-1} + a_z n'\delta_{m'+n'-m'',n'} \\
    &=a_+\l_+(n')\d_{m'',m'-1} +a_-\l_-(n')\d_{m'',m'+1} + a_z n'\delta_{m'',m'}  \,.
    \end{split}
\label{Eq_n''}
\end{equation}
Now, out of the nine terms in the product $\bra{m''}X\ket{m'}\bra{m'+n'-m''}X\ket{n'}$ appearing in Eq. \eqref{Eq_A}, only three terms are non-zero. That is, one only gets
\begin{equation}
    \begin{split}
        &T_1=a_{+}a_{-}\l_{+}(m')\l_{-}(n')\d_{m'',m'+1}\\
        &T_2=a_{+}a_{-}\l_{-}(m')\l_{+}(n')\d_{m'',m'-1}\\
        &T_3=a_{z}^2 m' n' \d_{m'',m'}\,.
    \end{split}
\end{equation}
which give rise to three terms on $Y$, namely $Y_1$, $Y_2$ and $Y_3$. Consider the first of these
\begin{equation}
    Y_1 = \sum_{m'n'}  a_{+}a_{-}\l_{+}(m')\l_{-}(n')
     \frac{1}{2}\int_0^{\pi}d\b\, \sin\b\, d_{m'm}(\b)\, d_{m'+1, m}(\b)\,  d_{n'n}(\b)\, d_{n'-1, n}(\b)\,,
\end{equation}
where, we recall that the coefficients $a_i=a_i(U_A)$ depend on $U_A$. Now, we can use the identity
\begin{equation}
    d^S_{rk}(\b)d^{S'}_{r'k'}(\b)=\sum_{J=|S-S'|}^{S+S'} \braket{S'rSr'}{J,r+r'}\braket{SkS'k'}{J,k+k'} d_{r+r',k+k'}^{J}(\b)\,.
\end{equation}
where we henceforth employ the more general notation $\ket{m}\rightarrow\ket{S,m}$. In our context, $S=S'$, and we will reference the Clebsh-Gordan coefficients as $g_{r,k}=\braket{SrSk}{J,r+k}$. We get

\begin{equation}
\begin{split}
    Y_1 =& \sum_{m'n'JJ'}  a_{+}a_{-}\l_{+}(m')\l_{-}(n')\frac{1}{2}\int_0^{\pi}d\b\, \sin\b\, d^J_{m'+n',m+n}(\b) d^{J'}_{m'+n',m+n}(\b) \\ &\times \braket{Sm'Sn'}{J,m'+n'}\, \braket{S,m'+1,S,n'-1}{J,m'+n'}\, \braket{SmSn}{J,m+n}\, \braket{SmSn}{J',m+n},
\end{split}
\end{equation}
and using the orthogonality relations yields
\small
\begin{align}
    Y_1 &= \sum_{m'n'J} \frac{ a_{+}a_{-}\l_{+}(m')\l_{-}(n')}{2J+1} \braket{Sm'Sn'}{J,m'+n'}\, \braket{S,m'+1,S,n'-1}{J,m'+n'}  (\braket{SmSn}{J,m+n})^2\nonumber\\
    &=\sum_{m'n'J} \frac{ a_{+}a_{-}\l_{+}(m')\l_{-}(n')}{2J+1} g_{m',n'} g_{m'+1,n'-1}g_{m,n}^2\,.
\end{align}
\normalsize
Similarly
\begin{equation}
\begin{split}
    &Y_2 = \sum_{m'n'J} \frac{a_{+}a_{-}\l_{-}(m')\l_{+}(n')}{2J+1} g_{m',n'} g_{m'-1,n'+1}g_{m,n}^2\,,  \\
    &Y_3 = \sum_{m'n'J} \frac{a_{z}^2 m' n'}{2J+1} g_{m',n'}^2g_{m,n}^2\,.
\end{split}
\end{equation}
These can be restated as $Y_1 = Y_2 = c_1\, a_{+}a_{-}$ and $Y_3 = c_3\, a_z^2$ with

\begin{equation}
    \begin{split}
    &c_1 = \sum_{m'n'J} \frac{\l_{-}(m')\l_{+}(n')}{2J+1} g_{m',n'} g_{m'-1,n'+1}g_{m,n}^2 = \frac{2}{3} m^2 \\
    &c_3 = \sum_{m'n'J} \frac{a_{z}^2 m' n'}{2J+1} g_{m',n'}^2g_{m,n}^2 = \frac{1}{3} m^2
    \end{split}
\end{equation}

We can proceed to integrate over $U_A$, assuming again we can replace $\Ubb_A$ with an integration over the Haar measure on  $\lieg$. Let us call $v_i$ each contribution to the variance, i.e. $\Var_{\thv}[\partial_{\mu} C(\thv)] = \sum_{i=1}^3 v_i$. First, one finds by explicit integration that let us consider 
\begin{equation}
\begin{split}
   v_3&= -c_3 \int_{U_A} dU_A a_z^2(U_A) \\
   &=-c_3 \frac{1}{8\pi^2} \int_0^{2\pi}d\aa\int_0^{\pi} d\bb \sin{\bb}\int_0^{2\pi} d\gg  a_{z}^2(\aa,\bb,\gg)\\
   &=c_3\,.
\end{split}
\end{equation}
 Similarly, it is easy to see that
\begin{equation}
    v_1= v_2= \frac{c_1}{4}\,.
\end{equation}
Altogether, we finally find
\begin{equation}\label{eq:resultsu2}
   \Var_{\thv}[\partial_{\mu} C(\thv)] = \Var_{A,B}[\partial_{\mu} C(\thv)]= \sum_{i=1}^3 v_i =\frac{2}{3} m^2
\end{equation}
where we have introduced $\Var_{A,B}$ to make explicit the assumption made, that is, that the before and after distributions of unitaries form $2$-designs.

\end{proof}

\subsubsection{Numerical simulations of the $\SC\UC$(2) toy model algebra.}
\begin{figure*}[t]
\centering
\includegraphics[width=.4\columnwidth]{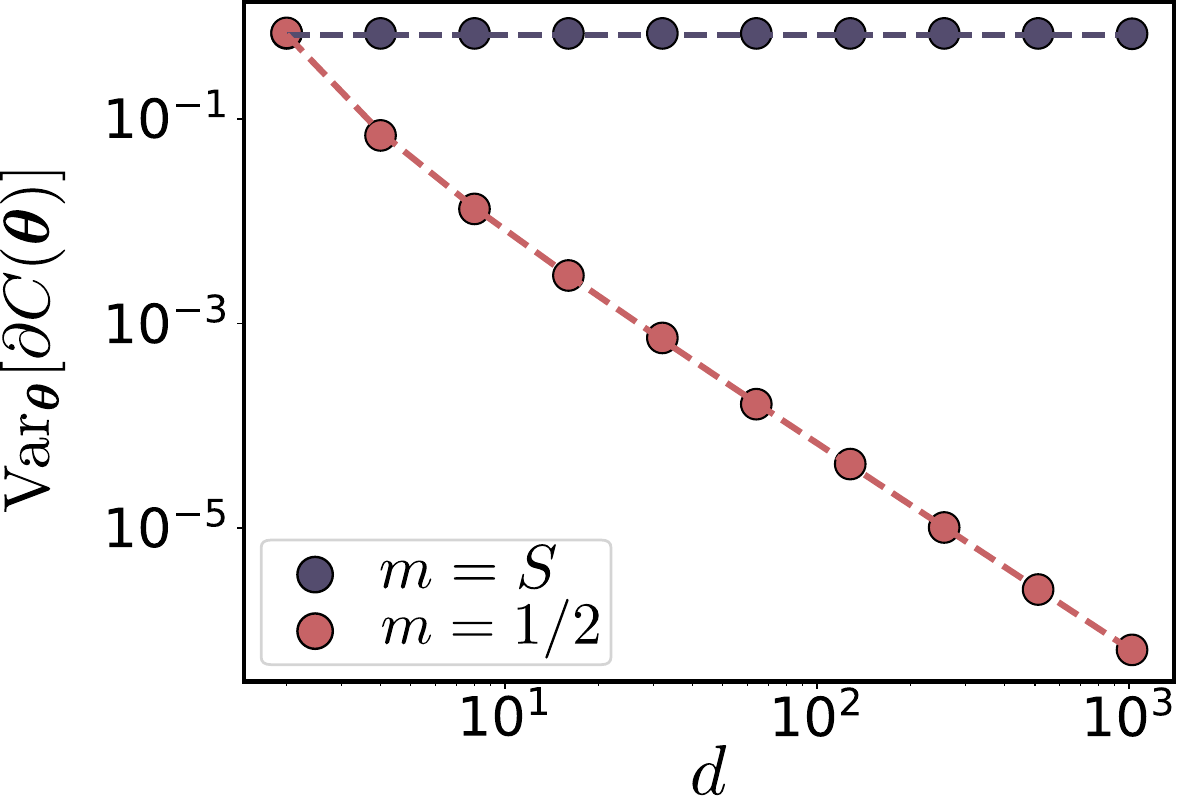}
\caption{\textbf{Numerical results for the $\SC\UC$(2) toy model algebra.}  Variance of the cost function partial derivative of the cost function in~\eqref{eq:costsu2SM} as a function of $d=\dim(\HC)$. Dashed lines indicate the theoretical prediction of~\eqref{eq:resultsu2}.}
\label{fig:su(2)}
\end{figure*}

Here we numerically simulate the model employing a PSA of the form
\begin{equation}
    U(\thv)=\prod_{l=1,\dots,L} e^{-i \theta_{l1} S_y} e^{-i \theta_{l2} S_x}\,,
\end{equation}
for minimizing the normalized cost function
\begin{equation}\label{eq:costsu2SM}
C(\thv) = \langle m|U\ad(\thv)\left( S_x  +  S_y  +  S_z\right)U(\thv)|m \rangle/S\,.
\end{equation}
 
We compute $\Var_{\thv}[\partial_\mu C(\thv)]$ for the initial states $\ket{m}=\ket{S}$ and $\ket{m}=\ket{1/2}$ for $L=100$ using $1200$ random initialization. As shown in Figure~\ref{fig:su(2)}, the  theoretical prediction of Eq.~\eqref{eq:resultsu2} matches the numerical results, indicating that an initial state $\ket{m}=\ket{1/2}$ is trainable, while an initial state $\ket{m}=\ket{S}$ leads to a barren plateau in the cost function.

\section{Symmetries in the $XXZ$ and $\TFIM$ systems}\label{App:XXZ}

\subsection{The $XXZ$ model}
Consider the generators for an HVA ansatz to the $XXZ$ model in Eq. \eqref{eq:genHVAXXZ} of section \ref{section:XXZ}.  The elements in $\GC_{XXZ_U}$, and thus in its associated DLA, share two common symmetries: they commute both with the magnetization operator $M=\sum_{i=1}^n Z_i$ and with the parity operator 

\begin{equation}\label{eq:parity}
    \Pi=\prod_{m=1}^{[n/2]} P_{m,n-m+1},\quad \text{where } P_{ij}=\frac{\id+\vec{\sigma}_i\vec{\sigma}_j}{2}\quad  \text{and } \vec{\sigma}_i = (X_i,Y_i,Z_i).
\end{equation}
 Here, $\Pi$ is the unitary $d$-dimensional irreducible representation of the element of the symmetric group $\SC_n$ that corresponds to a reflection over the central constituent. As a result of these two symmetries, the state space is broken into invariant subspaces with well defined parity and excitation number

\begin{equation}
\H=\bigoplus_{\substack{m=0\\ \sg=\pm}}^n \H_{m,\sg}
\end{equation}
The dimension of each excitation subspace is $\dim(\H_m)=\binom{n}{m} $, whereas, the joint parity-excitation subspaces have dimension $\dim(\H_{m,\sg})\approx \frac{1}{2} \dim(\H_m)$. The DLA is, accordingly, a direct sum of simple algebras on each invariant subspace
\begin{equation}
    \liea_{XXZ_U} = \bigoplus_{\substack{m=0\\ \sg=\pm}}^n \liea_{m,\sigma} \subseteq \mf{u}(d_{m,\sigma})\,.
\end{equation}
Notably, upon the addition of a generator consisting of local fields on either end of the chain, $\GC_{XXZ}=\GC_{XXZ_U}\cup \{\Z_1+\Z_N\}$, the DLA can be shown to become full rank on each subspace $ \liea_{\rm XXZ} = \bigoplus_{\substack{m=0\\ \sg=\pm}}^n \mf{u}(d_{m,\sigma})$ \cite{wang2016subspace}.

\subsection{The TFIM and LTFIM models}
In this section we review the symmetries of the different variants of the Transverse Field Ising Model (TFIM) presented in Section~\ref{section:Ising}.
\subsubsection{Open boundary condition}
Let us first consider open boundary conditions on the TFIM model. In this case, the generators $\GC_{\TFIM}= \left\{\sum_{i=1}^{n-1} Z_iZ_{i+1},\sum_{i=1}^{n} X_i\right\}$ have two symmetries. On one hand, they commute with the parity symmetry $\Pi$ defined in Eq. \eqref{eq:parity}. On the other hand, they commute with the so-called $\Zbb_2$ operator
\begin{equation}
    \Pi_{\Zbb_2} = \prod_{i=1}^n X_i
\end{equation}
that amounts to a global flip of the qubits. Consequently, $\HC$ is broken into four subspaces $\HC_{\sg,\sg'}$ with $\sg,\sg'\in \{1,-1\}$. Because the initial state $\ket{+}^{\otimes n}$ for the cost function in Eq. \eqref{eq:costTFIM} is an eigenstate of both symmetry operators with eigenvalues $\sg=\sg'=+1$, the dynamics under such a PQC is constrained to the $\HC_{+1,+1}$ subspace. We find, using Algorithm \ref{alg:lie}, that the dimension of the DLA scales polynomially, $\dim(\liea)=n^2$ (and so does the restriction to the $\sg=\sg'=+1$ subspace.

In turn, consider open boundary conditions on the \textit{Longitudinal} Transverse Field Ising Model (LTFIM), given by generators $\GC_{\rm LTFIM} = \GC_{\rm TFIM}\bigcup\left\{\sum_{i=1}^{n} Z_i\right\}$. While parity symmetry is conserved, the introduction of this new global longitudinal field breaks the $\Zbb_2$ symmetry. Thus, we are left with only two subspaces, $\HC=\bigoplus_{\sg=\pm 1}\HC_{\sg}$, of dimensions $\dim(\HC_{\sg})\approx \frac{2^n}{2}$. As expected, the DLA breaks into two corresponding subspace DLAs, each of which we find to be full rank on the corresponding subspaces, i.e., both subspaces are controllable.

\subsubsection{Closed boundary conditions}
Let us now consider closed boundary conditions. This case is slightly more involved since the parity symmetry is replaced by $C_n$, the cyclic group of $n$ elements. Hereafter, we follow Ref.~\cite{d2021dynamical}. Consider the operator $R$ whose action is to cycle the qubits in a state, i.e., $R\ket{a_1,\ldots,a_n}=\ket{a_n,a_1,\ldots,a_{n-1}}$. Clearly, $R^n=\id$. Now, as a consequence of the invariance of the generators under the action of this group of symmetries, the state space is broken into $n$ invariant subspaces $\HC=\bigoplus_{k=0}^{n-1} \HC_k$, where the projector onto each subspace is given by
\begin{equation}
    \P_k = \frac{1}{n}\sum_{j=0}^{n-1}\varepsilon^{kj}R^j\,.
\end{equation}
In particular, the state $\ket{+}^{\otimes n} \in \HC_0$ and for that reason we will only focus on this subspace. Let us note that a general formula for the dimension of this subspace is, as far as we know, not known, but it is possible to derive a closed expression in the case of  $n$ prime~\cite{d2021dynamical}
\begin{equation}
    \dim(\HC_0)= 2+ \frac{(2^n-2)}{n}\,,
\end{equation}
which makes evident that the subspace is exponentially large. Furthermore, we compute the DLA reduced to this subspace in the LTFIM case and see an exponential behaviour. In contrast, for the TFIM model we find a DLA that reduced to the $k=0$ subspace and the $\sg'=+1$ (also has $\Zbb_2$ symmetry), and that  has a linear scaling, i.e., $\dim(\liea_{0,+})\in\OC(n)$.

\section{Initial state for the numerical simulations of the $XXZ$ spin chain model}\label{App:numerical}
The initial states $\ket{ \psi_{m,+}}$ in the cost function of Eq.~\eqref{eq:costxxz} are chosen to be eigenstates of the symmetries of the generators, namely $M=\sum_{i=1}^n Z_i$ and $\Pi$ (the reflective spatial symmetry w.r.t. the chain's center) defined in Eq.~\eqref{eq:parity}. That is, the initial state is
\begin{align}
\ket{ \psi_{m,+}} = \frac{1}{\sqrt{2}} \left(\bigotimes_{i=1}^m |0\rangle_i \bigotimes_{i=m+1}^n |1\rangle_i +\bigotimes_{i=1}^{n-m} |1\rangle_i \bigotimes_{i=n-m+1}^n |0\rangle_i \right).
\end{align}
Here $m$ denotes the number of excitations in the state, and $+$ the fact that it is an even parity eigenstate.

\end{document}